\PassOptionsToPackage{unicode}{hyperref}
\PassOptionsToPackage{hyphens}{url}
\PassOptionsToPackage{dvipsnames,svgnames,x11names}{xcolor}
\documentclass[
  12pt]{article}

\usepackage{amsmath,amssymb,amsthm,needspace}
\usepackage{graphicx,url,latexsym,multirow,enumitem,booktabs,
bbm,soul,epsfig,cancel,makecell,lscape,algorithmicx,bm,array,
longtable,caption,algorithm,algpseudocode,setspace}

\usepackage{iftex}
\ifPDFTeX
  \usepackage[T1]{fontenc}
  \usepackage[utf8]{inputenc}
  \usepackage{textcomp} %
\else %
  \usepackage{unicode-math}
  \defaultfontfeatures{Scale=MatchLowercase}
  \defaultfontfeatures[\rmfamily]{Ligatures=TeX,Scale=1}
\fi
\usepackage{lmodern}
\ifPDFTeX\else
\fi
\IfFileExists{upquote.sty}{\usepackage{upquote}}{}
\IfFileExists{microtype.sty}{%
  \usepackage[]{microtype}
  \UseMicrotypeSet[protrusion]{basicmath} %
}{}
\makeatletter
\@ifundefined{KOMAClassName}{%
  \IfFileExists{parskip.sty}{%
    \usepackage{parskip}
  }{%
    \setlength{\parindent}{0pt}
    \setlength{\parskip}{6pt plus 2pt minus 1pt}}
}{%
  \KOMAoptions{parskip=half}}
\makeatother
\usepackage{xcolor}
\makeatletter
\ifx\paragraph\undefined\else
  \let\oldparagraph\paragraph
  \renewcommand{\paragraph}{
    \@ifstar
      \xxxParagraphStar
      \xxxParagraphNoStar
  }
  \newcommand{\xxxParagraphStar}[1]{\oldparagraph*{#1}\mbox{}}
  \newcommand{\xxxParagraphNoStar}[1]{\oldparagraph{#1}\mbox{}}
\fi
\ifx\subparagraph\undefined\else
  \let\oldsubparagraph\subparagraph
  \renewcommand{\subparagraph}{
    \@ifstar
      \xxxSubParagraphStar
      \xxxSubParagraphNoStar
  }
  \newcommand{\xxxSubParagraphStar}[1]{\oldsubparagraph*{#1}\mbox{}}
  \newcommand{\xxxSubParagraphNoStar}[1]{\oldsubparagraph{#1}\mbox{}}
\fi
\makeatother

\usepackage{longtable,booktabs,array}
\usepackage{calc} %
\usepackage{etoolbox}
\makeatletter
\patchcmd\longtable{\par}{\if@noskipsec\mbox{}\fi\par}{}{}
\makeatother
\IfFileExists{footnotehyper.sty}{\usepackage{footnotehyper}}{\usepackage{footnote}}
\makesavenoteenv{longtable}
\usepackage{graphicx}
\makeatletter
\def\maxwidth{\ifdim\Gin@nat@width>\linewidth\linewidth\else\Gin@nat@width\fi}
\def\maxheight{\ifdim\Gin@nat@height>\textheight\textheight\else\Gin@nat@height\fi}
\makeatother
\setkeys{Gin}{width=\maxwidth,height=\maxheight,keepaspectratio}
\makeatletter
\def\fps@figure{htbp}
\makeatother

\makeatletter
\@ifpackageloaded{caption}{}{\usepackage{caption}}
\AtBeginDocument{%
\ifdefined\contentsname
  \renewcommand*\contentsname{Table of contents}
\else
  \newcommand\contentsname{Table of contents}
\fi
\ifdefined\listfigurename
  \renewcommand*\listfigurename{List of Figures}
\else
  \newcommand\listfigurename{List of Figures}
\fi
\ifdefined\listtablename
  \renewcommand*\listtablename{List of Tables}
\else
  \newcommand\listtablename{List of Tables}
\fi
\ifdefined\figurename
  \renewcommand*\figurename{Figure}
\else
  \newcommand\figurename{Figure}
\fi
\ifdefined\tablename
  \renewcommand*\tablename{Table}
\else
  \newcommand\tablename{Table}
\fi
}
\@ifpackageloaded{float}{}{\usepackage{float}}
\floatstyle{ruled}
\@ifundefined{c@chapter}{\newfloat{codelisting}{h}{lop}}{\newfloat{codelisting}{h}{lop}[chapter]}
\floatname{codelisting}{Listing}

\makeatother
\makeatletter
\@ifpackageloaded{caption}{}{\usepackage{caption}}
\@ifpackageloaded{subcaption}{}{\usepackage{subcaption}}
\makeatother

\ifLuaTeX
  \usepackage{selnolig}  %
\fi
\usepackage[]{natbib}
\usepackage{bookmark}

\IfFileExists{xurl.sty}{\usepackage{xurl}}{} %
\hypersetup{
  pdftitle={Fitting time-varying parameter models via partially noncentered parametrizations},
  pdfauthor={Lingjie Wang; Rufeng Liu; Abhinek Shukla; Linda S. L. Tan},
  pdfkeywords={Time-varying parameter models, Partially noncentered
parametrization, Variational inference, Mean-field variational Bayes, Triple Gamma prior, State-space models},
  colorlinks=true,
  linkcolor={blue},
  filecolor={Maroon},
  citecolor={Blue},
  urlcolor={Blue},
  pdfcreator={LaTeX via pandoc}}

\newcommand{\anon}{1}

\def\half{\frac{1}{2}}

\def\tr{\text{\rm tr}}
\def\vec{\text{\rm vec}}

\def\Cov{\text{\rm Cov}}
\def\Var{\text{\rm Var}}

\def\opt{\text{opt}}
\def\KL{\text{\rm KL}}

\def\CP{\text{CP}}
\def\NCP{\text{NCP}}
\def\G{\text{G}}

\def\IG{\text{IG}}
\def\GIG{\text{GIG}}

\def\argmax{\text{\rm argmax}}
\def\diag{\text{\rm diag}}
\def\blockdiag{\text{\rm blockdiag}}

\def\vech{\text{\rm vech}}

\def\bfa{{\bf a}}
\def\bfb{{\bf b}}
\def\bfc{{\bf c}}

\def\bfm{{\bf m}}

\def\bfu{{\bf u}}

\def\bfx{{\bf x}}
\def\bfy{{\bf y}}
\def\bfz{{\bf z}}

\def\bfbeta{{\pmb \beta}}

\def\wbeta{{\widetilde{\beta}}}
\def\wbfbeta{{\widetilde{\pmb \beta}}}

\def\bfmu{{\pmb \mu}}
\def\bfkappa{{\pmb \kappa}}
\def\bfgamma{{\pmb \gamma}}
\def\bftheta{{\pmb \theta}}

\def\bfpsi{{\pmb \psi}}
\def\bftau{{\pmb \tau}}
\def\bfxi{{\pmb \xi}}
\def\bflambda{{\pmb \lambda}}
\def\bfvartheta{{\pmb \vartheta}}

\def\obs{{\text{obs}}}
\def\aug{{\text{aug}}}

\def\N{{\text{N}}}
\def\IG{{\text{IG}}}
\def\E{{\text{E}}}

\def\bfone{{\bf 1}}
\def\bfzr{{\bf 0}}

\newcommand\mL{{\mathcal{L}}}

\newcommand{\bracket}[1]{\left( #1 \right)}

\newtheorem{theorem}{Theorem}
\newtheorem{corollary}{Corollary}

\newcolumntype{P}[1]{>{\centering\arraybackslash}p{#1}}

\allowdisplaybreaks
\begin{document}

\def\spacingset#1{\renewcommand{\baselinestretch}%
  {#1}\small\normalsize} \spacingset{1}

\allowdisplaybreaks

\if1\anon
  {
    \title{\bf Fitting time-varying parameter models via partially noncentered parametrizations}
    \author{
    Lingjie Wang (lingjie@u.nus.edu)\\
    Department of Statistics and Data Science, National University of Singapore\\
    Rufeng Liu (rufeng-liu@nus.edu.sg)\\
    Department of Statistics and Data Science, National University of Singapore\\
    Abhinek Shukla (abhushukla@gmail.com)\\
    Duke-NUS Medical School, Singapore\\
    Linda S. L. Tan (statsll@nus.edu.sg)\thanks{Linda Tan's research is supported by the Ministry of Education, Singapore, under its Academic Research Fund Tier 2 (Award MOE-T2EP20222-0002). Corresponding author.}\hspace{.2cm}\\
    Department of Statistics and Data Science, National University of Singapore
    }
    \date{October 23, 2026}
    \maketitle
  } \fi

\if0\anon
  {
    \bigskip
    \bigskip
    \bigskip
    \begin{center}
      {\LARGE\bf Fitting time-varying parameter models via partially noncentered parametrizations}
    \end{center}
    \medskip
  } \fi

\bigskip
\begin{abstract}
Time-varying parameter (TVP) models are widely used to capture evolving relationships but inference is often computationally demanding due to high-dimensional latent states. The efficiency of iterative algorithms depends critically on the choice of latent state augmentation, with centering preferred for genuinely time-varying coefficients and noncentering favored when coefficients are nearly constant. We propose a partially noncentered parametrization that bridges these extremes via dimension-specific working parameters for both location and scale. Optimal working parameters are derived by analyzing EM convergence rates, and an alternating expectation-conditional maximization algorithm is developed to implement the approach. As the convergence rate of mean-field variational Bayes (MFVB) is closely related to EM, we extend the approach to develop partially noncentered MFVB algorithms for the TVP model, placing triple gamma priors on baseline coefficients and process variances to induce shrinkage. Simulations and empirical applications show that the resulting MFVB algorithms yield posterior approximations closer to Markov chain Monte Carlo than conventional centered or noncentered parametrizations. We also develop an efficient solver for the block tridiagonal inverse matrices arising in the proposed algorithms.
\end{abstract}

\noindent%
{\it Keywords:} Time-varying parameter models, Partially noncentered
parametrization, Mean-field variational Bayes, Triple Gamma prior, State-space models
\vfill

\newpage
\spacingset{1.8} %

\section{Introduction}\label{sec:Introduction}
Time-varying parameter (TVP) models capture structural shifts in predictor–response relationships by allowing coefficients to evolve dynamically, making them indispensable in macroeconomics, finance and forecasting \citep{liu2024economic, coulombe2025time}. However, fitting TVP state space models is computationally challenging due to joint inference over the high-dimensional latent states and model parameters. For iterative algorithms, such as expectation-maximization (EM), Markov chain Monte Carlo (MCMC) and mean-field variational Bayes \citep[MFVB,][]{Blei2017}, efficiency hinges on the data augmentation scheme used to construct the latent states. Hence, data augmentation strategies have been widely developed across frequentist and Bayesian settings to accelerate convergence.

In the EM context, \citet{Meng1997} introduce working parameters to minimize the fraction of missing information \citep{Liu1999} via conditional augmentation, while \cite{Meng1999} assign priors to working parameters and integrate them out under marginal augmentation. More recently, \cite{Goplerud2021} adapts marginal augmentation to MFVB, while \citet{Tan2023} incorporates working parameters into state space models to construct efficient MCMC algorithms using block-specific reparametrization.

For hierarchical models, two standard augmentation schemes are the centered parametrization (CP) and noncentered parametrization (NCP) \citep{Gelfand1995}, whose relative performance depends on data informativeness and posterior geometry \citep{Zanella2021}. As CP and NCP exhibit complementary convergence properties, the ancillarity-sufficiency interweaving strategy \citep[ASIS,][]{Yu2011} interweaves them within MCMC iterations to reduce autocorrelation. While ASIS has been successfully applied to TVP models \citep{Bitto2019, He2024, Iacopini2026}, its implementation remains challenging in complex architectures, such as time-varying multi-seasonal autoregressive models \citep{Fagerberg2026}, where sampling under at least one parametrization is non-trivial. Alternatively, the partially noncentered parametrization \citep[PNCP,][]{Papaspiliopoulos2007} interpolates between CP and NCP via a data driven scheme that adapts the posterior correlation structure to improve efficiency. PNCP effectively accelerates convergence in both MCMC for spatial models \citep{Christensen2006, Bass2019}, and MFVB for generalized linear mixed models \citep{Tan2013, Tan2021}.

Building on these insights, we propose a dimension-specific PNCP for TVP models, which recenters and rescales latent states using working parameters $(W, \bfa)$. By analyzing the EM convergence rate via the augmented information matrix, we analytically derive optimal working parameters to construct an alternating expectation-conditional maximization \citep[AECM,][]{Meng1997} algorithm. This framework dynamically updates data augmentation schemes across parameter blocks and iterations, ensuring tractable updates and monotonic likelihood growth, while consistently outperforming CP and NCP in both convergence speed and maximized likelihood.

As the EM algorithm, Gibbs sampler and MFVB share identical convergence rates under Gaussian targets \citep{Sahu1999, Tan2014}, these optimal working parameters extend directly to MCMC and variational Bayes (VB). We propose two schemes to implement PNCP within MFVB: (1) D-CP/NCP, which preserves analytic updates by selecting between CP and NCP per dimension, and (2) CLBS (continuous location binary scale), which allows arbitrary $W$ and thresholds elements of $\bfa$ in $\{0,1\}$. In simulations and macroeconomic applications, D-CP/NCP consistently accelerates convergence and improves posterior approximation accuracy over CP and NCP, while CLBS often yields further gains through added flexibility, all while maintaining MFVB's computational edge over MCMC. To ensure high-dimensional scalability, we build on \cite{Meurant1992} to derive an efficient algorithm for computing the inverses of block tridiagonal matrices and solving their associated linear systems, as required by the AECM and MFVB updates.

As TVP models grow in dimension, shrinkage priors are essential to prevent overfitting \citep{huber2021jbes} and to classify coefficients as time-varying, static, significant or irrelevant \citep{Fruhwirth2021}. We employ triple gamma shrinkage priors \citep{Cadonna2020} on state variances and coefficients, which provide flexible control over zero-concentration and tail behavior, while unifying the horseshoe \citep{Carvalho2009}, double gamma \citep{Bitto2019} and Bayesian lasso \citep{Belmonte2014} as special cases. Dynamic variants \citep{Kowal2019, hauzenberger2024dynamic, Knaus2025} further allow shrinkage to evolve over time.

As MCMC scales poorly in high-dimensional TVP models, VB offers a computationally efficient alternative, well-suited for economic forecasting applications that require frequent model refitting. In particular, MFVB enables fast analytic updates via factorized posterior approximations. While \cite{Koop2023} and \cite{Bianco2026} implement MFVB using spike-and-slab dynamic variable selection under a CP, we consider a static triple gamma shrinkage prior under a dimension-specific PNCP. Related developments include hybrid VB \citep{Loaizamaya2022} and importance variational approximations \citep{Loaiza2025}, while \cite{frazier2023} evaluate the theoretical accuracy of VB for state space models. Alternatively, \cite{Rockova2021} employ a one-step-late EM algorithm for maximum a posteriori smoothing in dynamic variable selection.

The rest of the article is organized as follows. Section \ref{sec:TVP_model} introduces the PNCP for the TVP model. Section \ref{sec:EM} derives the optimal working parameters and details the AECM algorithm, followed by simulation and empirical studies in Section \ref{sec:EM_Applications}. Section \ref{sec:Meurant} outlines efficient computational methods for inverting high-dimensional block tridiagonal matrices. Section \ref{sec:Prior_and_density} introduces the  triple gamma shrinkage prior, and Section \ref{sec:MFVB_TVP} develops the MFVB algorithm alongside the proposed D-CP/NCP and CLBS schemes. Section \ref{sec:VA_Applications} presents empirical findings for the MFVB implementation, and Section \ref{sec:Conclusion} concludes.

\section{Partial noncentering in TVP models}\label{sec:TVP_model}
In the univariate TVP model with observations $y_t \in \mathbb{R}$, we have
\begin{align*}
y_t &= \bfx_t^\top \bfbeta_t + \epsilon_t, \quad
\epsilon_t \sim \N(0, \sigma^2), \qquad t=1, \dots, T,
\end{align*}
where the regressors $\bfx_t = (x_{t1}, \dots, x_{td})^\top\in \mathbb{R}^d$ and the latent state $\bfbeta_t  = (\beta_{t1}, \dots, \beta_{td})^\top \sim \N( \bfbeta_{t-1}, D_\bftheta)$ evolves as a random walk with initial state $\bfbeta _0  \sim \N(\bfgamma, D_\bftheta)$. Here $\bfgamma =(\gamma_1, \dots, \gamma_d)^\top$, $\bftheta = (\theta_1, \dots, \theta_d)^\top$ and $D_\bftheta = \diag(\bftheta)$, where $\theta_j > 0$ for $j=1, \dots, d$. This specification of the latent states ${\bfbeta} = (\bfbeta_0^\top, \bfbeta_1^\top, \dots, \bfbeta_T^\top)^\top$ is the CP, since each $\bfbeta_t$ is centered about $\bfgamma$ with $\E(\bfbeta_t) = \bfgamma$ and $\Cov(\bfbeta_t) = (t+1) D_\bftheta$. \citet{Bitto2019} introduce a NCP, which can improve posterior geometry and MCMC sampling efficiency when the CP mixes poorly. It defines standardized latent state $\bar{\bfbeta}_t = (\bar{\beta}_{t1}, \dots, \bar{\beta}_{td})^\top$ via
\[
\bar{\bfbeta}_t = D_\bftheta^{-1/2} (\bfbeta_t - \bfgamma),
\quad t=0, \dots, T.
\]
Under this transformation, $\bar{\bfbeta}_t \sim \N( \bar{\bfbeta}_{t-1}, I_d)$ for $t=1, \dots, T$ with initial state $\bar{\bfbeta}_0 \sim \N(\bfzr_d,I_d )$. Thus, the latent states $\bar{\bfbeta} = (\bar{\bfbeta}_0^\top, \bar{\bfbeta}_1^\top, \dots, \bar{\bfbeta}_T^\top )^\top$ are a priori independent of both $\bfgamma$ and $\bftheta$.

To interpolate between CP and NCP, we introduce a {\em dimension-specific} PNCP. Define partially noncentered latent state $\wbfbeta_t$ by
\begin{align*}
\wbfbeta_t = D_\bftheta^{-\bfa/2} (\bfbeta_t - W_t \bfgamma),
\quad t=0, \dots, T,
\end{align*}
where $W_t \in \mathbb{R}^{d \times d}$ and $\bfa = (a_1, \dots, a_d)^\top \in \mathbb{R}^d$ are working parameters controlling the degree of centering in location and scale respectively, and $D_\bftheta^{\bfa/2} = \diag(\theta_1^{a_1/2}, \dots, \theta_d^{a_d/2})$. Let $\overline{W}_t = I_d - W_t$, and stack $W= (W_0^\top, W_1^\top, \dots, W_T^\top)^\top$, $\overline{W} = (\overline{W}_0^\top, \overline{W}_1^\top, \dots,
\overline{W}_T^\top)^\top$. The PNCP nests both standard parametrizations: it reduces to the CP when ($\bfa=\bfzr$, $W = 0$), and NCP when ($\bfa=\bfone$, $\overline{W} = 0$). Let $\bfy = (y_1, \dots, y_T)^\top$ and $X = \blockdiag(\bfx_1^\top, \dots, \bfx_T^\top)$. Writing $W_{1:T}=(W_1^\top, \dots, W_T^\top)^\top$ and $\bfbeta_{1:T} = (\bfbeta_1^\top, \dots, \bfbeta_T^\top)^\top$, the PNCP can be expressed as
\begin{align*}
\bfy \mid \wbfbeta, \bfgamma, \bftheta, \sigma^2
&\sim \N \left( X W_{1:T}\bfgamma + X(I_T \otimes D_\bftheta^{\bfa/2}) \wbfbeta_{1:T}, \; \sigma^2 I_T \right), \\
\wbfbeta \mid \bfgamma, \bftheta &\sim \N \left( (I_{T+1}\otimes D_\bftheta^{-\bfa/2}) \overline{W} \bfgamma, \;
\Lambda^{-1} \otimes D_\bftheta^{\bfone-\bfa} \right),
\end{align*}
where $\Lambda \in \mathbb{R}^{(T+1)\times (T+1)}$ is a tridiagonal matrix with diagonal $(2, \dots, 2, 1)^\top$ and $-1$ on the off-diagonals. The prior precision matrix for $\wbfbeta$ is block tridiagonal, reflecting the Markov property that each $\wbfbeta_t$ is conditionally independent of all other states given its neighbors.

The PNCP is motivated by \citet{Bitto2019}, who find that in MCMC estimation, CP is more efficient for time-varying coefficients, while NCP is preferred for nearly constant coefficients. A single global parametrization is thus suboptimal when time variation differs across dimensions, and we address this by allowing the working parameters $(\bfa, W)$ to vary by dimension. Incorporating a time-varying $W_t$ further reduces posterior dependence between $\bfgamma$ and $\wbfbeta$. Unlike \citet{Bitto2019}, who interweave CP and NCP via ASIS \citep{Yu2011}, we target a computationally efficient parametrization per dimension.  To establish a principled basis for selecting the working parameters, we exploit the convergence rate of the EM algorithm, which is theoretically governed by the augmented information matrix \citep{Dempster1977}.

\section{AECM algorithm for TVP model}\label{sec:EM}
In this section, we develop an AECM algorithm for the TVP model that adaptively selects an optimal parametrization per coefficient dimension. First, we derive closed-form updates for the model parameters, restricting $\bfa \in \{0,1\}^d$ to ensure tractability. Second, we determine optimal working parameters by addressing location and scale parameters separately. Finally, we present an AECM algorithm \citep{Meng1997}, which allows the data augmentation scheme to vary across parameter blocks and iterations.

Let $\bfpsi = (\bfgamma^\top, \bftheta^\top, \sigma^2)^\top$ denote the model parameters. Our goal is to find $\bfpsi^*$ that maximizes the observed data likelihood $p(\bfy \mid \bfpsi) = \int p(\bfy, \wbfbeta \mid \bfpsi) \, d\wbfbeta$, where $\wbfbeta$ represents the latent (missing) data and $(\bfy, \wbfbeta)$ constitutes the augmented (complete) data. Given an initial estimate $\bfpsi^{(0)}$, each iteration $i$ of the EM algorithm alternates between an E-step, computing $Q(\bfpsi \mid \bfpsi^{(i)}) = \E_{\wbfbeta \mid \bfy, \bfpsi^{(i)}} [\log p(\bfy, \wbfbeta \mid \bfpsi)]$, and an M-step, updating $\bfpsi^{(i+1)} = \argmax_\bfpsi Q(\bfpsi \mid \bfpsi^{(i)})$.

Under the PNCP, the conditional posterior of the latent states is $\wbfbeta \mid \bfy, \bfpsi \sim N(\bfm_{\wbfbeta, \bfa W}, V_{\wbfbeta, \bfa})$, where the subscripts emphasize dependence on the working parameters $(\bfa, W)$. For instance, $\bfm_{\wbfbeta, \bfzr 0}$ and $\bfm_{\wbfbeta, \bfzr I}$ denote the posterior mean evaluated at $(\bfa=\bfzr, W = 0)$ and $(\bfa=\bfzr, \overline{W} = 0)$ respectively. We have
\begin{gather}
V_{\wbfbeta, \bfa} = (I_{T+1} \otimes D_\bftheta^{-\bfa/2}) V_{\wbfbeta, \bfzr} (I_{T+1} \otimes D_\bftheta^{-\bfa/2}), \quad
\bfm_{\wbfbeta, \bfa W} = (I_{T+1} \otimes D_\bftheta^{-\bfa/2})\bfm_{\wbfbeta, \bfzr W}, \nonumber \\
V_{\wbfbeta, \bfzr} = (\Lambda \otimes D_\bftheta^{-1} + Z)^{-1}, \quad
\bfm_{\wbfbeta, \bfzr W} = V_{\wbfbeta, \bfzr} \{ (\Lambda  \otimes D_\bftheta^{-1}) \overline{W} \bfgamma + \bfz \},  \label{eq: m & V}
\end{gather}
where $Z = \blockdiag(0_{d\times d}, X^\top X/\sigma^2)$ and $\bfz = (\bfzr_d^\top, \{ X^\top (\bfy -  X W_{1:T} \bfgamma)/\sigma^2 \}^\top)^\top$. As $\Lambda \otimes D_\bftheta^{-1} + Z$ is block tridiagonal, the diagonal and superdiagonal blocks of $V_{\wbfbeta, \bfzr}$ along with $\bfm_{\wbfbeta, \bfzr W}$ can be computed efficiently using Algorithm \ref{Alg: Meurant} presented in Section~\ref{sec:Meurant}. In addition,
\begin{align*}
Q(\bfpsi &\mid \bfpsi^{(i)})
= - \frac{1}{2\sigma^2} \sum_{t=1}^T \{(y_t - \bfx_t^\top W_t \bfgamma - \bfx_t^\top D_\bftheta^{\bfa/2} \bfm_{\wbfbeta_t, \bfa W}^{(i)})^2 + \bfx_t^\top D_\bftheta^{\bfa/2} V_{\wbfbeta_t, \bfa}^{(i)} D_\bftheta^{\bfa/2} \bfx_t \} -\frac{T}{2} \log \sigma^2\\
&\quad -\frac{T + (T+1)d}{2} \log(2\pi) - \frac{(T+1)}{2} \sum_{j=1}^d (1-a_j)\log \theta_j -\frac{1}{2} \tr\{(\Lambda \otimes D_\bftheta^{\bfa-\bfone})V_{\wbfbeta, \bfa}^{(i)} \}\\
&\quad  - \frac{1}{2}\{ \bfm_{\wbfbeta, \bfa W}^{(i)} - (I_{T+1}\otimes D_\bftheta^{-\bfa/2}) \overline{W} \bfgamma \}^\top (\Lambda \otimes D_\bftheta^{\bfa-\bfone}) \{ \bfm_{\wbfbeta, \bfa W}^{(i)} - (I_{T+1}\otimes D_\bftheta^{-\bfa/2}) \overline{W} \bfgamma \}.
\end{align*}

Solving $\nabla_{\bfpsi} Q(\bfpsi \mid \bfpsi^{(i)}) = 0$, we obtain closed-form updates of $\sigma^2$ and $\bfgamma$ for any $(\bfa, W)$:
\begin{align}
\sigma^2 &\leftarrow \frac{1}{T} \sum_{t=1}^T \{(y_t - \bfx_t^\top W_t \bfgamma - \bfx_t^\top D_\bftheta^{\bfa/2} \bfm_{\wbfbeta_t, \bfa W}^{(i)})^2 + \bfx_t^\top D_\bftheta^{\bfa/2} V_{\wbfbeta_t, \bfa}^{(i)} D_\bftheta^{\bfa/2} \bfx_t \},
\label{eq: sig2 update}\\
\bfgamma &\leftarrow \left\{ W^\top Z W + \overline{W}^\top (\Lambda \otimes D_\bftheta^{-1}) \overline{W} \right\}^{-1} \left[ W_{1:T}^\top X^\top \{\bfy  - X (I_T \otimes D_\bftheta^{\bfa/2}) \bfm_{\wbfbeta_{1:T}, \bfa W}^{(i)}\} / \sigma^2 \right. \nonumber \\
& \left. \quad  + \overline{W}^\top (\Lambda \otimes D_\bftheta^{\bfa/2-1}) \bfm_{\wbfbeta, \bfa W}^{(i)}\right]. \label{eq: gamma update}
\end{align}

To derive closed-form updates for $\bftheta$, we partition the coordinate set $\{1, \dots, d\}$ into two disjoint subsets, $S_\NCP$ and $S_\CP$. Let $W_{[j]}$ and $\overline{W}_{[j]}$ denote the $j$th rows of $W_t$ and $\overline{W}_t$ stacked across $t=0, \dots, T$, respectively. For $j \in S_\NCP$, we set $(a_j=1,\overline{W}_{[j]} = 0)$. For $j \in S_\CP$, we set $a_j=0$, leaving $W_{[j]}$ unconstrained. For $j \in S_\CP$,  the update for $\theta_j$ is
\begin{align} \label{eq: theta CP}
\theta_j &\leftarrow \frac{1}{T+1} \tr\left[ \Lambda \left\{ V_{\wbfbeta_{[j]}, \bfa}^{(i)} + (\bfm_{\wbfbeta_{[j]}, \bfa W}^{(i)} - \overline{W}_{[j]} \bfgamma)(\bfm_{\wbfbeta_{[j]}, \bfa W}^{(i)} - \overline{W}_{[j]} \bfgamma)^\top \right\} \right],
\end{align}
where $\bfm_{\wbfbeta_{[j]}, \bfa W} \in \mathbbm{R}^{T+1}$ collects the $j$th entry of $\bfm_{\wbfbeta_t, \bfa W}$ across $t=0, \dots, T$, and $V_{\wbfbeta_{[j]}, \bfa} \in \mathbbm{R}^{(T+1)\times (T+1)}$ has its $(t, t')$ element equal to the $(j, j)$ element of the $(t,t')$ block of $V_{\wbfbeta, \bfa}$. Let $\sqrt{\bftheta_\NCP} = (\sqrt{\theta_j})_{j \in S_\NCP}$, and define $\bfgamma_\NCP$, $\bfx_{t, \CP}$, $\bfx_{t, \NCP}$, $\wbfbeta_{t, \NCP}$ and $\bfm_{\wbfbeta_{t, \NCP}, \bfa W}^{(i)}$ analogously. Further, let $D_\bfu = \diag(\bfu)$ for any vector $\bfu$. For $j \in S_\NCP$,
\begin{align}  \label{eq: theta NCP}
\sqrt{\bftheta_\NCP} &\leftarrow \left( \sum_{t=1}^{T} D_{\bfx_{t, \NCP}}( \Var(\wbfbeta_{t, \NCP}) + \bfm_{\wbfbeta_{t, \NCP}, \bfa W}^{(i)} \bfm_{\wbfbeta_{t, \NCP}, \bfa W}^{(i)\top}) D_{\bfx_{t, \NCP}}  \right) ^{-1} \sum_{t=1}^{T} D_{\bfx_{t, \NCP}}\nonumber
\\
& \times \{ \bfm_{\wbfbeta_{t, \NCP}, \bfa W}^{(i)}  (y_t - \bfx_t^\top W_t \bfgamma -  \bfx_{t, \CP}^\top \bfm_{\wbfbeta_{t, \CP}, \bfa W}^{(i)} ) -\Cov(\wbfbeta_{t, \NCP}, \widetilde{\bfbeta}_{t, \CP}) \bfx_{t, \CP} \},
\end{align}
where $\Var(\wbfbeta_{t, \NCP})$ and $\Cov(\wbfbeta_{t, \NCP}, \widetilde{\bfbeta}_{t, \CP})$ are extracted from corresponding blocks in $V_{\wbfbeta, \bfa}^{(i)}$.

From \cite{Meng1994}, the matrix rate of convergence of the EM algorithm is $J(\bfpsi^*) = I - I_\aug^{-1}(\bfpsi^*) I_\obs(\bfpsi^*)$, where $I_\obs(\bfpsi) = - \nabla_\bfpsi^2 \log p(\bfy \mid \bfpsi)$ is the observed information and $I_\aug(\bfpsi) = - \E_{p(\wbfbeta \mid \bfy, \bfpsi)} [\nabla_\bfpsi^2 \log p(\bfy, \wbfbeta \mid \bfpsi)]$ is the augmented data information. The global rate of convergence is defined as the spectral radius (largest eigenvalue) of $J(\bfpsi^*)$, where larger values indicate slower convergence. Since $I_\obs(\bfpsi^*)$ depends only on the observed data likelihood and is invariant to reparametrization, this motivates choosing $(\bfa, W)$ to minimize $I_\aug(\bfpsi^*)$ in the positive-semidefinite order \citep{Meng1997}.

Using $I_{\aug}(\bfpsi^*)=-\nabla_{\bfpsi}^2 Q(\bfpsi \mid \bfpsi^{(i)}) \big|_{\bfpsi=\bfpsi^{(i)}=\bfpsi^*}$, we can write
\begin{equation*}
I_\aug(\bfpsi^*) = \begin{bmatrix}
I_{\bfgamma, \bfgamma} (W) & I_{\bfgamma, \bftheta} (\bfa, W) & I_{\bfgamma, \sigma^2} (W) \\
\cdot & I_{\bftheta, \bftheta} (\bfa, W) & I_{\bftheta, \sigma^2} (\bfa, W) \\
\cdot & \cdot & I_{\sigma^2, \sigma^2} \\
\end{bmatrix},
\end{equation*}
where each block depends only on the working parameters indicated in parentheses. Following \cite{Tan2023}, rather than minimizing $I_\aug(\bfpsi^*)$ globally over $(\bfa, W)$, we optimize the augmentation for each parameter block conditional on the others. These conditional schemes integrate naturally into an AECM algorithm by inserting an E-step before each conditional maximization step, enabling the augmentation to vary across parameter blocks and iterations. Consequently, it suffices to focus on the diagonal blocks of $I_\aug(\bfpsi^*)$. Since $I_{\sigma^2, \sigma^2}$ is independent of $(\bfa, W)$, our analysis concentrates on two reduced settings in which either the location parameter $\bfgamma$ or scale parameter $\bftheta$ is treated as the only unknown.

\subsection{Optimal working parameters for location}\label{sec:optimal_location}
Suppose that only the location parameter $\bfgamma$ is unknown. The augmented information,
\begin{align*}
I_{\bfgamma, \bfgamma} (W)&= W^\top Z W  + \overline{W}^\top (\Lambda \otimes D_\bftheta^{-1})\overline{W},
\end{align*}
is independent of $\bfa$. Minimizing $I_{\bfgamma, \bfgamma}$ with respect to $W$ leads to Theorem \ref{thm_opt_location}, which implies that Algorithm 1, which alternately updates $\bfm_{\wbfbeta, \bfa W}$ and $\bfgamma$, conditional on $(\bftheta, \sigma^2)$, converges instantly when $W = W_\opt$. Indeed, substituting $W = W_\opt$ into \eqref{eq: gamma update} reveals that the next iterate equals the maximum likelihood estimator, $\hat{\bfgamma}_{\text{MLE}} = (X_s^\top \Sigma_y^{-1} X_s)^{-1} (X_s^\top \Sigma_y^{-1} \bfy)$. Here $X_s = (\bfx_1, \dots, \bfx_T)^\top$, $\Sigma_y = X(\Lambda_1^{-1} \otimes D_\bftheta) X^\top + \sigma^2 I_T$ and $\Lambda_1 \in \mathbbm{R}^{T \times T}$ is a tridiagonal matrix with diagonal $(1.5, 2, \dots, 2, 1)^\top$ and off-diagonal entries $-1$. Alternatively, $\hat{\bfgamma}_{\text{MLE}}$ can be derived by integrating out the latent states to obtain the marginal distribution $\bfy \sim \N(X_s \bfgamma, \Sigma_y)$ and maximizing the log-likelihood with respect to $\bfgamma$.

\begin{theorem} \label{thm_opt_location}
The convergence rate of Algorithm 1, which alternately updates $\bfm_{\wbfbeta, \bfa W}$ and $\bfgamma$, conditional on $(\bftheta, \sigma^2)$, is minimized to zero when
\begin{align*}
W_\opt = (\Lambda \otimes D_\bftheta^{-1} + Z)^{-1}(\Lambda\bfone_{T+1} \otimes D_\bftheta^{-1}) = V_{\wbfbeta, \bfzr} (D_\bftheta^{-1}, 0, \dots, 0)^\top.
\end{align*}
\end{theorem}

\begin{corollary} \label{cor_opt_location}
For a fixed $T$, $\overline{W}_\opt \rightarrow 0$ as the signal-to-noise ratios $\bftheta/\sigma^2 \rightarrow \bfzr$.
\end{corollary}

Expressing $W_\opt = \{ (\Lambda \otimes I_d) + (I_d \otimes D_\bftheta) Z\}^{-1} (\Lambda \bfone_{T+1} \otimes I_d)$ reveals that $W_\opt$ depends on $(\bftheta, \sigma^2)$ only through the signal-to-noise ratios $\bftheta/\sigma^2$, since $(I_d \otimes D_\bftheta) Z$ is a block diagonal matrix with a zero first block and $(t+1)$th block $D_\bftheta \bfx_t \bfx_t^\top /\sigma^2$ for $t=1, \dots, T$.  Moreover, Corollary \ref{cor_opt_location} shows that noncentering is preferred as $\bftheta/\sigma^2 \rightarrow 0$, aligning with \citet{Bitto2019}'s finding that NCP is preferred when coefficients are nearly constant over time. Evaluating $W_\opt$ reduces to computing the first block-column of a block tridiagonal inverse matrix, which can be executed efficiently via Algorithm \ref{Alg: Meurant} in Section \ref{sec:Meurant}.

To examine the behavior of $W_\opt$, we set $\bftheta=(10^{-8}, 10^{-4}, 0.01, 0.1)^\top$, $\sigma^2=1$ and $T\in \{100, 500, 2500\}$, generating covariates  $x_{t1} = 1$ and $x_{tj} \sim \N(0,1)$ for $j>1$ and $t=1, \dots, T$. As shown in Figure \ref{Fig - locationW}, the diagonal elements of $W_{\opt, t}$ generally lie in $[0,1]$ and decrease with both $t$ and $\theta_j$. This pattern aligns with Corollary  \ref{cor_opt_location}, which implies $W_{\opt, t,jj} \rightarrow 1$ as $\theta_j \rightarrow 0$. However, this convergence toward 1 slows significantly for larger $T$.

\begin{figure}[htb!]
\centering
\includegraphics[width=0.83\textwidth]{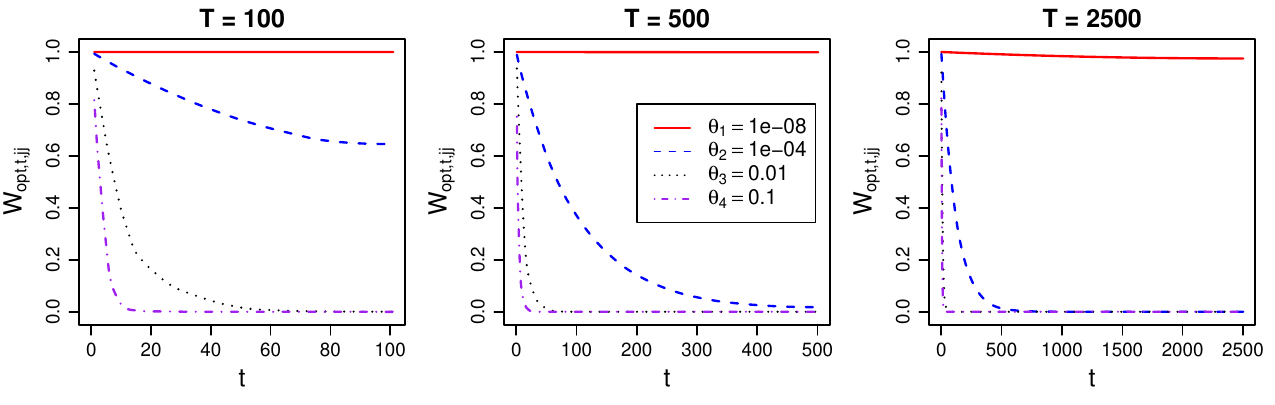}
\caption{Diagonal entries of $W_{\opt, t}$ for $\bftheta=(10^{-8}, 10^{-4}, 0.01, 0.1)^\top$.}
\label{Fig - locationW}
\end{figure}

Intuitively, $W_\opt$ reflects the prior coupling between $\bfgamma$ and the latent states $\bfbeta$. Small signal-to-noise ratios $\bftheta/\sigma^2$ tighten this coupling, making $\bfbeta$ highly informative about $\bfgamma$. Treating $\bfbeta$ as missing data causes CP to suffer from excessive missing information, and thus NCP is preferred. Conversely, large ratios $\bftheta/\sigma^2$ weaken this coupling, favoring CP. This preference for CP strengthens over time as the random walk variance accumulates ($\Var(\bfbeta_t) = (t+1)D _\theta$), making $\bfbeta_t$ progressively less informative about $\bfgamma$.

Finally, Corollary \ref{cor2_opt_location} reveals an insightful connection. The optimal working parameter $W_\opt$, derived to optimize the convergence rate of Algorithm 1, also induces conditional posterior independence between $\wbfbeta$ and $\bfgamma$. Thus, setting $W = W_\opt$ yields a zero convergence rate for the two-block Gibbs sampler that alternately samples $\wbfbeta$ and $\bfgamma$ given $(\bftheta, \sigma^2)$. This highlights $W_\opt$ as a principled parametrization choice extending beyond EM to MCMC and MFVB.

\begin{corollary}\label{cor2_opt_location}
At $W = W_\opt$, $\wbfbeta$ and $\bfgamma$ are conditionally independent given $(\bfy, \bftheta, \sigma^2)$.
\end{corollary}

\subsection{Optimal working parameters for scale}\label{sec: optimal scale}
Next, we consider the scale $\bftheta$ as the only unknown parameter. The entries of $I_{\bftheta, \bftheta}$ are
\begin{multline*}
I_{\theta_j, \theta_j} = \frac{1}{2\theta_j^2} \bigg\{ (T+1) \left(1 - 2a_j + \frac{a_j^2}{2}\right) - \frac{2a_j \bfm_{\wbfbeta_{[j]}, \bfzr I}^\top \Lambda \overline{W}_{[j]} \bfgamma}{\theta_j} + \frac{a_j^2 \bfm_{\wbfbeta_{[j]}, \bfzr W}^\top \bar{V}_{\wbfbeta_{[j]}, \bfzr}^{-1} \bfm_{\wbfbeta_{[j]}, \bfzr W} }{2}
\\
+ \frac{a_j^2}{2} \tr(\bar{V}_{\wbfbeta_{[j]}, \bfzr}^{-1} V_{\wbfbeta_{[j]}, \bfzr} ) \bigg\}, \;\;
I_{\theta_j, \theta_k} =\frac{a_j a_k }{4\sigma^2\theta_j \theta_k} \sum_{t=1}^T x_{tj}\left(  m_{\wbeta_{tj}, 0W} m_{\wbeta_{tk}, 0W}
+ V_{\wbfbeta_t, 0, jk}  \right)x_{tk} \text{ if $j \neq k$}.
\end{multline*}
Here, $\bar{V}_{\wbfbeta_{[j]}, \bfzr}^{-1} \in \mathbbm{R}^{(T+1) \times (T+1)}$ has $(t, t')$ entry equal to the $(j, j)$ element of the $(t, t')$ block of $V_{\wbfbeta, \bfzr}^{-1}$. Unlike in the location case, achieving a zero convergence rate or establishing a clear ordering among information matrices is generally impossible. Moreover, for $d>1$, jointly optimizing over $\bfa$ is considerably more difficult than in the univariate setting. To maintain tractability, we minimize the diagonal entries $I_{\theta_j,\theta_j}$ as a surrogate, as positive definiteness of $I_{\bftheta, \bftheta}$ generally requires its diagonal elements to dominate the off-diagonal terms.

Theorem \ref{thm_aopt} provides the optimal choice of $a_j$ for any given $\overline{W}_{[j]}$, and in the special case $\overline{W}_{[j]} = 0$. Although $I_{\theta_j, \theta_j}$ can in principle be minimized over $a_j$ and $\overline{W}_{[j]}$ jointly, the resulting optimal $\overline{W}_{[j]}$ contains a $\gamma_j^{-1}$ term. This can cause numerical instability when $\gamma_j$ is close to zero, a plausible scenario in TVP models when the $j$th predictor is irrelevant.

\begin{theorem} \label{thm_aopt}
For any $j \in \{1, \dots, d\}$, $I_{\theta_j, \theta_j}$ is minimized with respect to $a_j$ at
\begin{align} \label{gen_a_opt}
a_{\opt, j} = \frac{2(T+1) + 2\bfm_{\wbfbeta_{[j]}, \bfzr I}^\top \Lambda \overline{W}_{[j]} \bfgamma/\theta_j}{T+1 +  \bfm_{\wbfbeta_{[j]}, \bfzr W}^\top \bar{V}_{\wbfbeta_{[j]}, \bfzr}^{-1} \bfm_{\wbfbeta_{[j]}, \bfzr W} + \tr(\bar{V}_{\wbfbeta_{[j]}, \bfzr}^{-1} V_{\wbfbeta_{[j]}, \bfzr} )}.
\end{align}
If $\overline{W}_{[j]} = 0$, then $a_{\opt, j}^* = 2\{1 + C_j/(T+1) \}^{-1}$, where $C_j = \tr\{ \bar{V}_{\wbfbeta_{[j]}, \bfzr}^{-1} (V_{\wbfbeta_{[j]}, \bfzr} + \bfm_{\wbfbeta_{[j]}, \bfzr I} \bfm_{\wbfbeta_{[j]}, \bfzr I}^\top)\}$.
\end{theorem}

To obtain closed-form updates for $\bftheta$, we partitioned the coordinate set $\{1, \dots, d\}$ into $S_\CP$ and $S_\NCP$. Under this scheme, each coordinate $j$ can be assigned to $S_\CP$ or $S_\NCP$ by comparing its augmented information $I_{\theta_j, \theta_j}$ under both parametrizations. For $j \in S_\CP$ ($a_j=0$), $I_{\theta_j, \theta_j} = (T+1)/(2\theta_j^2)$ for any $\overline{W}_{[j]}$. Conversely, for $j \in S_\NCP$ ($a_j=1$, $\overline{W}_{[j]} = 0$), $I_{\theta_j, \theta_j} = \left\{ C_j -(T+1) \right\} / (4\theta_j^2)$. Comparing these quantities yields the assignment rule: $j \in S_\NCP$ if $C_j < 3(T+1)$, and $j \in S_\CP$ otherwise. Notably, this criterion is equivalent to assigning $j\in S_\NCP$ if $a_{\opt,j}^* > 0.5$ and $j\in S_\CP$ otherwise. This equivalence holds because $I_{\theta_j, \theta_j}$ for $j\in S_\CP$ is independent of $\overline W_{[j]}$, and setting $\overline W_{[j]}=0$ recovers $a_{\opt,j}^*$ from Theorem~\ref{thm_aopt}.

To investigate the behavior of $\bfa_\opt^*$, we simulate data from the TVP model. We generate $X$ as in Section \ref{sec:optimal_location} with $d=11$, set $\bfgamma = (0, 1.5, 0, 1.5, \dots, 0)$, $\bftheta=(10^{-10}, 10^{-9}, \dots, 1)^\top$ and $\sigma^2 = 1$, and consider $T \in \{100, 500, 2500\}$. For each configuration, we report the mean and standard deviation of $\bfa$ over $1000$ replications. Unlike $W_\opt$, $\bfa_\opt^*$ varies with $\bfy$ even when $X$ and $\bfpsi$ are held fixed. Figure \ref{Fig - scalea} shows that $a_{\opt, j}^*$ typically lies within $[0,1]$ and decreases toward $0$ as $\theta_j$ increases, consistent with the observation in \citet{Bitto2019} that NCP is preferred when coefficients are nearly constant.

\begin{figure}[htb!]
\centering
\includegraphics[width=0.82\textwidth]{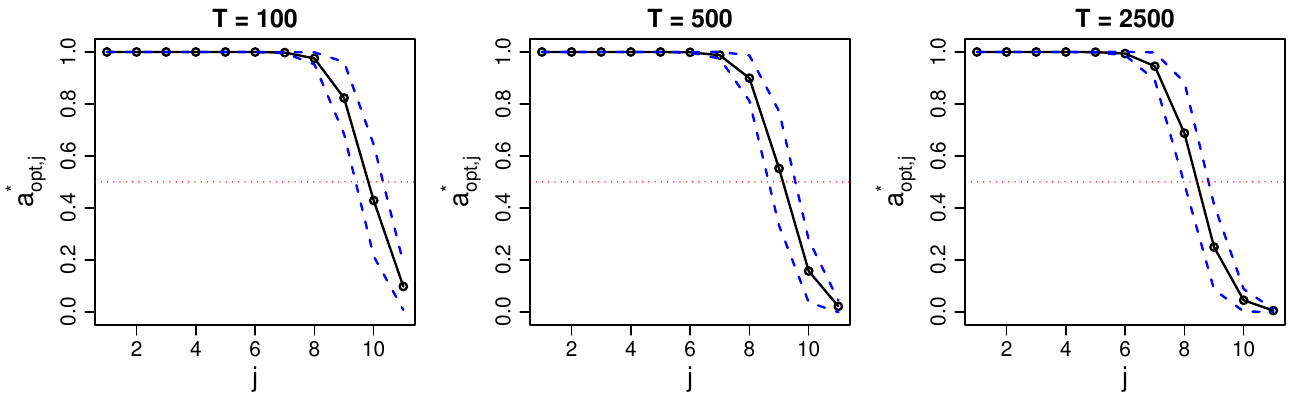}
\caption{Mean (black solid line) and one standard deviation (blue dotted line) of $\bfa_{\opt}^*$ across $1000$ simulations, for $\bftheta=(10^{-10}, 10^{-9}, \dots, 1)^\top$ and $\bfgamma = (0, 1.5, 0, 1.5, \dots, 0)$}
\label{Fig - scalea}
\end{figure}

\subsection{AECM algorithm}\label{sec: AECM algorithm}
Next, we present Algorithm \ref{Alg: AECM} (AECM) to compute the maximum likelihood estimate $\hat{\bfpsi}$, leveraging the optimal augmentation schemes derived in Sections \ref{sec:optimal_location} and \ref{sec: optimal scale} prior to each conditional maximization step.

For the location block, plugging $W = W_\opt$ (Theorem \ref{thm_opt_location}) into \eqref{eq: gamma update} yields the simplified update for $\bfgamma$ presented in Algorithm \ref{Alg: AECM}, which is independent of both $\bfa$ and $\bfm_{\wbfbeta, \bfa W}$. This formulation is computationally more efficient than, yet mathematically equivalent to, directly evaluating the maximum likelihood estimator $\hat{\bfgamma}_{\text{MLE}}$ obtained by integrating out the latent states.

For the scale block, we set $\overline{W}=0$ without loss of generality and partition $\{1, \dots, d\}$ into $S_\CP$ and $S_\NCP$. Specifically, we evaluate $a_{\opt,j}^*$ (Theorem \ref{thm_aopt}) for each $j$, setting $j\in S_\NCP$ ($a_j=1$) if
$a_{\opt,j}^* > 0.5$ and $j\in S_\CP$ ($a_j = 0$) otherwise. Subsequently, $\bftheta$ is updated using \eqref{eq: theta CP} or \eqref{eq: theta NCP} according to the assignment of each $j$. Finally, for convenience, the observation variance is also updated within this block, as its convergence rate is invariant to $(\bfa, W)$.

\setcounter{algorithm}{1}
\begin{algorithm}[tb!]
\caption{(AECM Algorithm)}
\label{Alg: AECM}
Initialize $\bfpsi$.
For $i=0,1,2,\ldots$:
\begin{enumerate}[leftmargin=*]
\item \textbf{Location block:} Set $W = W_\opt$ (Theorem \ref{thm_opt_location}) and update
\[
\bfgamma \leftarrow \{ W^\top Z W + \overline{W}^\top (\Lambda \otimes D_\bftheta^{-1}) \overline{W} \}^{-1}  W_{1:T}^\top X^\top \bfy / \sigma^2.
\]
\item \textbf{Scale block:}
\begin{itemize}[leftmargin=*]
\item Set $\overline W=0$ and compute $\bfa_\opt^*$ (Theorem \ref{thm_aopt}). Set $a_j = \mathbbm{1} \{a_{\opt,j}^* > 0.5\}$ for $j=1,\ldots,d$.
\item Compute $\bfm_{\wbfbeta, \bfa W}$ and $V_{\wbfbeta, \bfa}$.
\item Update $\bftheta$ using \eqref{eq: theta CP} or \eqref{eq: theta NCP} according to the assignment of $j=1, \dots, d$.
\item Update $\sigma^2$ using \eqref{eq: sig2 update}.
\end{itemize}
\item Compute the observed data log-likelihood $\log p(\bfy \mid \bfpsi)$.
\end{enumerate}
\end{algorithm}

In Algorithm \ref{Alg: AECM}, $W_\opt$ (location block) and $\bfa_\opt^*$ (scale block) can be computed jointly since both depend on $V_{\wbfbeta, \bfzr}$, the inverse of a block tridiagonal matrix. Once $\bfa$ is determined, the moments $\bfm_{\wbfbeta, \bfzr I}$ and $V_{\wbfbeta, \bfzr}$ used to evaluate $\bfa_\opt^*$, can be readily transformed into $\bfm_{\wbfbeta, \bfa W}$ and $V_{\wbfbeta, \bfa}$ via the identities in \eqref{eq: m & V}. Thus, each iteration of Algorithm \ref{Alg: AECM} requires only one execution of Algorithm \ref{Alg: Meurant} (Section \ref{sec:Meurant}) to evaluate $V_{\wbfbeta, \bfzr}$, keeping the computational cost of the proposed PNCP scheme comparable to standard CP and NCP implementations.

The data augmentation scheme in Algorithm \ref{Alg: AECM} adapts dynamically across both parameter blocks and iterations, while maintaining monotonic increase of the observed data likelihood. By leveraging the block-specific reparametrization framework of \cite{Tan2023}, the PNCPs underpinning the AECM algorithm, and their associated computational efficiencies, can also be extended to MCMC sampling, while preserving the invariance of the target posterior. This represents a promising direction for future research.

\section{Applications of AECM algorithm for TVP model}\label{sec:EM_Applications}
To evaluate performance, we compare PNCP (Algorithm \ref{Alg: AECM}) against CP ($W = 0$, $\bfa = \bf0$ in scale block) and NCP ($\overline{W} = 0$, $\bfa = \bfone$ in scale block), across one simulated dataset and three empirical applications (with all predictors standardized to zero mean and unit variance). All schemes retain the location block update, which is equivalent to $\hat{\bfgamma}_{\text{MLE}}$. Algorithm \ref{Alg: AECM} is terminated when the relative increase in $\log p(\bfy \mid \bfpsi)$ falls below 0.001\% or upon reaching a maximum of 2000 iterations. Table \ref{tab:EM_results} summarizes the maximized log-likelihoods, iteration counts, runtimes and $|S_\NCP|/d$ ratios, while Figure \ref{fig:em_loglik} illustrates the log-likelihood and $\bfa_\opt^*$ trajectories. Further results are given in Section~\ref{supp:EM_results} of the Supplement. All experiments are performed on an Apple M3 Pro computer with 18 GB of memory, using R and Julia 1.12.7.

For the simulation study, we generate $T=200$ observations from the TVP model with $d=6$, setting $\bfgamma = (0, -0.3, 1.5, 0, -0.3, 1.5)^\top$, $\bftheta = (0, 0, 0.01, 0.05, 0.1, 0.5)^\top$ and $\sigma^2 = 1$. The covariates are generated as $x_{t1} = 1$ and $x_{tj} \sim \N(0,1)$ for $j>1$ and $t=1, \dots, T$. Under this setup, $\beta_{t1}$ is static and insignificant, $\beta_{t2}$ is static but significant, while $\beta_{t3}$ through $\beta_{t6}$ are time-varying with increasing state variances. PNCP automatically assigns NCP to the first three dimensions which have the smallest state variances, and CP to the remaining three. As shown in Table \ref{tab:EM_results}, PNCP achieves faster convergence than both CP and NCP while attaining the highest maximized log-likelihood.

\begin{table}[tb!]
\centering
\caption{AECM algorithm performance under CP, NCP and PNCP. Bold values indicate the highest attained log-likelihood, and the lowest iteration count and runtime.}
\label{tab:EM_results}
\begin{small}
\begin{tabular}{llrrrr}
\hline
& {Scheme} & {Simulation} & {U.S. Macro} & {U.S. Inflation} & {ECB} \\
\hline
\multirow{3}{*}{Log-likelihood}
 & CP & $-406.12$ & $159.18$ & $-170.33$ & $26.14$ \\
 & NCP & $-404.57$ & $159.97$ & $\bm{-165.73}$ & $27.10$ \\
 & PNCP & ${\bf -404.25}$ & ${\bf 159.99}$ & $\bm{-165.73}$ & $\bm{30.49}$ \\
\hline
\multirow{3}{*}{\makecell[l]{Iterations}}
 & CP & $275$ & $919$ & $1567$ & $2000$ \\
 & NCP & $410$ & $1445$ & $\bm{63}$ & $2000$ \\
 & PNCP & $\bm{41}$ & $\bm{871}$ & $\bm{63}$ & $2000$ \\
\hline
\multirow{3}{*}{Time (s)}
 & CP & $0.8$ & $1.9$ & $34.5$ & $\bm{107.9}$ \\
 & NCP & $0.8$ & $2.6$ & $\bm{1.3}$ & $108.1$ \\
 & PNCP & $\bm{0.1}$ & $\bm{1.6}$ & $\bm{1.3}$ & $108.4$ \\
\hline
$|S_\NCP|/d$ & PNCP & $3 / 6$ & $2 / 4$ & $18 / 18$ & $24 / 37$ \\
\hline
\end{tabular}
\end{small}
\end{table}

Next, we consider the {\tt usmacro.update} dataset from the {\tt R} package {\tt bvarsv} \citep{Krueger2024}, which comprises 250 quarterly observations from 1953:Q1 to 2015:Q2. Following \cite{Knaus2021}, we model the U.S. inflation rate using an intercept, its own lag, lagged unemployment rate and lagged treasury bill interest rate ($d=4$). PNCP assigns CP to the two dimensions with the highest signal-to-noise ratios and NCP to the remaining two, effectively separating time-varying coefficients from near static ones. This flexibility allows PNCP to outperform both CP and NCP in maximized log-likelihood and convergence rate.

\begin{figure}[tb!]
\centering
\includegraphics[width=0.9\textwidth]{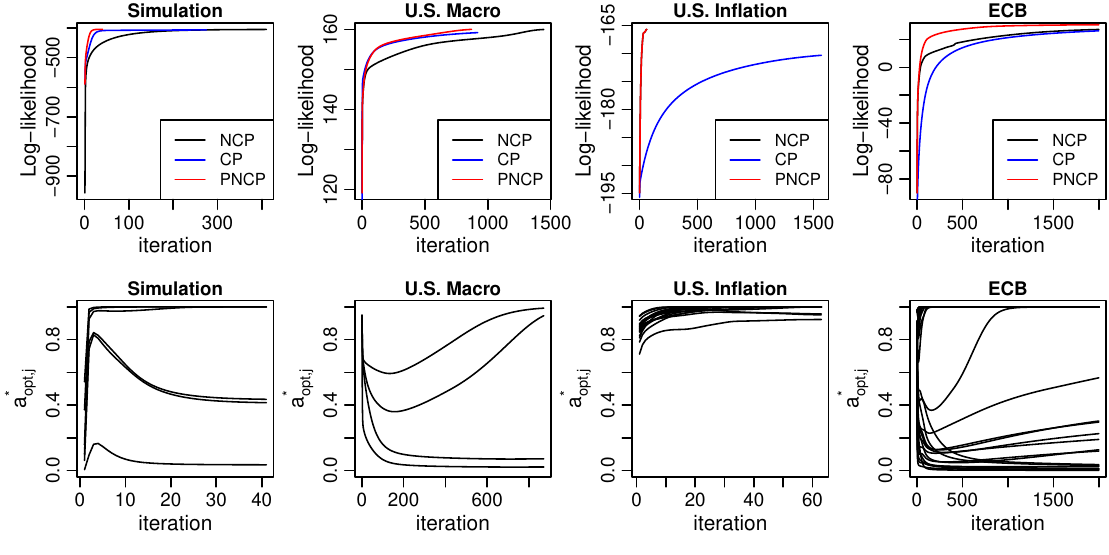}
\caption{Log-likelihood and PNCP's $\bfa_{\opt}^*$ trajectories of AECM algorithm. For U.S. inflation data, PNCP's likelihood trajectory overlaps with NCP.}
\label{fig:em_loglik}
\end{figure}

The quarterly U.S. inflation data (1961:Q1--2015:Q4) studied by \cite{Fruhwirth2021} models inflation $y_t$ for $t=1, \dots, 220$, using $d=18$ predictors. These include an intercept, three autoregressive lags $(y_{t-1}$, $y_{t-2}$, $y_{t-3})$, and 14 potential economic determinants (detailed variable descriptions are given in the Supplement). Due to negligible signal-to-noise ratios indicating minimal time variation in the coefficients, PNCP selects NCP across all dimensions. Thus, PNCP yields log-likelihood trajectories, iteration counts, and runtimes identical to NCP. Notably, NCP achieves a higher log-likelihood of $-165.73$ in just 63 iterations whereas CP converges in 1567 iterations to a lower log-likelihood of $-170.33$. This demonstrates that PNCP can automatically select the optimal parametrization globally.

Finally, we analyze Euro area inflation data from the European Central Bank (ECB) Data Portal, using the economic indicators introduced by \cite{Belmonte2014} and \cite{Bitto2019}. Owing to portal updates and changes in series availability, we consider $T=190$ monthly observations spanning a later period from January 1997 to October 2012. Following \cite{Belmonte2014}, we model core inflation based on the  harmonized index of consumer prices, using a generalized Phillips curve to forecast inflation $h=1$ month ahead. The response variable is the logged first difference of the price index, standardized to unit variance. We include $d=37$ regressors, comprising an intercept, $p=12$ autoregressive lags of the response variable, 13 economic indicators, and 11 monthly dummies. Detailed variable descriptions are given in the Supplement.

For this ECB data, all three schemes are terminated upon reaching the maximum of 2000 iterations. PNCP assigns 24 dimensions to NCP and the remaining 13 to CP. As illustrated in Figure \ref{fig:em_loglik}, PNCP converges much faster than both CP and NCP, achieving a higher log-likelihood of 30.49  at termination, compared to 27.10 for NCP and 26.14 for CP. This highlights the clear advantage of flexibly selecting the parametrization per dimension.

PNCP provides an adaptive alternative to enforcing CP or NCP across all dimensions, consistently accelerating convergence and improving maximized log-likelihoods across all four applications (Table \ref{tab:EM_results}). The dimension-wise optimal parametrization is governed by $\bfa_\opt^*$, which evolves dynamically as optimization proceeds (bottom row of Figure \ref{fig:em_loglik}).

\section{Compute inverse of block tridiagonal matrices} \label{sec:Meurant}
In the AECM algorithm, updates for the model and working parameters rely on $\bfm_{\wbfbeta, \bfa W}$ and $V_{\wbfbeta, \bfa}$, which are expressed in terms of $V_{\wbfbeta, \bfzr}$ and $\bfm_{\wbfbeta, \bfzr W}$ via \eqref{eq: m & V}. As computing $V_{\wbfbeta, \bfzr}$ requires inverting a high-dimensional block tridiagonal matrix of order $(T+1)d$, we adapt recursive algorithms from \cite{Meurant1992} to efficiently and reliably evaluate $\log|V_{\wbfbeta, \bfzr}|$, $\bfm_{\wbfbeta, \bfzr W}$ as well as the first block column and the diagonal and superdiagonal blocks of $V_{\wbfbeta, \bfzr}$.

\begin{algorithm}[tb!]
\caption{(Meurant's Algorithm)}
\label{Alg: Meurant}
Given invertible symmetric block tridiagonal matrix $\Omega \in \mathbb{R}^{nd\times nd}$ and $\bfc = (\bfc_1^\top, \ldots, \bfc_n^\top)^\top \in \mathbb{R}^{nd}$:
\begin{enumerate}[leftmargin=*]
\item Compute $\Delta_1, \dots, \Delta_n$ and $\Sigma_1, \dots, \Sigma_n$:
\begin{gather*}
\begin{cases}
\Delta_1 = D_1 &\\
\Delta_i = D_i - A_i \Delta_{i-1}^{-1} A_i^\top, &i = 2, \dots, n,
\end{cases}\;\;
\begin{cases}
\Sigma_n = D_n, & \\
\Sigma_i = D_i - A_{i+1}^\top\Sigma_{i+1}^{-1}A_{i+1}, &i = n-1, \dots, 1,
\end{cases}
\end{gather*}
\item $\log|\Omega| = \sum_{i=1}^n \log |\Delta_i|$.
\item Compute first block column of $\Omega^{-1}$ given by $(V_1, V_2^\top, \dots, V_n^\top)^\top$:
\[
V_1 = \Sigma_1^{-1}, \quad V_j = \Sigma_j^{-1} A_j V_{j-1} \quad  \text{for} \quad j=2,\dots, n.
\]
\item Compute diagonal blocks and superdiagonal blocks of $\Omega^{-1}$:
\begin{gather*}
\Omega_{j,j}^{-1} =
\begin{cases}
\Sigma_1^{-1} & \text{if $j=1$}, \\
(\Delta_j + \Sigma_j - D_j)^{-1} & \text{if $2 \leq j \leq n-1$}, \\
\Delta_n^{-1} & \text{if $j = n$}.
\end{cases} \;\;
\Omega_{j-1,j}^{-1} =
\Delta_{j-1}^{-1} A_j^\top \Omega_{j,j}^{-1}
\;\; \text{for} \;\; j=2, \dots, n.
\end{gather*}
\item Evaluate $\bfmu = (\bfmu_1^\top, \ldots, \bfmu_n^\top)^\top = \Omega^{-1}\bfc$:
\begin{align*}
\begin{cases}
\widetilde{\bfmu}_1 = \Delta_1^{-1}\bfc_1, \quad\;\;
\widetilde{\bfmu}_{i} = \Delta_i^{-1}(\bfc_i + A_i  \widetilde{\bfmu}_{i-1}),  &i = 2, \dots, n, \\
\bfmu_n = \widetilde{\bfmu}_n, \quad \qquad
\bfmu_{i} = \widetilde{\bfmu}_{i} + \Delta_i^{-1}A_{i+1}^\top\bfmu_{i+1}, &i = n-1, \dots, 1.
\end{cases}
\end{align*}
\end{enumerate}
\end{algorithm}

Consider a general invertible symmetric block tridiagonal matrix $\Omega \in \mathbb{R}^{nd\times nd}$ with diagonal blocks, $D_1, \dots, D_n$ and subdiagonal blocks $-A_2, \dots, -A_n$, where each $D_i$ and $A_i$ is a $d \times d$ invertible matrix. \cite{Meurant1992} derived explicit expressions for $\Omega^{-1}$ using the block LU, block UL and {\em twisted} factorizations of $\Omega$, given respectively by
\begin{align} \label{Omeg LU and UL}
\Omega = (\Delta + L)\Delta^{-1}(\Delta + L^\top)
= (\Sigma + L^\top)\Sigma^{-1}(\Sigma+L)
= (\Phi + \tilde{L})\Phi^{-1}(\Phi + \tilde{L}^\top).
\end{align}
Here $\Delta = \diag(\Delta_1, \dots, \Delta_n)$, $\Sigma = \diag(\Sigma_1, \dots, \Sigma_n)$ and $\Phi = \diag(\Phi_1, \dots, \Phi_n)$ are block diagonal matrices, $L$ is the block strictly lower triangular part of $\Omega$, and $\tilde{L}$ is a twisted block matrix with subdiagonal blocks $-A_2, \dots, -A_j, 0, \dots, 0$ and superdiagonal blocks $0, \dots, 0, -A_{j+1}^\top, \dots, -A_n^\top$ for some twist index $j\in \{2,\dots, n-1\}$.

The recursive identities for $\{\Delta_i\}_{i=1}^n$ and $\{\Sigma_i\}_{i=1}^n$ in Step 1 of Algorithm \ref{Alg: Meurant} are derived by matching the diagonal blocks of the block LU and block UL factorizations to those of $\Omega$. Evaluating $\log |\Omega|$ in Step 2 then reduces to computing the determinant via the block LU factorization in \eqref{Omeg LU and UL}. Theorem 3.4 in \cite{Meurant1992} expresses the $(i,j)$th block of $\Omega^{-1}$ as $U_i V_j^\top$ for $j \geq i$. We enforce uniqueness by setting $U_1 = I_d$, which implies that the first block column of $\Omega^{-1}$ is $(V_1, V_2^\top, \dots, V_n^\top)^\top$, where $V_1$ is symmetric. The update formulas in Step 3 are then derived by premultiplying this first block column by the block UL factorization of $\Omega$ and equating the result to $(I_d, 0, \dots, 0)^\top$.

Theorem 3.1 in \cite{Meurant1992} provides explicit expressions for the $j$th block column of $\Omega^{-1}$, derived by premultiplying it by the twisted factorization of $\Omega$ and equating the product to $(0, \dots, 0, I_d, 0, \dots, 0)^\top$, where $I_d$ is the $j$th block. This yields the diagonal and superdiagonal block update formulas in Step 4. Finally, to evaluate $\bfmu = \Omega^{-1} \bfc$, we exploit the block LU factorization of $\Omega$ to express the linear system $\Omega\bfmu = \bfc$ as $(\Delta + L)\Delta^{-1}(\Delta + L^\top) \bfmu = \bfc$, and solve it via a two-step forward-backward substitution procedure. The first step solves $(\Delta + L)\tilde{\bfmu} = \bfc$ for $\tilde{\bfmu}$ and the second step solves $\Delta^{-1}(\Delta + L^\top)\bfmu = \tilde{\bfmu}$ for $\bfmu$.

The core building blocks of Algorithm \ref{Alg: Meurant} are in Step 1. Steps 2 and 3 serve to compute $\log p(\bfy \mid \bfpsi)$ and $W_\opt$, while Steps 4 and 5 facilitate the updates in Algorithm \ref{Alg: AECM} involving $V_{\wbfbeta, \bfa}$ and $\bfm_{\wbfbeta, \bfa W}$, respectively. Algorithm \ref{Alg: Meurant} is also applied later in this article, where we transition to Bayesian inference and construct an MFVB algorithm for the TVP model using shrinkage priors on the baseline coefficients and state variances.

\section{Bayesian inference for TVP model}\label{sec:Prior_and_density}
The convergence rate of the EM algorithm is closely connected to those of the Gibbs sampler \citep{Sahu1999} and MFVB. In particular, \cite{Tan2014} established that for Gaussian targets, all three methods share identical convergence rates. This theoretical connection motivates developing MCMC and MFVB algorithms using the optimal parametrizations derived for the EM algorithm in Sections \ref{sec:optimal_location} and \ref{sec: optimal scale}. In this work, we focus on MFVB algorithms, leaving MCMC extensions to future research.

For Bayesian inference in the TVP model, we specify a conjugate hierarchical inverse gamma prior for $\sigma^2$, where $\sigma^2 \mid C_0 \sim \IG(c_0, C_0),$ and $C_0\sim \G(g_0, G_0)$. Additionally, we place triple gamma shrinkage priors \citep{Cadonna2020} on $\bftheta$ and $\bfgamma$, as detailed below.

\subsection{Triple gamma shrinkage priors}\label{sec:priors}
In many TVP applications, only a subset of coefficients varies over time. Distinguishing static from time-varying parameters is thus critical to prevent overfitting. Following \cite{Cadonna2020}, we place triple gamma priors on $\theta_j$ and $\gamma_j^2$ for $j=1, \dots, d$. Shrinking $\theta_j$ toward 0 differentiates time-varying from static coefficients, whereas shrinking $\gamma_j$ toward 0 isolates significant from insignificant coefficients. Jointly, these priors classify the coefficients into time-varying ($\theta_j$ is not shrunk to 0), static but significant ($\theta_j$ is shrunk to $0$ but $\gamma_j$ is not) and static and insignificant (both $\theta_j$ and $\gamma_j$ are shrunk to $0$). Here we adopt the alternative representation of the triple gamma prior \citep{Knaus2025}:
\begin{equation}\label{alt_rep_TG}
\begin{aligned}
\theta_j \mid \xi_j, \kappa_B &\sim \G\bracket{\half, \frac{\kappa_B}{4\xi_j}}, \quad
\xi_j \mid c^\xi, \kappa_j \sim \IG\bracket{c^\xi, \kappa_j}, \quad
\kappa_j \mid a^\xi, c^\xi \sim \G\bracket{a^\xi, \frac{a^\xi}{c^\xi}},  \\
\gamma_j^2 \mid \tau_j, \lambda_B &\sim \G\bracket{\half, \frac{\lambda_B}{4\tau_j}}, \quad
\tau_j \mid c^\tau, \lambda_j \sim \IG\bracket{c^\tau, \lambda_j}, \quad
\lambda_j \mid a^\tau, c^\tau \sim \G\bracket{a^\tau, \frac{a^\tau}{c^\tau}}.
\end{aligned}
\end{equation}
The representation in \eqref{alt_rep_TG} is computationally more efficient as the conditional posteriors as well as mean-field variational posteriors for ($\xi_j$, $\tau_j$, $\kappa_B$, $\lambda_B$) involve simpler gamma and inverse gamma distributions rather than generalized inverse Gaussian distributions.

Following \cite{Bitto2019}, let $\surd \theta_j \in \mathbb{R}$ denote both the positive and negative square roots of $\theta_j$. Then \eqref{alt_rep_TG} admits the global-local prior representation, $\surd{\theta_j} \mid \xi_j, \kappa_B \sim \N\left( 0, 2\xi_j/\kappa_B \right)$, $\xi_j \mid a^\xi, c^\xi \sim F(2a^\xi, 2c^\xi)$ and $\gamma_j \mid \tau_j, \lambda_B \sim \N(0, 2\tau_j/\lambda_B)$, $\tau_j \mid a^\tau, c^\tau \sim F(2a^\tau, 2c^\tau)$, while the marginal prior distributions $p(\surd{\theta_j} \mid \kappa_B, a^\xi, c^\xi)$ and $p(\gamma_j \mid \lambda_B, a^\tau, c^\tau)$ can be expressed in terms of the confluent hypergeometric function of the second kind \cite[Theorem 1,][]{Cadonna2020}. The triple gamma prior encompasses several popular shrinkage priors, such as the normal-gamma \citep{Griffin2010}, horseshoe \citep{Carvalho2009} and Bayesian Lasso \citep{Park2008}, as special or limiting cases. Its behavior is governed by pole parameters $(a^\xi, a^\tau)$ controlling mass concentration near 0, tail parameters $(c^\xi, c^\tau)$ controlling heavy-tailedness, and global shrinkage parameters $(\kappa_B, \lambda_B)$ regulating overall shrinkage strength. Setting $(a^\xi, a^\tau, c^\xi, c^\tau)$ to lie in $(0, 0.5)$ yields a pole at 0 and infinite moments. As illustrated in Figure \ref{fig:cases_triple_gamma}, the triple gamma prior ($a^\xi = c^\xi = 0.1$, $\kappa_B = 2$) exhibits a sharp spike at 0 accompanied by heavy tails, effectively shrinking small coefficients while preserving large signals. We treat the hyperparameters  $(a^\xi, a^\tau, c^\xi, c^\tau)$ as fixed, and place hierarchical priors on the global shrinkage parameters: $\kappa_B \mid d_2\sim \G(a^\xi, d_2)$, $d_2 \sim \G\bracket{c^\xi, {2c^\xi}/{a^\xi}}$ and $\lambda_B \mid e_2\sim \G(a^\tau, e_2)$, $e_2 \sim \G\bracket{c^\tau, {2 c^\tau}/{a^\tau}}$.

\begin{figure}[tb!]
\centering
\includegraphics[width=0.6\textwidth]{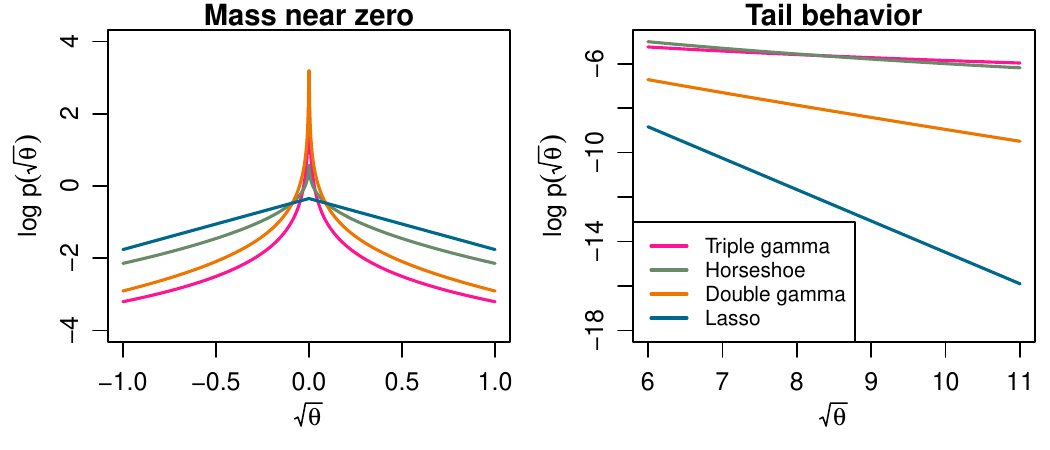} \;
\includegraphics[width=0.37\textwidth]{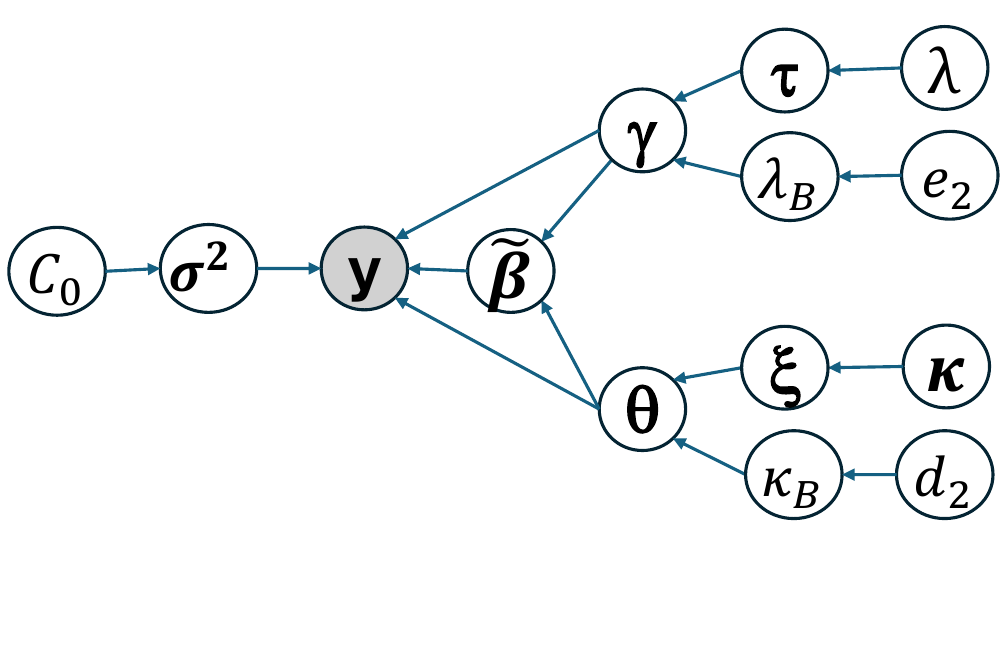}
\caption{Left two panels: shrinkage priors---triple gamma ($a^\xi = c^\xi = 0.1$, $\kappa_B = 2$), horseshoe ($\phi^\xi = 1$), double gamma ($a^\xi = 0.1$, $\kappa_B = 2$), and Bayesian Lasso ($\kappa_B = 2$). Parametrizations for the latter three appear in Section \ref{supp:table_shrinkage} of the Supplement. Right panel: Directed acyclic graph for TVP model.}
\label{fig:cases_triple_gamma}
\end{figure}

Let $\bfxi =(\xi_1, \dots, \xi_d)^\top$, $\bftau =(\tau_1, \dots, \tau_d)^\top$,
$\bfkappa =(\kappa_1, \dots, \kappa_d)^\top$ and $\bflambda =(\lambda_1, \dots, \lambda_d)^\top$. Furthermore, let $\Theta = \{\wbfbeta, \bfgamma, \bftheta, \sigma^2, C_0, \bfxi, \bftau, \bfkappa, \bflambda, \kappa_B, \lambda_B, d_2, e_2\}$ denote the complete set of variables under the PNCP. The right panel of Figure \ref{fig:cases_triple_gamma} illustrates the conditional independence structure of the TVP model. Explicit expressions for the likelihood and priors are provided in Section \ref{supp:joint_density_terms} of the Supplement, and the log joint density is
\begin{align*}
\log \; & p(\bfy, \Theta) = \log p(\bfy \mid \wbfbeta, \bfgamma, \bftheta, \sigma^2)
+ \log p(\wbfbeta \mid \bfgamma, \bftheta)
+ \log p(\sigma^2 \mid C_0) + \log p(C_0)
\\
& + \sum_{j=1}^d \{\log p(\theta_j \mid \xi_j, \kappa_B)
+ \log p(\xi_j \mid \kappa_j) + \log p(\kappa_j)
+ \log p(\gamma_j \mid \tau_j, \lambda_B) + \log p(\tau_j \mid \lambda_j) \\
& + \log p(\lambda_j) \} + \log p(\kappa_B \mid d_2) + \log p(d_2)
+ \log p(\lambda_B \mid e_2) + \log p(e_2).
\end{align*}

\section{Mean-field variational inference for TVP model}\label{sec:MFVB_TVP}
To mitigate the high computational cost of MCMC as the dimension $d$ and time steps $T$ grow, we develop a MFVB scheme for the TVP model that balances accuracy and efficiency. VB approximates the intractable posterior $p(\Theta \mid \bfy)$ with a tractable distribution $q(\Theta)$ by minimizing the Kullback-Leibler divergence, $\KL(q \| p) = \int q(\Theta) \log \{q(\Theta) / p(\Theta \mid \bfy) \} d\Theta$. From the identity, $\log p(\bfy) = \KL(q \| p) + \mL(q)$, minimizing $\KL(q \| p)$ is equivalent to maximizing the evidence lower bound (ELBO) on $\log p(\bfy)$, $\mL(q) = \E_q \{ \log p(\bfy, \Theta) - \log q(\Theta) \},$
where $\E_q(\cdot)$ denotes expectation with respect to $q(\Theta)$.

In MFVB, $q(\Theta)$ is assumed to factorize over a partition $\Theta = (\Theta_1, \dots, \Theta_n)$ such that $q(\Theta) = \prod_{i=1}^n q_i(\Theta_i)$. The optimal variational density $q_i(\Theta_i)$ satisfies
\begin{align} \label{CAVI_update}
q_i(\Theta_i) \propto \exp [\E_{-\Theta_i} \{\log p(\bfy, \Theta)\}], \quad
i=1, \dots, n,
\end{align}
where $\E_{-\Theta_i}$ denotes expectation with respect to $\prod_{j \neq i} q_j(\Theta_j)$ \citep{Ormerod2010}. Under conditionally conjugate priors, these optimal variational densities belong to known parametric families. Iteratively updating the parameters of each variational density conditional on the rest  yields a coordinate ascent algorithm that guarantees a monotonic increase in the ELBO at each iteration until convergence.

\subsection{Optimal densities in MFVB}\label{sec:opt dens}
For the TVP model under the PNCP, we assume the mean-field variational family,
\begin{align*}
q(\Theta) = q(\wbfbeta) \, q(\bfgamma)
\, q(\bftheta) \,  q(\sigma^2) \, q( \bfxi, \bftau) \, q(\bfkappa, \bflambda)
\, q(C_0, \kappa_B, \lambda_B) \, q(d_2, e_2).
\end{align*}
Applying \eqref{CAVI_update} reveals that the optimal variational densities induce further independence structures, factorizing as $q( \bfxi, \bftau) = \prod_{j=1}^d q(\xi_j) q(\tau_j)$, $q(\bfkappa, \bflambda) = \prod_{j=1}^d q(\kappa_j) q(\lambda_j)$, $q(C_0, \kappa_B, \lambda_B) = q(C_0) q(\kappa_B) q(\lambda_B)$ and $q(d_2, e_2) = q(d_2)q( e_2)$. With the exception of $q(\bftheta)$, all optimal variational distributions belong to standard parametric families:
\begin{align*}
q(\wbfbeta) &= \N(\bfmu_\wbfbeta^q, \Sigma_\wbfbeta^q), &
q(\bfgamma) &= \N(\bfmu_\bfgamma^q, \Sigma_\bfgamma^q), &
q(\sigma^2) &= \IG(c^q, C^q), & q(C_0) &= \G(g^q, G^q),
\\
q(\kappa_B) &= \G(a_{\kappa_B}^q, b_{\kappa_B}^q), &
q(\lambda_B) &= \G(a_{\lambda_B}^q, b_{\lambda_B}^q), &
q(d_2) &= \G(a_{d_2}^q, b_{d_2}^q), & q(e_2) &= \G(a_{e_2}^q, b_{e_2}^q), \\
q(\xi_j) &= \IG(a_\xi^q, b_{\xi_j}^q), &
q(\tau_j) &= \IG(a_\tau^q, b_{\tau_j}^q), &
q(\kappa_j) &= \G(a_\kappa^q, b_{\kappa_j}^q), &
q(\lambda_j) &= \G(a_\lambda^q, b_{\lambda_j}^q),
\end{align*}
for $j=1, \dots, d$. Specifically, the covariance and mean of $q(\wbfbeta)$ are given by
\begin{align}
\Sigma_{\wbfbeta}^q &= \left[ \Lambda \otimes \E_q (D_\bftheta^{\bfa-\bfone}) + \blockdiag\left( 0_{d \times d}, (c^q/C^q)\E_q(U)
\right) \right]^{-1}, \label{wbfbeta update}
\\
\bfmu_{\wbfbeta}^q &= \Sigma_{\wbfbeta}^q \big[ \{\Lambda \otimes \E_q(D_\bftheta^{\bfa/2 - \bfone})\} \overline{W} \bfmu_\bfgamma^q + (c^q/C^q)\left(\bfzr_d^\top, (\bfy -  X W_{1:T} \bfmu_\bfgamma^q)^\top X \{I_T \otimes \E_q(D_\bftheta^{\bfa/2})\} \right)^\top \big],
\nonumber
\end{align}
where $U = (I_T \otimes D_\bftheta^{\bfa/2}) X^\top X (I_T \otimes D_\bftheta^{\bfa/2}) = \diag(D_\bftheta^{\bfa/2} \bfx_1 \bfx_1^\top D_\bftheta^{\bfa/2}, \dots, D_\bftheta^{\bfa/2} \bfx_T \bfx_T^\top D_\bftheta^{\bfa/2})$ and $\E_q \{D_\bftheta^{\bfa/2} \bfx_t \bfx_t^\top D_\bftheta^{\bfa/2}\} = D_{\bfx_t}  \E_q\{\bftheta^{\bfa/2}(\bftheta^{\bfa/2})^\top \}D_{\bfx_t}$ for $t=1, \dots, T$. As $\Sigma_\wbfbeta^q$ is the inverse of a block tridiagonal matrix, the required blocks of $\Sigma_\wbfbeta^q$ and $\bfmu_\wbfbeta^q$ can be  efficiently computed using Algorithm \ref{Alg: Meurant}. The parameters of the remaining variational distributions are summarized in Algorithm \ref{alg:TVP_PNCP}, with derivations detailed in Section \ref{supp:derive_opt_dens} of the Supplement.

While the optimal variational densities derived above hold for arbitrary $(W, \bfa)$, the optimal $q(\bftheta)$ falls into a tractable parametric family only under the CP $(W=0, \bfa = \bfzr)$ or NCP $(\overline{W} = 0, a = \bfone)$. Under the NCP, $q(\surd{\bftheta})=\N( \bfmu_{\surd{\bftheta}}^q, \Sigma_{\surd{\bftheta}}^q)$, whereas $q(\bftheta) = \prod_{j=1}^{d} q(\theta_j)$ with $q(\theta_j) = \GIG(p_\theta^q, a_{\theta_j}^q, b_{\theta_j}^q)$ under the CP. Here, $\GIG(p, a, b)$ denotes the three-parameter generalized inverse Gaussian distribution, with probability density function $f(x) = (a/b)^{p/2}x^{p-1} \exp\{ -ax/2 - b/(2x) \} / \{2K_p(\sqrt{ab})\}$ for $x > 0$, $p\in \mathbb{R}$, $a, b > 0$, where $K_p(\cdot)$ is the modified Bessel function of the second kind. For any $r \in \mathbb{R}$, the $r$th moment of $X \sim \GIG(p, a, b)$ is $\E(X^r) = (b/a)^{r/2} K_{p+r}(\sqrt{ab}) / K_{p}(\sqrt{ab})$.

Updating $q(\theta_j)$ under the CP requires evaluating $K_p(\cdot)$, which can suffer from numerical underflow, overflow or precision loss when the argument or order $|p|$ is large. To maintain numerical stability, we dynamically switch between {\tt Julia}'s standard AMOS-based routine and Debye's large-order uniform asymptotic expansion \citep{Olver1974AsymptoticsAS} when evaluating $K_p(\cdot)$. Implementation details and switching criteria are provided in Section \ref{supp:Bessel} of the Supplement.

\subsection{Harnessing optimal working parameters in MFVB}
Like EM algorithms, MFVB is highly sensitive to latent variable parametrization, which governs both convergence rate and posterior approximation accuracy. MFVB thus stands to benefit from dimension-specific parametrization schemes. In this section, we adapt the optimal working parameters derived for the EM algorithm in Sections \ref{sec:optimal_location} and \ref{sec: optimal scale} to the MFVB framework. After highlighting the inherent challenges, we propose two novel implementations for the PNCP: a (dimension-wise) D-CP/NCP scheme that selects between CP and NCP for each dimension, and a continuous location binary scale  (CLBS) scheme that allows an arbitrary $W$ while constraining $\bfa \in \{0, 1\}^d$.

While the AECM framework can dynamically update data augmentation schemes before each conditional maximization step and preserve monotonic increase of the marginal log-likelihood, MFVB lacks this flexibility. Changing the augmentation scheme alters the variational posterior over latent states, thereby redefining the variational objective and breaking the monotonic growth guarantee of the ELBO. Furthermore, as established in Section \ref{sec:opt dens}, the optimal variational density for $q(\bftheta)$ belongs to standard parametric families only under CP or NCP. Preserving closed-form updates alongside ELBO monotonicity thus confines MFVB to static CP or NCP. Finally, as MFVB optimizes variational posteriors rather than point estimates of model parameters, the optimal working parameters from  Sections \ref{sec:optimal_location} and \ref{sec: optimal scale} cannot be evaluated exactly. To resolve this final limitation, both proposed schemes approximate the optimal working parameters $(W_\opt, \bfa_\opt, \bfa_\opt^*)$ from Theorems \ref{thm_opt_location} and \ref{thm_aopt} by replacing $(D_\bftheta^{-1}$, $\sigma^{-2}, \bfgamma)$ with their variational expectations, $(\E_q[D_\bftheta^{-1}], \E_q[\sigma^{-2}], \E_q[\bfgamma])$, yielding $(\hat{W}_\opt, \hat{\bfa}_{\opt}, \hat{\bfa}_{\opt}^*)$.

\subsection{Dimension-wise CP/NCP scheme}\label{sec:dim_wise}
The D-CP/NCP scheme assigns each dimension to CP or NCP by partitioning $\{1, \dots, d\}$ into disjoint subsets, $S_{\CP}$ and $S_{\NCP}$, setting $(a_j=0, W_{[j]}=0)$ for $j\in S_{\CP}$, and $(a_j=1, \overline{W}_{[j]}=0)$ for $j\in S_{\NCP}$. Thus, the optimal variational density factorizes as $q(\bftheta) = q(\bftheta_\CP)q(\bftheta_\NCP)$, where $q(\bftheta_\CP) = \prod_{j\in S_\CP} \GIG(\theta_j \mid p_\theta^q, a_{\theta_j}^q, b_{\theta_j}^q)$ and $q(\surd\bftheta_\NCP) = \N( \bfmu_{\surd\bftheta_\NCP}^q, \Sigma_{\surd\bftheta_\NCP}^q)$. The corresponding parameters are given by $p_\theta^q=-T/2$, $a_{\theta_j}^q = {(a_{\kappa_B}^q a_\xi^q)} / (2b_{\kappa_B}^q b_{\xi_j}^q)$,
\begin{align}
&b_{\theta_j}^q  = \sum_{t=1}^T \big[ (\mu_{\widetilde{\beta}_{tj}}^q -\mu_{\widetilde{\beta}_{t-1,j}}^q)^2 + v_{\widetilde{\beta}_{tj}}^q + v_{\widetilde{\beta}_{t-1, j}}^q \negmedspace  - 2\Cov_q(\widetilde{\beta}_{tj}, \widetilde{\beta}_{t-1,j}) \big] + (\mu_{\widetilde{\beta}_{0j}}^q-\mu_{\gamma_j}^q)^2 + v_{\gamma_j}^q + v_{\widetilde{\beta}_{0j}}^q,
\nonumber \\
&\Sigma_{\surd\bftheta_\NCP}^q = \left\{ \frac{c^q}{C^q} \sum_{t=1}^{T} D_{\bfx_{t, \NCP}} (\Sigma_{\wbfbeta_{t, \NCP}}^q + \bfmu_{\wbfbeta_{t, \NCP}}^q \bfmu_{\wbfbeta_{t, \NCP}}^{q\top}) D_{\bfx_{t, \NCP}} + \frac{a_{\kappa_B^q}}{2 b_{\kappa_B^q}} \diag \left( \frac{a_\xi^q} {\bfb_\xi^q} \right) \right\}^{-1},
\nonumber \\
&\bfmu_{\surd\bftheta_\NCP}^q = \Sigma_{\surd\bftheta_\NCP}^q \frac{c^q}{C^q} \sum_{t=1}^{T} D_{\bfx_{t, \NCP}} \left\{ \bfmu_{\wbfbeta_{t, \NCP}}^q (y_t - \bfx_{t, \NCP}^\top\bfmu_{\bfgamma_\NCP}^q - \bfx_{t, \CP}^\top \bfmu_{\wbfbeta_{t, \CP}}^q) \right.
\label{theta updtae} \\
&\qquad \qquad \left. -\Cov_q(\wbfbeta_{t, \NCP}, \widetilde{\bfbeta}_{t, \CP}) \bfx_{t, \CP} \right\}, \nonumber
\end{align}
where $\mu_\alpha^q$ and $v_\alpha^q$ denote the mean and variance of a variable $\alpha$ under $q(\Theta)$. At each iteration, the partition is dynamically updated using $\hat{a}_{\opt, j}^*$ (Theorem \ref{thm_aopt}), by assigning $j$ to $S_\CP$ if $\hat{a}_{\opt, j}^* \leq 0.5$ and to $S_\NCP$ otherwise, for $j=1, \dots, d$. While this scheme enables adaptive, dimension-specific assignments without sacrificing closed-form updates, ELBO monotonicity is no longer guaranteed when partitions switch between iterations. Thus, temporary fluctuations must be tolerated.

\subsection{Continuous location binary scale scheme}\label{sec:lognormal}
The continuous location binary scale (CLBS) scheme sets $W = \hat{W}_\opt$ (Theorem \ref{thm_opt_location}) and computes $\hat{a}_{\opt,j}$ (Theorem \ref{thm_aopt}) by evaluating \eqref{gen_a_opt} at $W = \hat{W}_\opt$. We then discretize $a_j = \mathbbm{1}\{\hat{a}_{\opt,j} > 0.5\}$ for $j=1,\dots, d$. Recall that when $W= W_\opt$, $\wbfbeta$ and $\bfgamma$ are conditionally independent given $(\bfy, \bftheta, \sigma^2)$ (Corollary \ref{cor2_opt_location}). While setting $W= W_\opt$ exactly is infeasible (since MFVB optimizes variational posteriors rather than point estimates of model parameters), Corollary \ref{cor2_opt_location} suggests that setting $W = \hat{W}_\opt$ induces approximate conditional independence between $\wbfbeta$ and $\bfgamma$. This refinement improves the variational approximation by relaxing the mean-field independence assumption between $q(\wbfbeta)$ and $q(\bfgamma)$.

As closed form updates for $q(\bftheta)$ are unavailable for arbitrary $W$, we reparametrize via $\vartheta_j = \log \theta_j \in \mathbbm{R}$ and assume a Gaussian variational posterior $q(\bfvartheta) = \N(\bfmu_\bfvartheta^q, \Sigma_{\bfvartheta}^q)$, where $\bfvartheta = (\vartheta_1, \dots, \vartheta_d)^\top$ and $\Sigma_{\bfvartheta}^q = C_{\bfvartheta}^q C_{\bfvartheta}^{q \top}$ denotes the Cholesky factorization. Although $(\bfmu_\bfvartheta^q, C_{\bfvartheta}^q)$ lack closed form updates, the ELBO remains tractable, as required expectations under $q(\bfvartheta)$ can be evaluated analytically via the identities: $\E_q(\log \theta_j) = \mu_{\vartheta_j}^q$, $\E_q(\theta_j^{r_j}) = \exp(r_j \mu_{\vartheta_j}^q + r_j^2 v_{\vartheta_j}^q/2)$ and $\E_q(\theta_j^{r_j} \theta_k^{r_k}) = \E_q(\theta_j^{r_j}) \E_q(\theta_k^{r_k}) \exp(r_jr_k v_{\vartheta_j, \vartheta_k}^q)$ $\forall r_j, r_k \in \mathbb{R}$. The ELBO is then maximized with respect to $(\bfmu_\bfvartheta^q, C_{\bfvartheta}^q)$ via gradient ascent using Adam \citep{Kingma2015}.

Given the use of gradient ascent and dynamic updates of the working parameters $(\bfa, W)$ at each iteration, monotonic increase of the ELBO is not strictly guaranteed. In practice, however, the ELBO trajectory remains predominantly monotonic. This stability is driven by the continuous updates of $W$ steadily refining the variational approximation, while the thresholding of $\bfa \in \{0,1\}^d$ suppresses frequent switching.

\subsection{MFVB algorithm}\label{sec:MFVB_alg}

\begin{algorithm}[tb!]
\caption{(MFVB algorithm)}
\label{alg:TVP_PNCP}
Set prior hyperparameters $(c_0, g_0, G_0, a^\xi, c^\xi, a^\tau, c^\tau)$. For $i=1, 2, \ldots$:
\begin{enumerate}[leftmargin=*]
\item Update working parameters $(\bfa, W)$ under D-CP/NCP and CLBS schemes accordingly.
\item Update  $q(\wbfbeta)$ as in \eqref{wbfbeta update} using Algorithm \ref{Alg: Meurant}.
\item Update $q(\bftheta)$ under D-CP/NCP scheme using \eqref{theta updtae}
or $q(\bfvartheta) = \N(\bfmu_\bfvartheta^q, \Sigma_\bfvartheta^q)$ via gradient ascent using Adam under CLBS scheme.
\item Update $q(\xi_j)$: $a_\xi^q = c^\xi + 1/2$, $b_{\xi_j}^q = {a_\kappa^q}/{b_{\kappa_j}^q} + {a_{\kappa_B}^q\E_q(\theta_j)}/ (4b_{\kappa_B}^q)$ for $j=1, \dots, d$.

\item Update $q(\kappa_B)$: $a_{\kappa_B}^q = a^\xi + d/2$, $b_{\kappa_B}^q = a_{d_2}^q/b_{d_2}^q + (a_\xi^q/4) \sum_{j=1}^d \E_q(\theta_j)/b_{\xi_j}^q$.

\item Update $q(\kappa_j)$: $a_\kappa^q = a^\xi + c^\xi$, $b_{\kappa_j}^q = {a_\xi^q}/{b_{\xi_j}^q} + {{a}^\xi}/{{c}^\xi}$ for $j=1, \dots, d$.

\item Update $q(d_2)$: $a_{d_2}^q = a^\xi + c^\xi$, $b_{d_2}^q = {a_{\kappa_B}^q}/{b_{\kappa_B}^q} + {2{c}^\xi}/{{a}^\xi}$.

\item Update $q(\sigma^2)$: $c^q = c_0 + T/2$ and
\begin{align*}
C^q &= \frac{g^q}{G^q} + \frac{1}{2}\sum_{t=1}^T \big[ (y_t - \bfx_t^\top W_t \bfmu_\bfgamma^q)^2 - 2(y_t - \bfx_t^\top W_t \bfmu_\bfgamma^q) \bfx_t^\top \E_q(D_\bftheta^{\bfa/2}) \bfmu_{\wbfbeta_t}^q \\
&\quad + \bfx_t^\top W_t \Sigma_\bfgamma^q W_t^\top \bfx_t + \tr\{(\bfmu_{\wbfbeta_t}^q\bfmu_{\wbfbeta_t}^{q\top} + \Sigma_{\wbfbeta_t}^q) D_{\bfx_t} \E_q(\bftheta^{\bfa/2}\bftheta^{\bfa/2 \top})D_{\bfx_t} \} \big].
\end{align*}

\item Update $q(C_0)$: $g^q = g_0 + c_0$ and $G^q = G_0 + {c^q}/{C^q}$.

\item Update $q(\bfgamma)$:
\begin{align*}
\Sigma_\bfgamma^q &= \left[ \frac{c^q}{C^q} \sum_{t=1}^T W_t^\top \bfx_t \bfx_t^\top W_t + \overline{W}^\top \{\Lambda \otimes \E_q(D_\bftheta^{-1})\} \overline{W} + \frac{a_{\lambda_B}^q}{2 b_{\lambda_B}^q} \diag\left( \frac{a_\tau^q}{\bfb_\tau^q}\right) \right]^{-1},
\\
\bfmu_\bfgamma^q &= \Sigma_\bfgamma^q \left[ \frac{c^q}{C^q}  \sum_{t=1}^T \{y_t - \bfx_t^\top \E_q (D_\bftheta^{\bfa/2}) \mu_{\wbfbeta_t}^q\} W_t^\top\bfx_t + \overline{W}^\top \{ \Lambda \otimes \E_q(D_\bftheta^{\bfa/2-1}) \} \mu_{\wbfbeta}^q \right].
\end{align*}
\item Update $q(\tau_j)$: $a_\tau^q = c^\tau + 1/2$, $b_{\tau_j}^q = {a_\lambda^q}/{b_{\lambda_j}^q} + {a_{\lambda_B}^q(\mu_{\gamma_j}^{q2} + v_{\gamma_j}^q)}/(4b_{\lambda_B}^q)$ for $j=1, \dots, d$.

\item Update $q(\lambda_B)$: $a_{\lambda_B}^q = a^\tau + d/2$, $b_{\lambda_B}^q = a_{e_2}^q/b_{e_2}^q + (a_\tau^q/4) \sum_{j=1}^d (\mu_{\gamma_j}^{q2} + v_{\gamma_j}^q)/b_{\tau_j}^q$.

\item Update $q(\lambda_j)$: $a_\lambda^q = a^\tau + c^\tau$, $b_{\lambda_j}^q = {a_\tau^q}/{b_{\tau_j}^q} + {{a}^\tau}/{{c}^\tau}$ for $j=1, \dots, d$.

\item Update $q(e_2)$: $a_{e_2}^q = a^\tau + c^\tau$, $b_{e_2}^q = {a_{\lambda_B}^q}/{b_{\lambda_B}^q} + {2{c}^\tau}/{{a}^\tau}$.
\item Evaluate the ELBO.
\end{enumerate}
\end{algorithm}

Algorithm \ref{alg:TVP_PNCP} summarizes the proposed MFVB algorithm under D-CP/NCP and CLBS. Variational parameters are updated iteratively until convergence, which is monitored via the ELBO detailed in Section \ref{supp:derive_elbo} of the Supplement. Although the ELBOs differ between the two schemes (D-CP/NCP employs the optimal $q(\bftheta)$ derived under the mean-field assumption, whereas CLBS assumes a Gaussian variational approximation for $\log \bftheta$), both yield valid, directly comparable lower bounds to $\log p(\bfy)$. Algorithm \ref{alg:TVP_PNCP} is terminated when the relative change in the ELBO falls below $10^{-6}$ or a maximum of 2000 iterations is reached.

Each iteration begins by updating the working parameters $(\bfa, W)$ in Step 1, which may alter the latent state parametrization. Step 2 then updates $q(\wbfbeta)$ using the revised $(\bfa, W)$. This update requires evaluating the expectations $\E_q(D_\bftheta^{\bfa-1})$, $\E_q(D_\bftheta^{\bfa/2-1})$, $\E_q(D_\bftheta^{\bfa/2})$ and $\E_q\{ D_\bftheta^{\bfa/2}(D_\bftheta^{\bfa/2})^\top\}$. Under the D-CP/NCP scheme, switching between CP and NCP changes $q(\bftheta)$, which can cause required expectations to be undefined under the latest variational posterior. In such cases, these expectations are approximated, for instance, by setting $\E_q(1/\theta_j) = 1/\E_q(\theta_j)$ when switching from NCP to CP. Guided by the directed acyclic graph in Figure \ref{fig:cases_triple_gamma}, the remaining variational distributions are updated across three sequential blocks: $\bftheta$, followed by $\sigma^2$ and $\bfgamma$, each alongside their respective dependent variables.

\section{Applications of MFVB algorithm for TVP model}\label{sec:VA_Applications}
We compare our proposed D-CP/NCP and CLBS schemes against MFVB under CP and NCP, which correspond to assigning all coordinates in D-CP/NCP to $S_\CP$ or $S_\NCP$ respectively. CLBS can also implement CP and NCP by fixing $(\bfa=0, W=0)$ and $(\bfa=\bfone, \overline{W} = 0)$, but D-CP/NCP remains superior as it utilizes exact optimal variational densities for $q(\bftheta)$, whereas CLBS relies on a Gaussian approximation for $\log \bftheta$. Prior hyperparameters are set to $c_0=2.5$, $g_0=5$, $G_0=g_0/(c_0-1)$ \citep{Cadonna2020}, and $a^\xi=a^\tau=c^\xi=c^\tau=0.4$.

Performance is evaluated across the four datasets described in Section \ref{sec:EM_Applications} against MCMC benchmarks generated via the \texttt{shrinkTVP} package \citep{Knaus2021}. For each dataset, we draw 60,000 MCMC samples, discarding the first 10,000 as burn-in and thinning by a factor of 10. As MFVB guarantees convergence only to local modes, careful initialization is crucial. For CLBS, Adam uses an elevated learning rate of 0.1, enabled by analytic gradients.

\begin{table}[tb!]
\centering
\caption{MFVB performance under CP, NCP, D-CP/NCP and CLBS schemes. Bold values indicate the highest ELBO, and the lowest iteration count and runtime.}
\label{tab:VA_results}
\begin{small}
\begin{tabular}{llrrrr}
\hline
{} & {Scheme} & {Simulated} & {U.S. Macro} & {U.S. Inflation} & {ECB} \\
\hline
\multirow{4}{*}{ELBO}
 & CP & $-848.7$ & $-153.1$ & $-1997.4$ & $-3430.2$ \\
 & NCP & $-847.8$ & $-162.7$ & $\bm{-1956.5}$ & $-3346.7$ \\
 & D-CP/NCP & $-841.6$ & $-152.7$ & $\bm{-1956.5}$ & $-3342.7$ \\
 & CLBS & $\bm{-840.8}$ & $\bm{-147.6}$ & $-1959.7$ & $\bm{-3339.0}$ \\
\hline
\multirow{5}{*}{Iterations}
 & CP & $222$ & $449$ & $550$ & $1,238$ \\
 & NCP & $630$ & $811$ & $\bm{47}$ & $854$ \\
 & D-CP/NCP & $\bm{56}$ & $329$ & $\bm{47}$ & $539$ \\
 & CLBS & $302$ & $\bm{98}$ & $443$ & $\bm{395}$ \\
\hline
\multirow{5}{*}{Time (sec)}
 & CP & $0.5$ & $0.4$ & $6.8$ & $49.8$ \\
 & NCP & $1.0$ & $0.6$ & $\bm{0.5}$ & $33.5$ \\
 & D-CP/NCP & $\bm{0.1}$ & $0.6$ & $1.0$ & $38.2$ \\
 & CLBS & $0.8$ & $\bm{0.2}$ & $10.7$ & $\bm{33.3}$ \\
 & MCMC & $34.3$ & $27.2$ & $125.6$ & $444.7$ \\
\hline
\multirow{2}{*}{$|S_\NCP|/d$}
 & D-CP/NCP & $3 / 6$ & $1 / 4$ & $18 / 18$ & $33 / 37$ \\
 & CLBS & $2 / 6$ & $1 / 4$ & $17 / 18$ & $34 / 37$ \\
\hline
\end{tabular}
\end{small}
\end{table}

Table~\ref{tab:VA_results} summarizes key metrics, including converged ELBOs, iteration counts, runtimes, and $|S_\NCP|/d$ ratios, while Figure~\ref{fig:va_elbo} displays the ELBO trajectories. Marginal posterior accuracy is evaluated in Figure~\ref{fig:va_iae} using $\mathrm{Accuracy}(q)=100\{1-\mathrm{IAE}(q)/2\}\%$, where the integrated absolute error $\mathrm{IAE}(q)=\int|q(z)-q_{\mathrm{MCMC}}(z)|\,\mathrm{d}z$ and $q_{\mathrm{MCMC}}$ is the MCMC kernel density estimate. Higher accuracy values indicate closer agreement with MCMC. Additional results are provided in Section~\ref{supp:VA_results} of the Supplement.

\begin{figure}[tb!]
\centering
\includegraphics[width=0.88\textwidth]{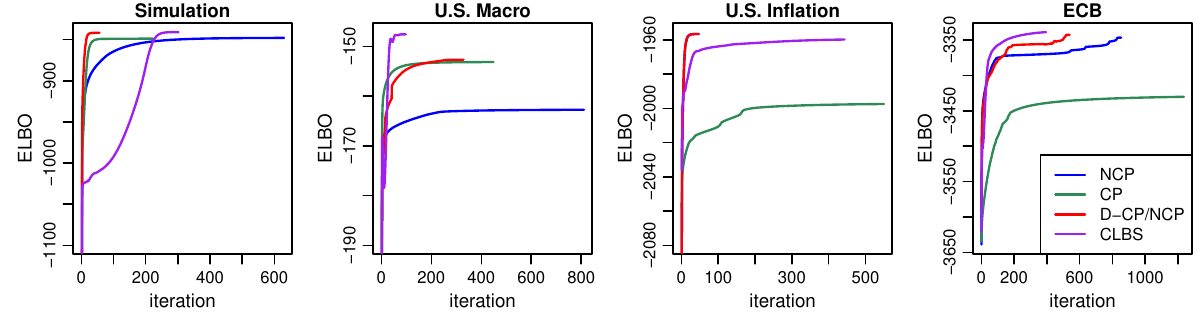}
\caption{ELBO trajectories of MFVB under CP, NCP, D-CP/NCP and CLBS schemes. For U.S. inflation data, D-CP/NCP's ELBO trajectory overlaps with NCP.}
\label{fig:va_elbo}
\end{figure}

\begin{figure}[tb!]
\centering
\includegraphics[width=0.86\textwidth]{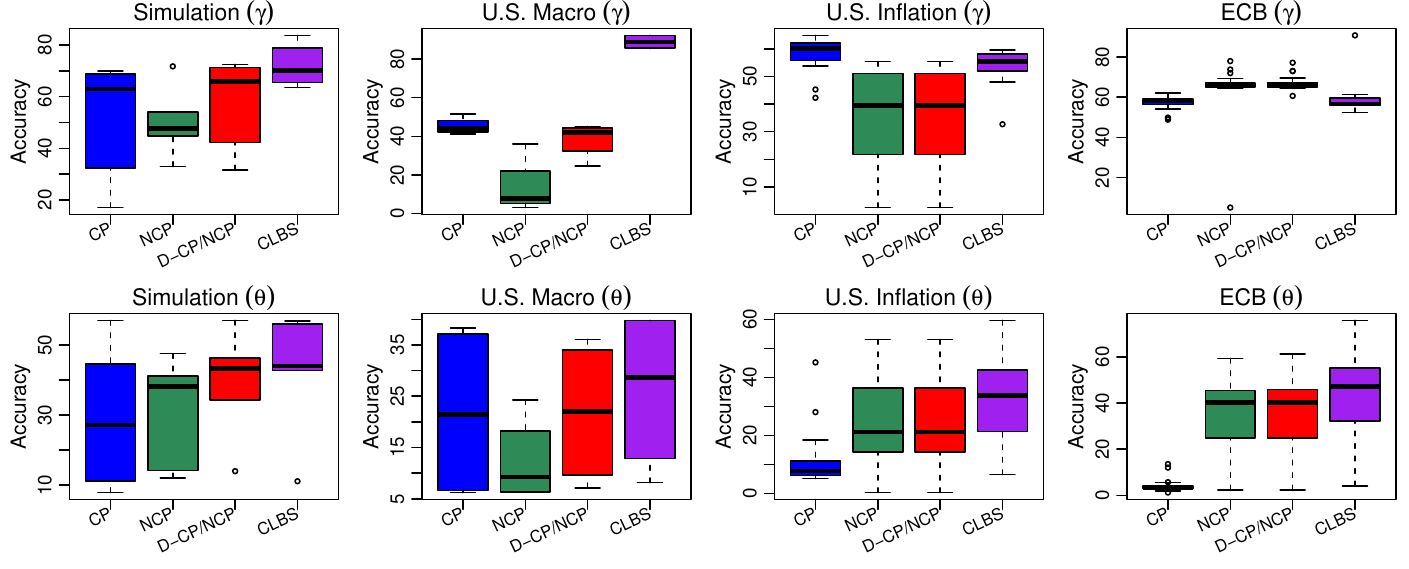}
\caption{IAE accuracy of MFVB algorithms under CP, NCP, D-CP/NCP and CLBS.}
\label{fig:va_iae}
\end{figure}

Recall that the simulated data were generated using $\bfgamma = (0, -0.3, 1.5, 0, -0.3, 1.5)^\top$, $\bftheta = (0, 0, 0.01, 0.05, 0.1, 0.5)^\top$ and $\sigma^2 = 1$. As shown in Table \ref{tab:VA_results}, CLBS achieves the highest ELBO, followed closely by D-CP/NCP, with both substantially outperforming CP and NCP. D-CP/NCP also requires the shortest runtime and fewest iterations to converge. Figure \ref{fig:va_iae} shows that CLBS and D-CP/NCP generally attain higher IAE accuracy than CP and NCP, with CLBS performing best overall. As shown in Figure \ref{fig:va_simulation_marginal}, CP overestimates the posterior mode of $\theta_1$, $\theta_2$ and $\theta_3$, whereas NCP underestimates the posterior variance of $\gamma_4$, $\gamma_5$ and $\gamma_6$. D-CP/NCP automatically selects the better parametrization for each dimension (NCP for $j \in \{1,2,3\}$ and CP for $j \in \{4,5,6\}$). CLBS differs slightly by assigning $j=3$ to CP. It also captures the wide spread of the MCMC posterior more effectively than D-CP/NCP.

\begin{figure}[tb!]
\centering
\includegraphics[width=\textwidth]{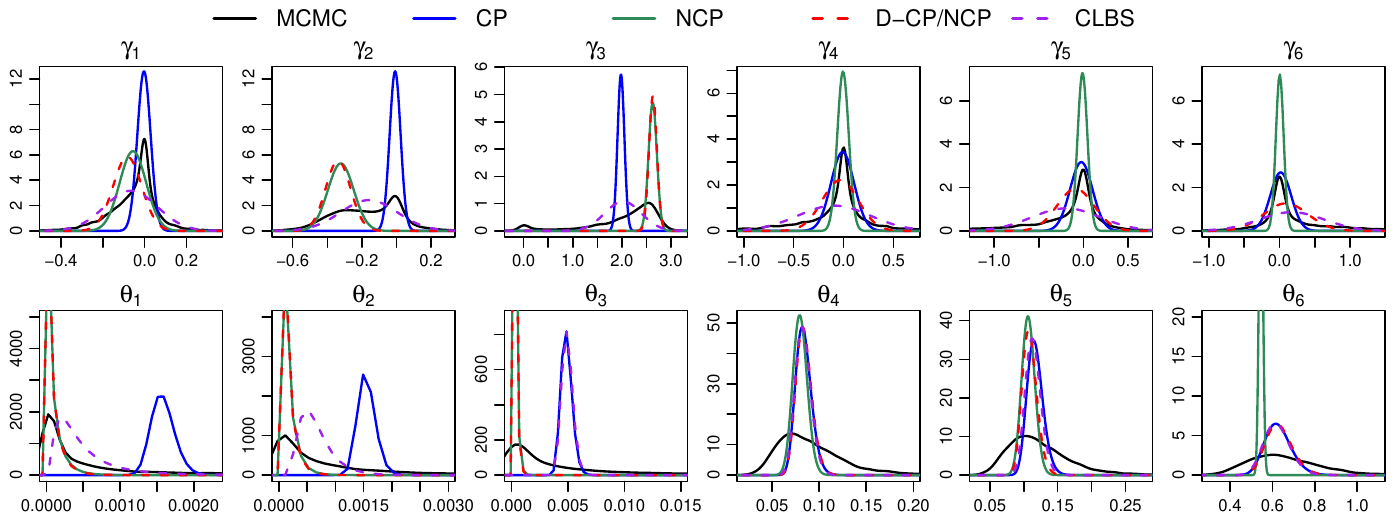}
\caption{Marginal posterior densities for $\bfgamma$ and $\bftheta$ in the simulated data.}
\label{fig:va_simulation_marginal}
\end{figure}

For the U.S. macroeconomic data, both D-CP/NCP and CLBS achieve higher ELBOs than CP and NCP. As illustrated in Figure \ref{fig:va_usmacro_marginal}, NCP underestimates the posterior variance of $\gamma_1$, $\gamma_2$ and $\gamma_3$, whereas CP overestimates the posterior mode of $\theta_4$. Both adaptive methods mitigate these issues by selecting CP for $j \in \{1,2,3\}$ and NCP for $j=4$. Despite sharing the same $\bfa$, CLBS sets $W=W_\opt$, whereas D-CP/NCP restricts $W$ to follow $\bfa$ (enforcing a full CP or NCP in each dimension). This added flexibility enables CLBS to achieve superior ELBO, runtime, and IAE accuracy across all dimensions of $\bfgamma$ and $\bftheta$.

\begin{figure}[tb!]
\centering
\includegraphics[width=0.87\textwidth]{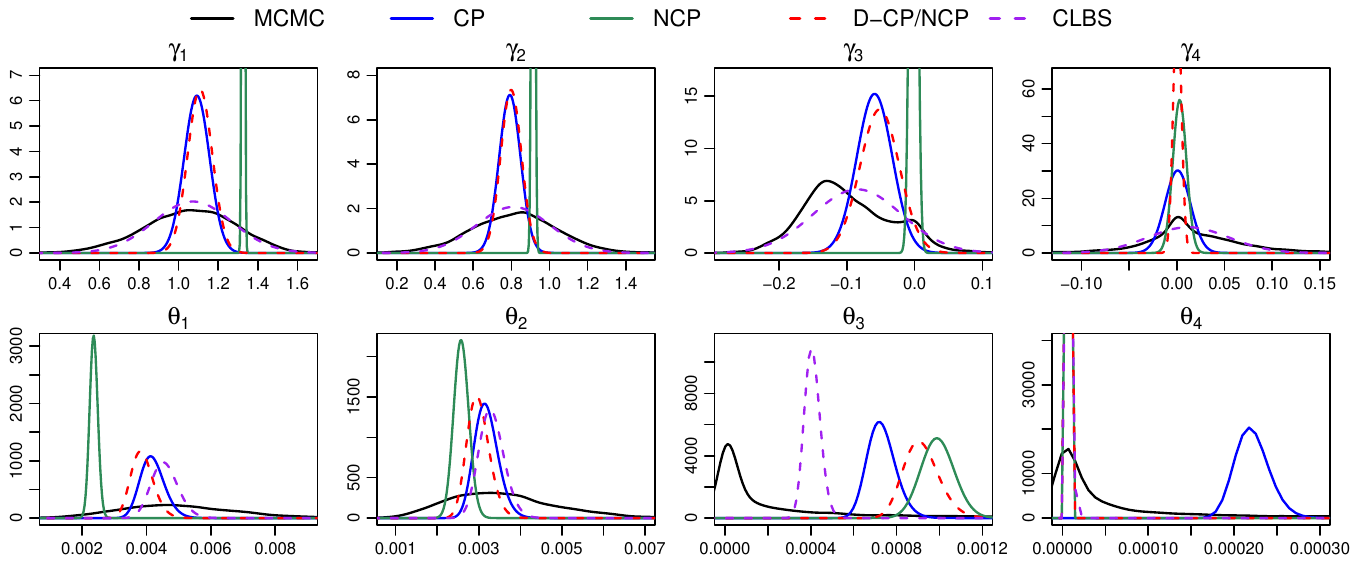}
\caption{Marginal posterior densities for $\bfgamma$ and $\bftheta$ in the U.S. Macro data.}
\label{fig:va_usmacro_marginal}
\end{figure}

For the U.S. inflation data, NCP achieves the highest ELBO and shortest runtime. D-CP/NCP matches this performance by assigning all 18 dimensions to NCP from the outset, incurring only minor computational overhead to update working parameters. Conversely, CLBS assigns the intercept to CP, yielding a lower ELBO than NCP but still outperforming CP. Regarding marginal posterior accuracy, CP aligns most closely with MCMC for $\bfgamma$ but performs worst for $\bftheta$, whereas NCP and D-CP/NCP achieve lower accuracy for $\bfgamma$ but higher accuracy for $\bftheta$. CLBS achieves the best overall balance, nearly matching CP for $\bfgamma$ while outperforming all alternatives for $\bftheta$. Overall, NCP and D-CP/NCP favor predominantly static coefficients, whereas CP, CLBS and MCMC support more extensive time variation.

For the ECB data, CLBS achieves the highest ELBO in the fewest iterations, followed closely by D-CP/NCP. Of the $d=37$ dimensions, D-CP/NCP assigns 33 to NCP, reserving CP for only the {\tt Intercept}, {\tt I\_1MO}, {\tt M3} and {\tt LOANS}variables. CLBS differs slightly by assigning {\tt LOANS} to NCP instead. In terms of marginal posterior accuracy, D-CP/NCP mirrors NCP, and generally outperforms CP for both $\bfgamma$ and $\bftheta$. CLBS improves on CP for $\bfgamma$ but falls short of D-CP/NCP. Nevertheless, it achieves the best overall performance for $\bftheta$. Notably, for $\gamma_1$, NCP and CP yield accuracies of 5.1\% and 57.0\%, whereas D-CP/NCP and CLBS (both assigning $j=1$ to CP) achieve 60.7\% and 90.8\%, respectively. Overall, the primary time variation is concentrated in the intercept, {\tt I\_1MO} and {\tt M3}, whereas the remaining predictors undergo substantial shrinkage, particularly under NCP, D-CP/NCP and CLBS.

Across all four applications, D-CP/NCP consistently matches or exceeds the ELBO of the better parametrization (CP or NCP) in fewer iterations. By accommodating an arbitrary $W$ and binary $\bfa$, CLBS flexibly deviates from pure CP or NCP per dimension, achieving the highest ELBO in three of the four applications (excluding U.S. inflation) and consistently delivering strong accuracy gains for $\bftheta$. Altogether, D-CP/NCP and CLBS provide effective automated alternatives to MFVB via standard CP or NCP, significantly boosting both convergence speed and posterior approximation accuracy.

\section{Conclusion}\label{sec:Conclusion}
We propose a dimension-specific PNCP for TVP models, introducing scale and location working parameters $(\bfa, W)$ to interpolate between CP and NCP. By optimizing the EM algorithm's convergence rate via the augmented information matrix, we establish a theoretical foundation for working parameter selection that adaptively guides parametrization toward CP for strongly time-varying coefficients, and NCP for nearly constant ones. We operationalize these insights via an AECM algorithm that dynamically switches data augmentation schemes across parameter blocks and iterations, while preserving tractable updates and monotonic log-likelihood growth. Empirically, PNCP consistently achieves faster convergence and higher log-likelihoods than pure CP or NCP.

Motivated by the close connections among the convergence rates of EM, MFVB and MCMC, we extend PNCP and its optimal working parameters to MFVB, leaving MCMC for future work. For Bayesian inference in the TVP model, we apply triple gamma shrinkage priors to both state variances and coefficients to mitigate overfitting. We address the limitations of implementing PNCP within MFVB by proposing two schemes: D-CP/NCP and CLBS. The former selects between CP and NCP on a per-dimension basis, ensuring tractable updates and monotonic ELBO growth, except during parametrization shifts. The latter sets $W = W_\opt$ to relax MFVB's independence assumption between latent states and coefficients, thresholding $\bfa \in \{0,1\}^d$ accordingly. Due to the intractability of $q(\bftheta)$, we adopt a Gaussian variational approximation for $\log \bftheta$ updated via gradient ascent.

Across both simulation and macroeconomic applications, D-CP/NCP consistently outperforms CP and NCP in convergence rate and posterior approximation accuracy. CLBS often further improves upon D-CP/NCP through its added flexibility, consistently yielding more accurate marginal posterior approximations for process variances. Crucially, all MFVB algorithms are an order of magnitude faster than MCMC, with D-CP/NCP and CLBS offering substantially more effective, automated alternatives to global CP or NCP.

Additionally, we present an efficient algorithm for operations involving the inverse of block tridiagonal matrices, that is essential to both AECM and MFVB for computing the mean and precision matrix of the conditional posterior of the latent states.

This work suggests several promising directions for future research. First, extending PNCP and its optimal working parameters to MCMC offers a powerful alternative to ASIS for accelerating exact Bayesian inference. Second, adapting D-CP/NCP and CLBS to dynamic factor models and non-Gaussian state space settings (e.g. stochastic volatility models) will broaden their econometric utility. Finally, integrating PNCP with structured variational families or normalizing flows can further bridge the gap between computational speed and posterior approximation accuracy.

\begingroup
\begin{singlespace}
  \footnotesize
  \setlength{\bibsep}{0pt}
  \bibliography{ref.bib}
\end{singlespace}
\endgroup

\clearpage
\appendix
\setcounter{figure}{0}
\setcounter{table}{0}
\setcounter{equation}{0}
\makeatletter
\renewcommand{\thefigure}{S\@arabic\c@figure}
\renewcommand{\thetable}{S\@arabic\c@table}
\renewcommand{\thesection}{S\@arabic\c@section}
\renewcommand{\theequation}{S\@arabic\c@equation}
\makeatother
\begin{center}
  {\LARGE\bf Supplementary Material}
\end{center}
\spacingset{1.2}
\section{Derivation of AECM algorithm}\label{supp:EM}
This section provides detailed derivations for the AECM algorithm presented in Section \ref{sec:EM}. The tridiagonal matrix $\Lambda \in \mathbb{R}^{(T+1) \times (T+1)}$ defined in Section \ref{sec:TVP_model} and its inverse are given by

\[
\Lambda = \begin{bmatrix}
2 & -1 \\
-1 & 2 & -1 \\
& \ddots & \ddots & \ddots &\\
& & -1 & 2 & -1 \\
& & & -1 & 1
\end{bmatrix}, \quad
\Lambda^{-1} = \begin{bmatrix}
1 & 1 & 1 & \dots & 1\\
1 & 2 & 2 & \dots & 2\\
\vdots & \vdots & \ddots & \ddots &\\
1 & 2 & \dots & T & T \\
1 & 2 & \dots & T & T+1
\end{bmatrix},
\]
and $|\Lambda| = 1$.

\subsection{Marginal log-likelihood}
We first consider the distribution of $\bfy$ in the CP. Integrating out the latent states $\bfbeta$ yields a marginal Gaussian distribution $\bfy \mid \bfpsi \sim \N(\bfmu_\bfy, \Sigma_\bfy)$, where
\begin{align*}
\bfmu_\bfy &= \E(\bfy \mid \bfpsi) = \E\{\E(\bfy \mid \bfbeta, \bfpsi)\} = \E(X\bfbeta_{1:T}) = X_s \bfgamma,\\
\Sigma_\bfy &= \Var(\bfy \mid \bfpsi) = \Var\{\E(\bfy \mid \bfbeta, \bfpsi)\} + \E\{\Var(\bfy \mid \bfbeta, \bfpsi)\}\\
&=\Var(X\bfbeta_{1:T}) + \E(\sigma^2 I_T)
= X(\Lambda_1^{-1} \otimes D_\bftheta) X^\top + \sigma^2 I_T.
\end{align*}
Note that $p(\bfy \mid \bfpsi)$ is independent of the parametrization used for the latent states. The log-likelihood for $\bfy$ is
\begin{align*}
\ell = -\frac{T}{2} \log (2\pi) - \frac{1}{2} \log |\Sigma_\bfy| - \frac{1}{2} (\bfy - \bfmu_\bfy)^\top \Sigma_\bfy^{-1} (\bfy - \bfmu_\bfy).
\end{align*}
Let $\Omega=\Lambda_1 \otimes D_\bftheta^{-1} + X^\top X / \sigma^2$ denote a block tridiagonal matrix and $\bfc = X^\top (\bfy - \bfmu_\bfy)/ \sigma^2$.
Using the Woodbury matrix identity, we obtain
\begin{align*}
\Sigma_\bfy^{-1} &= \sigma^{-2} I_T - \frac{1}{\sigma^4} X \left(\Lambda_1 \otimes D_\bftheta^{-1} + \sigma^{-2} X^\top X \right)^{-1} X^\top,\\
\Sigma_\bfy^{-1}(\bfy - \bfmu_\bfy) &= \sigma^{-2}(\bfy - \bfmu_\bfy) - \sigma^{-2}X\Omega^{-1}\bfc.
\end{align*}
The term $\Omega^{-1}\bfc$ can be computed efficiently using Algorithm \ref{Alg: Meurant} in Section \ref{sec:Meurant}.
Moreover,
\begin{align*}
|\Sigma_\bfy| &= |\sigma^2 I_T + X(\Lambda_1^{-1}\otimes D_\bftheta)X^\top |\\
&= \sigma^{2T} |I_{Td} +  (\Lambda_1^{-1}\otimes D_\bftheta)X^\top X / \sigma^2)|\\
&= \sigma^{2T} |\Lambda_1^{-1} \otimes D_\bftheta| \cdot |\Lambda_1 \otimes D_\bftheta^{-1} + X^\top X/\sigma^2|\\
&=\sigma^{2T}|\Lambda_1^{-1}|^d |D_\bftheta|^T |\Omega|.
\end{align*}
As $|\Lambda_1| = 1/2$, $|\Lambda_1^{-1}|^d = 2^d$.
Furthermore, since $D_\bftheta$ is diagonal, $|D_\bftheta|^T = (\prod_{j=1}^d \theta_j)^T$. Hence,
\begin{align*}
\log |\Sigma_\bfy| &= T \log \sigma^2 + d\log 2 + T \sum_{j=1}^d \log \theta_j + \log |\Omega|,
\end{align*}
where $\log |\Omega|$ is computed efficiently using Algorithm \ref{Alg: Meurant}.

\subsection{Deriving updates in EM algorithm}
Based on the complete data log-likelihood, we derive $Q(\bfpsi \mid \bfpsi^{(i)})$ used in the E-step of the EM algorithm as follows,
\begin{align*}
Q(\bfpsi &\mid \bfpsi^{(i)})
= \E_{\wbfbeta \mid \bfy, \bfpsi^{(i)}} \{\log p(\bfy \mid \wbfbeta, \bfgamma, \bftheta, \sigma^2) + \log p(\wbfbeta \mid \bfgamma, \bftheta) \} \\
&= \E_{\wbfbeta \mid \bfy, \bfpsi^{(i)}} \left\{  -\frac{T}{2} \log(2\pi) -\frac{T}{2} \log \sigma^2 - \frac{1}{2\sigma^2} \sum_{t=1}^T  (y_t - \bfx_t^\top W_t \bfgamma - \bfx_t^\top D_\bftheta^{\bfa/2} \wbfbeta_t)^2 \right. \\
&\quad \left. -\frac{(T+1)d}{2}  \log(2\pi) - \frac{(T+1)}{2} \sum_{j=1}^d (1-a_j)\log \theta_j \right. \\
&\quad \left. - \frac{1}{2}\{ \wbfbeta - (I_{T+1}\otimes D_\bftheta^{-\bfa/2}) \overline{W} \bfgamma \}^\top (\Lambda \otimes D_\bftheta^{\bfa-\bfone}) \{ \wbfbeta - (I_{T+1}\otimes D_\bftheta^{-\bfa/2}) \overline{W} \bfgamma \} \right\} \\
&= - \frac{1}{2\sigma^2} \sum_{t=1}^T  \{(y_t - \bfx_t^\top W_t \bfgamma - \bfx_t^\top D_\bftheta^{\bfa/2} \bfm_{\wbfbeta_t, \bfa W}^{(i)})^2 + \bfx_t^\top D_\bftheta^{\bfa/2} V_{\wbfbeta_t, \bfa}^{(i)} D_\bftheta^{\bfa/2} \bfx_t \} -\frac{T}{2} \log \sigma^2\\
&\quad -\frac{T + (T+1)d}{2} \log(2\pi) - \frac{(T+1)}{2} \sum_{j=1}^d (1-a_j)\log \theta_j -\frac{1}{2} \tr\{(\Lambda \otimes D_\bftheta^{\bfa-\bfone})V_{\wbfbeta, \bfa}^{(i)} \}\\
&\quad  - \frac{1}{2}\{ \bfm_{\wbfbeta, \bfa W}^{(i)} - (I_{T+1}\otimes D_\bftheta^{-\bfa/2}) \overline{W} \bfgamma \}^\top (\Lambda \otimes D_\bftheta^{\bfa-\bfone}) \{ \bfm_{\wbfbeta, \bfa W}^{(i)} - (I_{T+1}\otimes D_\bftheta^{-\bfa/2}) \overline{W} \bfgamma \}.
\end{align*}

The derivatives of $Q(\bfpsi \mid \bfpsi^{(i)})$ are given below.
\begin{align*}
    \nabla_{\sigma^2} Q(\bfpsi \mid \bfpsi^{(i)}) &=
    \frac{1}{2\sigma^4}\sum_{t=1}^T  \{(y_t - \bfx_t^\top W_t \bfgamma - \bfx_t^\top D_\bftheta^{\bfa/2} \bfm_{\wbfbeta_t, \bfa W}^{(i)})^2 + \bfx_t^\top D_\bftheta^{\bfa/2} V_{\wbfbeta_t, \bfa}^{(i)} D_\bftheta^{\bfa/2} \bfx_t \} - \frac{T}{2\sigma^2}, \\
\nabla_\bfgamma Q(\bfpsi \mid \bfpsi^{(i)})
&= \sigma^{-2} W_{1:T}^\top X^\top \{\bfy - X W_{1:T}\bfgamma - X(I_T \otimes D_\bftheta^{\bfa / 2})\bfm_{\wbfbeta_{1:T}, \bfa W}^{(i)}\}\\
&\quad + \overline{W}^\top (\Lambda \otimes D_\bftheta^{\bfa / 2 - \bfone}) \{\bfm_{\wbfbeta, \bfa W}^{(i)} - (I_{T+1} \otimes D_\bftheta^{-\bfa/2})\overline{W}\bfgamma\}.
\end{align*}

As discussed in Section \ref{sec:EM}, the index set $\{1, \ldots, d\}$ is partitioned into two disjoint subsets, $S_{\NCP}$ and $S_\CP$, to derive tractable updates for $\bftheta$. Under this setting, $Q(\bfpsi \mid \bfpsi^{(i)})$ is separable in $\bftheta_\CP$ and $\bftheta_\NCP$, and we can write $Q(\bfpsi \mid \bfpsi^{(i)}) = Q_\NCP + Q_\CP + C$, where $Q_\CP$ and $Q_\NCP$ contain the terms that depend only on $\bftheta_\CP$ and $\bftheta_\NCP$ respectively, while $C$ contains terms that are independent of $\bftheta$. We have
\begin{align*}
    Q_\CP &= -\frac{1}{2} \sum_{j\in S_\CP}  (T+1) \log \theta_j
    + \frac{ \tr\Big[ \Lambda \Big\{ V_{\wbfbeta_{[j]}, \bfa}^{(i)} + (\bfm_{\wbfbeta_{[j]}, \bfa W}^{(i)} - \overline{W}_{[j]} \bfgamma)(\bfm_{\wbfbeta_{[j]}, \bfa W}^{(i)} - \overline{W}_{[j]} \bfgamma)^\top \Big\} \Big]}{\theta_j},
\end{align*}
and the derivative with respect to each $\theta_j$ is
\begin{align*}
    \nabla_{\theta_j} Q_\CP = - \frac{T+1}{2\theta_j} + \frac{\tr\Big[ \Lambda \Big \{ V_{\wbfbeta_{[j]}, \bfa}^{(i)} + (\bfm_{\wbfbeta_{[j]}, \bfa W}^{(i)} - \overline{W}_{[j]} \bfgamma)(\bfm_{\wbfbeta_{[j]}, \bfa W}^{(i)} - \overline{W}_{[j]} \bfgamma)^\top \Big\} \Big]}{2\theta_j^2}.
\end{align*}
On the other hand,
\begin{align*}
Q_\NCP
    &= -\frac{1}{2\sigma^2} \sum_{t=1}^T  \bigg\{
    (y_t - \bfx_{t}^\top W_t \bfgamma - \bfx_{t, \NCP}^\top D_{\bftheta_\NCP}^{1 / 2} \bfm_{\wbfbeta_{t, \NCP}, \bfa W}^{(i)} - \bfx_{t, \CP}^\top \bfm_{\wbfbeta_{t, \CP}, \bfa W}^{(i)})^2\\
    &\quad+ \bfx_{t, \NCP}^\top D_{\bftheta_\NCP}^{1/2} \Var(\wbfbeta_{t, \NCP}) D_{\bftheta_\NCP}^{1/2}\bfx_{t, \NCP}
    +2\bfx_{t, \NCP}^\top D_{\bftheta_\NCP}^{1/2} \Cov(\wbfbeta_{t,\NCP},\wbfbeta_{t,\CP}) \bfx_{t, \CP}
    \bigg\}\\
    &=-\frac{1}{2\sigma^2}\sum_{t=1}^T  \bigg\{
    (y_t - \bfx_{t}^\top W_t \bfgamma -
    \surd \bftheta_\NCP^\top D_{\bfx_t, \NCP}
    \bfm_{\wbfbeta_{t, \NCP}, \bfa W}^{(i)} - \bfx_{t, \CP}^\top \bfm_{\wbfbeta_{t, \CP}, \bfa W}^{(i)})^2\\
    &\quad+ \surd \bftheta_\NCP^\top D_{\bfx_t, \NCP}
    \Var(\wbfbeta_{t, \NCP}) D_{\bfx_t, \NCP} \surd \bftheta_\NCP
    +2 \surd \bftheta_\NCP^\top D_{\bfx_t, \NCP}
    \Cov(\wbfbeta_{t,\NCP},\wbfbeta_{t,\CP}) \bfx_{t, \CP}
    \bigg\}.
\end{align*}
Taking derivatives with respect to $\surd \bftheta_\NCP$, $\nabla_{\surd \bftheta_\NCP} Q(\bfpsi \mid \bfpsi^{(i)}) = - (A\surd \bftheta_\NCP - \bfb) / \sigma^2$, where
\begin{align*}
    A &= \sum_{t=1}^T  D_{\bfx_t, \NCP}  \Big( \Var(\wbfbeta_{t, \NCP}) + \bfm_{\wbfbeta_{t, \NCP}, \bfa W}^{(i)}
    \bfm_{\wbfbeta_{t, \NCP}, \bfa W}^{(i)\top} \Big) D_{\bfx_t, \NCP},\\
    \bfb &= \sum_{t=1}^T  D_{\bfx_t, \NCP} \Big\{\bfm_{\wbfbeta_{t, \NCP}, \bfa W}^{(i)}(y_t - \bfx_{t}^\top W_t \bfgamma - \bfx_{t, \CP}^\top \bfm_{\wbfbeta_{t, \CP}, \bfa W}^{(i)}) - \Cov(\wbfbeta_{t, \NCP}, \wbfbeta_{t, \CP})\bfx_{t, \CP}\Big\}.
\end{align*}
The EM updates are obtained by setting the respective derivatives to zero.

\subsection{Deriving optimal working parameters for location}
This section provides detailed derivations for Section \ref{sec:optimal_location}. Recall that $Z = \blockdiag(0_{d\times d}, X^\top X/\sigma^2)$. The augmented information matrix for $\bfgamma$ is
\begin{align*}
I_{\bfgamma, \bfgamma} = -\nabla^2_\bfgamma Q(\bfpsi \mid \bfpsi^{(i)})
&= \sigma^{-2} W_{1:T}^\top X^\top X W_{1:T} + \overline{W}^\top (\Lambda \otimes D_\bftheta^{-1})\overline{W}\\
&= W^\top Z W + \overline{W}^\top (\Lambda \otimes D_\bftheta^{-1})\overline{W}.
\end{align*}

\begin{proof}[Proof of Theorem 1]
The convergence rate of Algorithm 1, which alternately updates $\bfm_{\wbfbeta, \bfa W}$ and $\bfgamma$ given $(\bftheta, \sigma^2)$ is optimized by minimizing $I_{\bfgamma, \bfgamma}$ with respect to $W$. Let $A = \Lambda \otimes D_\bftheta^{-1}$ and $B = \bfone_{T+1}\otimes I_d$. We have $I_{\bfgamma, \bfgamma} = W^\top Z W + (B - W)^\top A (B - W)$,
\begin{align*}
d\vec(I_{\bfgamma, \bfgamma}) &= \vec\{ dW^\top ZW + W^\top ZdW - dW^\top A (B - W) - (B - W)^\top A dW \} \\
&= \vec[ dW^\top \{ (A + Z) W -  AB \} + \{ (A + Z) W -  AB \}^\top dW ], \\
\nabla_{\vec(W)} \vec(I_{\bfgamma, \bfgamma})&= [ I_d \otimes\{ (A + Z) W -  AB \} ] (K + I),
\end{align*}
where $K$ denotes the commutation matrix \citep{Magnus1999}. The above gradient is zero when $W = (A + Z)^{-1} AB$. Thus
\begin{align*}
W_{\opt} &= (\Lambda \otimes D_\bftheta^{-1} + Z)^{-1} (\Lambda \bfone_{T+1} \otimes D_\bftheta^{-1})
= V_{\wbfbeta, \bfzr} (D_\bftheta^{-1}, 0, \dots, 0)^\top,
\end{align*}
since $V_{\wbfbeta, \bfzr} = (\Lambda \otimes D_\bftheta^{-1} + Z)^{-1}$. The matrix rate of convergence of Algorithm 1 is $I - I_{\aug, \bfgamma, \bfgamma}^{-1} I_{\obs, \bfgamma, \bfgamma}$. Thus, to show that the convergence rate is minimized to zero when $W=W_\opt$, it suffices to show that $I_{\aug, \bfgamma,\bfgamma} = I_{\obs, \bfgamma, \bfgamma}$. We have $I_{\obs, \bfgamma, \bfgamma} = -\nabla_\bfgamma^2 \ell = X_s^\top \Sigma_\bfy^{-1}X_s$. When $W=W_\opt$,
\begin{align*}
I_{\aug, \bfgamma, \bfgamma} = I_{\bfgamma, \bfgamma}
&= W_\opt^\top Z W_\opt + (B- W_\opt)^\top A (B - W_\opt) \\
&= W_\opt^\top (Z + A) W_\opt - W_\opt^\top AB - B^\top A W_\opt + B^\top A B \\
&= B^\top \{A - A (A+Z)^{-1} A\} B \\
&= B^\top A (A+Z)^{-1} Z B \\
&= B^\top \{Z - Z (A+Z)^{-1} Z\} B \\
&= B^\top \{ \tilde{X}^\top \tilde{X} / \sigma^2 -  (\tilde{X}^\top \tilde{X} / \sigma^2) (A + \tilde{X}^\top \tilde{X} / \sigma^2)^{-1}  (\tilde{X}^\top \tilde{X} / \sigma^2) \} B \\
&= B^\top \tilde{X}^\top \{ I / \sigma^2 -  (\tilde{X} / \sigma^2) (A + \tilde{X}^\top \tilde{X} / \sigma^2)^{-1}  (\tilde{X}^\top / \sigma^2) \} \tilde{X} B \\
&= (\tilde{X}B)^\top (\sigma^2 I_T + \tilde{X} A^{-1} \tilde{X}^\top)^{-1}\tilde{X} B \\
&= X_s^\top(\sigma^2 I_T + X(\Lambda_1^{-1} \otimes D_\bftheta)X^\top)^{-1} X_s \\
&= X_s^\top \Sigma_\bfy^{-1} X_s = I_{\obs, \bfgamma, \bfgamma}.
\end{align*}
We write $Z = \tilde{X}^\top \tilde{X} / \sigma^2$ by defining $\tilde{X} = \begin{bmatrix} 0_{T\times d} & X \end{bmatrix}$ in line 6 and apply the Woodbury matrix identity in line 8. In line 9, we apply the results that $\tilde{X}B = X_s$ and $\tilde{X}A^{-1}\tilde{X}^\top = \tilde{X} (\Lambda^{-1} \otimes D_\bftheta) \tilde{X}^\top = X(\Lambda_1^{-1} \otimes D_\bftheta)X^\top$. Note that $I_{\aug, \bfgamma,\bfgamma}$ and $I_{\obs, \bfgamma, \bfgamma}$ are both independent of $\bfgamma$.
\end{proof}

\begin{proof}[Proof of Corollary 1]
We can write
\begin{align*}
    W_\opt &= \left\{ (I_{T+1} \otimes D_\bftheta^{-1})\big(\Lambda \otimes I_d + (I_{T+1} \otimes D_\bftheta)Z\big) \right\}^{-1} (\Lambda \bfone_{T+1} \otimes D_\bftheta^{-1})\\
    &= \left\{ \Lambda \otimes I_d + (I_{T+1} \otimes D_\bftheta)Z \right\}^{-1} (\Lambda \bfone_{T+1} \otimes I_d).
\end{align*}
For a fixed $T$, as $\theta_j / \sigma^2 \rightarrow 0$ for $j = 1, \ldots, d$, we have $(I_{T+1} \otimes D_\bftheta)Z \rightarrow 0$,
\begin{align*}
W_\opt & \rightarrow \bfone_{T+1} \otimes I_d,\quad \text{and} \quad
\overline{W}_\opt = (\bfone_{T+1} \otimes I_d) - W_\opt \rightarrow 0.
\end{align*}
\end{proof}

\begin{proof}[Proof of Corollary 2]
From Section \ref{sec:EM}, $\wbfbeta \mid \bfgamma, \bfy, \bftheta, \sigma^2 \sim \N(\bfm_{\wbfbeta, \bfa W}, V_{\wbfbeta, \bfa})$,
where
\begin{align*}
\bfm_{\wbfbeta, \bfa W} &= (I_{T+1} \otimes D_\bftheta^{-\bfa / 2})\bfm_{\wbfbeta, \bfzr W}, \quad
\bfm_{\wbfbeta, \bfzr W} = V_{\wbfbeta, \bfzr}\{(\Lambda \otimes D_\bftheta^{-1}) \overline{W} \bfgamma + \bfz\}.
\end{align*}
and  $\bfz = (\bfzr_d^\top, \{ X^\top (\bfy -  X W_{1:T} \bfgamma)/\sigma^2 \}^\top)^\top$. Writing $\overline{W} = (\bfone_{T+1} \times I_d) - W$ and $\bfz = \bfz_\bfy - ZW\bfgamma$,
where $\bfz_\bfy = (\bfzr_d^\top, \{X^\top\bfy/\sigma^2\}^\top)^\top$,
\begin{align*}
\bfm_{\wbfbeta, \bfzr W}
&= V_{\wbfbeta, \bfzr}\{(\Lambda \otimes D_\bftheta^{-1}) \overline{W} \bfgamma + \bfz_\bfy - ZW\bfgamma\}\\
&= V_{\wbfbeta, \bfzr}\{(\Lambda \bfone_{T+1} \otimes D_\bftheta^{-1})  - (\Lambda \otimes D_\bftheta^{-1} + Z)W\} \bfgamma + V_{\wbfbeta, \bfzr} \bfz_\bfy.
\end{align*}
When $W=W_\opt$, $\bfm_{\wbfbeta, \bfzr W} =V_{\wbfbeta, \bfzr}\bfz_\bfy$. Thus, both $\bfm_{\wbfbeta, \bfa W}$ and $V_{\wbfbeta, \bfa}$ do not depend on $\bfgamma$ and $\wbfbeta$ is conditionally independent of $\bfgamma$ given $(\bfy, \bftheta, \sigma^2)$.
\end{proof}

\subsection{Deriving optimal working parameters for scale}
We first derive the entries of the augmented information matrix, $I_{\bftheta, \bftheta}$. Note that
\begin{align*}
Q(\bfpsi &\mid \bfpsi^{(i)})
= -\frac{T+1}{2}\sum_{j=1}^d (1-a_j)\log\theta_j
    -\frac{1}{2\sigma^2} \sum_{t=1}^T \left(
    y_t - \bfx_t^\top W_t \bfgamma - \sum_{j=1}^d x_{tj}\theta_j^{a_j / 2} m_{\wbeta_{tj}, \bfa W}^{(i)}\right)^2\\
&- \frac{1}{2\sigma^2} \sum_{t=1}^T  \sum_{j=1}^d \sum_{k=1}^d \theta_j^{a_j/2} x_{tj} V_{\wbfbeta_t, \bfa, jk}^{(i)}x_{tk}\theta_k^{a_k/2}
    -\frac{1}{2}\sum_{j=1}^d \theta_j^{a_j-1}\{
    \tr(\Lambda V^{(i)}_{\wbfbeta_{[j]}, \bfa})
    + \bfm_{\wbfbeta_{[j]}, \bfa W}^{(i)\top}
    \Lambda \bfm_{\wbfbeta_{[j]}, \bfa W}^{(i)}\}\\
&-\frac{1}{2}\sum_{j=1}^d \theta_j^{-1} \bfgamma^\top\overline{W}_{[j]}^\top \Lambda \overline{W}_{[j]}\bfgamma + \sum_{j=1}^d \theta_j^{a_j / 2 - 1}\bfm_{\wbfbeta_{[j]}, \bfa W}^{(i)\top}\Lambda \overline{W}_{[j]}\bfgamma + C,
\end{align*}
where $C$ is a constant independent of $\theta$. Taking the derivative with respect to $\theta_j$,
\begin{align*}
    \nabla_{\theta_j}Q(\bfpsi &\mid \bfpsi^{(i)})
    = -\frac{(T+1)(1-a_j)}{2\theta_j}
    +\frac{a_j\theta_j^{a_j/2-1}}{2\sigma^2}\sum_{t=1}^T  x_{tj}m_{\wbeta_{tj}, \bfa W}^{(i)} \bigg(
    y_t - \bfx_t^\top W_t \bfgamma - \sum_{k=1}^d x_{tk}\theta_k^{a_k/2}m_{\wbeta_{tk}, \bfa W}^{(i)}
    \bigg)\\
    & -\frac{a_j}{2\sigma^2}\sum_{t=1}^T  \left(
    \theta_j^{a_j-1} x_{tj}^2 V_{\wbfbeta_t, \bfa, jj}^{(i)} + \sum_{k\neq j}\theta_j^{a_j/2-1}x_{tj}V_{\wbfbeta_t, \bfa, jk}^{(i)}x_{tk}\theta_k^{a_k/2}
    \right)
    +\frac{\theta_j^{-2}}{2}\bfgamma^\top\overline{W}_{[j]}^\top\Lambda\overline{W}_{[j]}\bfgamma\\
    &-\frac{(a_j-1)\theta_j^{a_j-2}}{2}\{
    \tr(\Lambda V^{(i)}_{\wbfbeta_{[j]}, \bfa})
    +\bfm_{\wbfbeta_{[j]}, \bfa W}^{(i)\top}\Lambda\bfm_{\wbfbeta_{[j]}, \bfa W}^{(i)}\}
    +(a_j/2-1)\theta_j^{a_j/2-2}\bfm_{\wbfbeta_{[j]}, \bfa W}^{(i)\top}\Lambda \overline{W}_{[j]}\bfgamma.
\end{align*}
Taking the second derivative with respect to $\theta_j$,
\begin{align*}
    \nabla_{\theta_j}^2 Q(\bfpsi &\mid \bfpsi^{(i)})
    =\frac{(T+1)(1-a_j)}{2\theta_j^2}
    - \frac{a_j^2\theta_j^{a_j-2}}{4\sigma^2}
    \sum_{t=1}^T  x_{tj}^2 m_{\wbeta_{tj}, \bfa W}^{(i)2}\\
    &+\frac{a_j(a_j/2-1)\theta_j^{a_j/2-2}}{2\sigma^2} \sum_{t=1}^T  x_{tj} m_{\wbeta_{tj}, \bfa W}^{(i)} \left(
    y_t - \bfx_t^\top W_t\bfgamma - \sum_{k=1}^d x_{tk}\theta_k^{a_k/2}m_{\wbeta_{tk}, \bfa W}^{(i)}
    \right)\\
    &-\frac{a_j}{2\sigma^2}\sum_{t=1}^T \left(
    (a_j-1)\theta_j^{a_j-2}x_{tj}^2 V_{\wbfbeta_t, \bfa, jj}^{(i)} + \sum_{k\neq j}(a_j/2-1)\theta_j^{a_j/2-2}x_{tj}V_{\wbfbeta_t, \bfa, jk}^{(i)}x_{tk}\theta_k^{a_k/2}
    \right)\\
    &-\frac{(a_j-1)(a_j-2)\theta_j^{a_j-3}}{2}\{
    \tr(\Lambda V_{\wbfbeta_{[j]}, \bfa}^{(i)}) + \bfm_{\wbfbeta_{[j]}, \bfa W}^{(i)\top}\Lambda\bfm_{\wbfbeta_{[j]}, \bfa W}^{(i)}\}\\
    &-\frac{1}{\theta_j^{3}}\bfgamma^\top \overline{W}_{[j]}^\top \Lambda \overline{W}_{[j]}\bfgamma + (a_j/2-1)(a_j/2-2)\theta_j^{a_j/2-3}\bfm_{\wbfbeta_{[j]}, \bfa W}^{(i)\top}\Lambda\overline{W}_{[j]}\bfgamma.
\end{align*}
The second derivative with respect to $\theta_j$ and $\theta_k$, for $j\neq k$, is
\begin{align*}
    \nabla_{\theta_j, \theta_k} Q(\bfpsi \mid \bfpsi^{(i)})
    =-\frac{a_ja_k\theta_j^{a_j/2-1}\theta_k^{a_k/2-1}}{4\sigma^2}
    \sum_{t=1}^T  x_{tj}\left(
    m_{\wbeta_{tj}, a W}^{(i)} m_{\wbeta_{tk}, a W}^{(i)} + V_{\wbfbeta_t, a, jk}^{(i)}\right)x_{tk}.
\end{align*}
Evaluating the derivatives at convergence (where $\bfpsi^{(i+1)} = \bfpsi^{(i)} = \bfpsi^*$, dropping the superscript $^*$ hereafter for brevity) and applying the identities in \eqref{eq: m & V}, which imply that $\theta_j^{a_j/2} m_{\wbeta_{tj}, \bfa W} = m_{\wbeta_{tj}, \bfzr W}$, $\theta_j^{a_j} V_{\wbfbeta_t, \bfa, jj}= V_{\wbfbeta_t, \bfzr, jj}$ and $\theta_j^{a_j/2} \theta_j^{a_k/2} V_{\wbfbeta_t, \bfa, jk}= V_{\wbfbeta_t, \bfzr, jk}$, we obtain
\begin{align} \label{eqgrad}
    \nabla_{\theta_j} & Q(\bfpsi \mid \bfpsi)=
    -\frac{(T+1)(1-a_j)}{2\theta_j}
    - \frac{a_j}{2\sigma^2\theta_j} \sum_{t=1}^T  \bigg\{
    x_{tj}^2(m_{\wbeta_{tj}, \bfzr W}^2 + V_{\wbfbeta_t, \bfzr, jj})\\
    &-x_{tj}m_{\wbeta_{tj}, \bfzr W}(y_t - \bfx_t^\top W_t\bfgamma) + \sum_{k\neq j}x_{tj}(V_{\wbfbeta_t, \bfzr, jk} + m_{\wbeta_{tj}, \bfzr W}m_{\wbeta_{tk}, \bfzr W})x_{tk}
    \bigg\} \nonumber \\
    &-\frac{(a_j-1)
    \{\bfm_{\wbfbeta_{[j]}, \bfzr W}^\top\Lambda \bfm_{\wbfbeta_{[j]}, \bfzr W} + \tr(\Lambda V_{\wbfbeta_{[j]}, \bfzr})\}
    -(a_j-2)\bfm_{\wbfbeta_{[j]}, \bfzr W}^\top \Lambda \overline{W}_{[j]}\bfgamma - \bfgamma^\top \overline{W}_{[j]}^\top\Lambda\overline{W}_{[j]}\bfgamma
    }{2\theta_j^2}, \nonumber \\
    I_{\theta_j\theta_j} &= - \nabla_{\theta_j}^2 Q(\bfpsi \mid \bfpsi) =
    \frac{a_j(a_j-1)}{2\sigma^2\theta_j^2} \sum_{t=1}^T  x_{tj}^2 (m_{\wbeta_{tj}, \bfzr W}^2 + V_{\wbfbeta_t, \bfzr, jj})
    - \frac{a_j(a_j/2-1)}{2\sigma^2\theta_j^2}   \label{eqIjj}  \\
    &\quad \times
    \sum_{t=1}^T  \bigg\{
    x_{tj}m_{\wbeta_{tj}, \bfzr W}(y_t - \bfx_t^\top W_t \bfgamma) - \sum_{k\neq j}x_{tj}(V_{\wbfbeta_t, \bfzr, jk} + m_{\wbeta_{tj}, \bfzr W}m_{\wbeta_{tk}, \bfzr W}) x_{tk}
    \bigg\} \nonumber \\
    &\quad-\frac{(T+1)(1-a_j)}{2\theta_j^2} + \frac{1}{2\theta_j^3}(a_j-1)(a_j-2)\{\bfm_{\wbfbeta_{[j]}, \bfzr W}^\top\Lambda\bfm_{\wbfbeta_{[j]}, \bfzr W} + \tr(\Lambda V_{\wbfbeta_{[j]}, \bfzr})\} \nonumber \\
    &\quad-\frac{1}{2\theta_j^3}(a_j-2)(a_j/2-2)\bfm_{\wbfbeta_{[j]}, \bfzr W}^\top \Lambda \overline{W}_{[j]}\bfgamma + \frac{1}{\theta_j^3}\bfgamma^\top\overline{W}_{[j]}^\top \Lambda \overline{W}_{[j]}\bfgamma, \nonumber \\
    I_{\theta_j\theta_k} &= -\nabla_{\theta_j, \theta_k} Q(\bfpsi \mid \bfpsi) = \frac{a_ja_k}{4\sigma^2\theta_j\theta_k}\sum_{t=1}^T  x_{tj}\bigg(m_{\wbeta_{tj}, 0 W}m_{\wbeta_{tk}, 0 W} + V_{\wbfbeta_t, 0, jk}\bigg)x_{tk}. \nonumber
\end{align}

To further simplify the expression of $I_{\theta_j\theta_j}$, we use the fact that $\nabla_{\theta_j} Q(\bfpsi \mid \bfpsi) = 0$ at convergence and
\begin{align}
\bfm_{\wbfbeta, \bfzr I} &= \frac{V_{\wbfbeta, 0}}{\sigma^2} \begin{bmatrix} \bfzr_d \\  X^\top \{\bfy -  X(\bfone_T \otimes \bfgamma)\} \end{bmatrix}
= \bfm_{\wbfbeta, 0W} - \overline{W} \bfgamma. \label{Identity2}
\end{align}
From \eqref{eqgrad}, plugging $a_j=0$ and setting $\nabla_{\theta_j} Q(\bfpsi \mid \bfpsi) = 0$ yields
\begin{align} \label{Identity7}
(T+1) \theta_j &= \bfm_{\wbfbeta_{[j]}, \bfzr W}^\top \Lambda  \bfm_{\wbfbeta_{[j]}, \bfzr W} + \tr(\Lambda V_{\wbfbeta_{[j]}, \bfzr}) - 2 \bfm_{\wbfbeta_{[j]}, \bfzr W}^\top \Lambda \overline{W}_{[j]} \bfgamma + \bfgamma^\top \overline{W}_{[j]}^\top \Lambda \overline{W}_{[j]} \bfgamma,
\nonumber \\
\text{or} \quad
(T+1) \theta_j &= \bfm_{\wbfbeta_{[j]}, \bfzr I}^\top \Lambda  \bfm_{\wbfbeta_{[j]}, \bfzr I} + \tr(\Lambda V_{\wbfbeta_{[j]}, \bfzr}).
\end{align}
On the other hand, plugging $a_j=1$ and setting $\nabla_{\theta_j} Q(\bfpsi \mid \bfpsi) = 0$ yields
\begin{align} \label{Identity8}
&\sigma^{-2} \sum_{t=1}^T \left\{ x_{tj}^2 (m_{\wbeta_{tj}, 0W}^2 + V_{\wbfbeta_t, 0, jj} ) - x_{tj} m_{\wbeta_{tj}, 0W} (y_t - \bfx_t^\top W_t \bfgamma) \right.
\nonumber \\
& \left.  +  \sum_{k\neq j}^d x_{tj} (V_{\wbfbeta_t, 0, jk} + m_{\wbeta_{tj}, 0W} m_{\wbeta_{tk}, 0W}) x_{tk} \right\}
=\bfgamma^\top \overline{W}_{[j]}^\top \frac{\Lambda}{\theta_j} \overline{W}_{[j]} \bfgamma - \bfm_{\wbfbeta_{[j]}, \bfzr W}^\top \frac{\Lambda}{\theta_j} \overline{W}_{[j]} \bfgamma.
\end{align}
Substituting the identities in \eqref{Identity7} and \eqref{Identity8} into \eqref{eqIjj}, we have
\begin{align*}
I_{\theta_j\theta_j}
&= \frac{a_j(aj-1)}{2\sigma^2\theta_j^2} \sum_{t=1}^T x_{tj}^2 (m_{\wbeta_{tj}, 0W}^{2} + V_{\wbfbeta_t, 0, jj})- \frac{a_j(a_j/2-1)}{2\sigma^2\theta_j^2} \sum_{t=1}^T x_{tj}^2 (m_{\wbeta_{tj}, 0W}^2 + V_{\wbfbeta_t, 0, jj} )
\\
& \quad - \frac{a_j(a_j/2-1)\{ \bfm_{\wbfbeta_{[j]}, \bfzr W}^\top \Lambda \overline{W}_{[j]} \bfgamma - \bfgamma^\top \overline{W}_{[j]}^\top \Lambda \overline{W}_{[j]} \bfgamma \}}{2\theta_j^3} - \frac{(T+1)(1-a_j)}{2\theta_j^2}
\\
& \quad +\frac{(a_j-1)(a_j-2)\{ (T+1) \theta_j + 2 \bfm_{\wbfbeta_{[j]}, \bfzr W}^\top \Lambda \overline{W}_{[j]} \bfgamma - \bfgamma^\top \overline{W}_{[j]}^\top \Lambda \overline{W}_{[j]} \bfgamma \}}{2 \theta_j^3}
\\
& \quad - \frac{(a_j-2)(a_j/2-2)\bfm_{\wbfbeta_{[j]}, \bfzr W}^\top \Lambda \overline{W}_{[j]} \bfgamma}{2\theta_j^3}
\quad + \frac{\bfgamma^\top \overline{W}_{[j]}^\top \Lambda \overline{W}_{[j]} \bfgamma}{\theta_j^3}
\\
&= \frac{1}{2\theta_j^2} \left( (T+1)(a_j-1)^2 + \frac{a_j^2}{2\sigma^2} \sum_{t=1}^T x_{tj}^2 (m_{\wbeta_{tj}, 0W}^{2} + V_{\wbfbeta_t, 0, jj})
\right.
\\
& \quad \left. + \frac{(a_j^2-2a_j)\bfm_{\wbfbeta_{[j]}, \bfzr W}^\top \Lambda \overline{W}_{[j]} \bfgamma}{\theta_j}
+ \frac{(2a_j-a_j^2/2) \bfgamma^\top \overline{W}_{[j]}^\top \Lambda \overline{W}_{[j]} \bfgamma}{\theta_j} \right)
\\
&= \frac{1}{2\theta_j^2} \left( (T+1)(a_j-1)^2 + \frac{a_j^2}{2\sigma^2} \sum_{t=1}^T x_{tj}^2 (m_{\wbeta_{tj}, 0W}^{2} + V_{\wbfbeta_t, 0, jj})
\right.
\\
& \quad \left. + \frac{(a_j^2 - 2a_j) \bfm_{\wbfbeta_{[j]}, \bfzr I}^\top \Lambda \overline{W}_{[j]} \bfgamma}{\theta_j} + \frac{a_j^2 \bfgamma^\top \overline{W}_{[j]}^\top \Lambda \overline{W}_{[j]} \bfgamma}{2\theta_j} \right).
\end{align*}
Recall that $V_{\wbfbeta, \bfzr}^{-1} = \Lambda \otimes D_\bftheta^{-1} + Z$
and $\bar{V}_{\wbfbeta_{[j]}, \bfzr}^{-1} \in \mathbbm{R}^{(T+1) \times (T+1)}$ extracts the $(j,j)$ element in each $d\times d$ block of $V_{\wbfbeta, 0}^{-1}$, so that $\bar{V}_{\wbfbeta_{[j]}, \bf0}^{-1} = \Lambda/\theta_j + \diag(0, x_{1j^2}, \dots, x_{Tj}^2)/\sigma^2$. Combining this with \eqref{Identity7},
\begin{align*}
& \sigma^{-2} \sum_{t =1}^T x_{tj}^2 (m_{\wbeta_{tj}, 0W}^{2} + V_{\wbfbeta_t, 0, jj})
=  \tr\{ (\bfm_{\wbfbeta_{[j]}, \bfzr W}\bfm_{\wbfbeta_{[j]}, \bfzr W}^\top  + V_{\wbfbeta_{[j]}, \bfzr}) (\bar{V}_{\wbfbeta_{[j]}, \bfzr}^{-1}  - \Lambda/\theta_j)\} \\
& = \bfm_{\wbfbeta_{[j]}, \bfzr W}^\top \bar{V}_{\wbfbeta_{[j]}, \bfzr}^{-1}  \bfm_{\wbfbeta_{[j]}, \bfzr W} + \tr(\bar{V}_{\wbfbeta_{[j]}, \bfzr}^{-1} V_{\wbfbeta_{[j]}, \bfzr} ) - (T+1) - \frac{2\bfm_{\wbfbeta_{[j]}, \bfzr W}^\top \Lambda \overline{W}_{[j]} \bfgamma}{\theta_j} + \frac{\bfgamma^\top \overline{W}_{[j]}^\top \Lambda \overline{W}_{[j]} \bfgamma}{\theta_j}
\\
& = \bfm_{\wbfbeta_{[j]}, \bfzr W}^\top \bar{V}_{\wbfbeta_{[j]}, \bfzr}^{-1}  \bfm_{\wbfbeta_{[j]}, \bfzr W} + \tr(\bar{V}_{\wbfbeta_{[j]}, \bfzr}^{-1} V_{\wbfbeta_{[j]}, \bfzr} ) - (T+1) - \frac{2\bfm_{\wbfbeta_{[j]}, \bfzr I}^\top \Lambda \overline{W}_{[j]} \bfgamma}{\theta_j} - \frac{\bfgamma^\top \overline{W}_{[j]}^\top \Lambda \overline{W}_{[j]} \bfgamma}{\theta_j}.
\end{align*}
Plugging this result into $I_{\theta_j\theta_j} $, we obtain
\begin{align*}
    I_{\theta_j, \theta_j} &=
    \frac{1}{2\theta_j^2} \bigg\{
    (T+1) \bigg(1 - 2 a_j + \frac{a_j^2}{2}\bigg)
    - \frac{2a_j \bfm_{\wbfbeta_{[j]}, \bfzr I}^\top \Lambda \overline{W}_{[j]}\bfgamma}{\theta_j}\\
    &\quad + \frac{a_j^2\bfm_{\wbfbeta_{[j]}, \bfzr W}^\top \bar{V}_{\wbfbeta_{[j]}, \bfzr}^{-1} \bfm_{\wbfbeta_{[j]}, \bfzr W}}{2}
    + \frac{a_j^2}{2} \tr(\bar{V}^{-1}_{\wbfbeta_{[j]}, \bfzr} V_{\wbfbeta_{[j]}, \bfzr})
    \bigg\}.
\end{align*}

\begin{proof}[Proof of Theorem 2]
First, we minimize $I_{\theta_j, \theta_j}$ with respect to $a_j$, conditional on $W$.
\begin{align*}
\frac{\partial{I_{\theta_j, \theta_j}}}{\partial a_j}
&= \frac{1}{2\theta_j^2} \bigg\{
    (T+1)(a_j - 2)
    - \frac{2 \bfm_{\wbfbeta_{[j]}, \bfzr I}^\top \Lambda \overline{W}_{[j]}\bfgamma}{\theta_j}\\
    &\quad + {a_j\bfm_{\wbfbeta_{[j]}, \bfzr W}^\top \bar{V}_{\wbfbeta_{[j]}, \bfzr}^{-1} \bfm_{\wbfbeta_{[j]}, \bfzr W}}
    + {a_j} \tr(\bar{V}^{-1}_{\wbfbeta_{[j]}, \bfzr} V_{\wbfbeta_{[j]}, \bfzr})
    \bigg\}.
\end{align*}
Setting the derivative to zero, we obtain
\begin{align*}
    a_{\opt, j} &=
    \frac{2(T+1) + 2 \bfm_{\wbfbeta_{[j]}, \bfzr I}^\top \Lambda \overline{W}_{[j]}\bfgamma/\theta_j}{T+1 + \bfm_{\wbfbeta_{[j]}, \bfzr W}^\top \bar{V}_{\wbfbeta_{[j]}, \bfzr}^{-1} \bfm_{\wbfbeta_{[j]}, \bfzr W} + \tr(\bar{V}^{-1}_{\wbfbeta_{[j]}, \bfzr} V_{\wbfbeta_{[j]}, \bfzr})}.
\end{align*}
The second derivative is
\begin{align*}
    \frac{\partial^2 I_{\theta_j, \theta_j}}{\partial a_j^2}  &=
    \frac{1}{2\theta_j^2} \bigg\{
    (T+1)
    + {\bfm_{\wbfbeta_{[j]}, \bfzr W}^\top \bar{V}_{\wbfbeta_{[j]}, \bfzr}^{-1} \bfm_{\wbfbeta_{[j]}, \bfzr W}}
    + \tr(\bar{V}^{-1}_{\wbfbeta_{[j]}, \bfzr} V_{\wbfbeta_{[j]}, \bfzr})
    \bigg\} > 0.
\end{align*}
Hence, $a_{\opt, j}$ is the unique minimizer. In the case $\overline{W}_{[j]} = 0$, $a^*_{\opt, j} = 2 \{1 + C_j/(T+1) \}^{-1}$.
\end{proof}

\section{Data URLs and plots}\label{supp:application_fig}
This section provides dataset URLs and visual plots for the empirical applications analyzed in Section \ref{sec:EM_Applications} and Section \ref{sec:VA_Applications}. The U.S. macroeconomic dataset is included as {\tt usmacro.update} in the R package {\tt bvarsv} \citep{Krueger2024}. The U.S. inflation data is available in the supplementary material for \cite{Fruhwirth2021} at \url{https://www.stat.rice.edu/~marina/BVSBOOK2022/SUPPLEMENT/supportmaterial.html}. Lastly, the Euro area inflation data can be accessed directly from the European Central Bank (ECB) Data Portal at \url{https://data.ecb.europa.eu/data/datasets/RTD}.

\begin{figure}[H]
\centering
\includegraphics[width=0.9\textwidth]{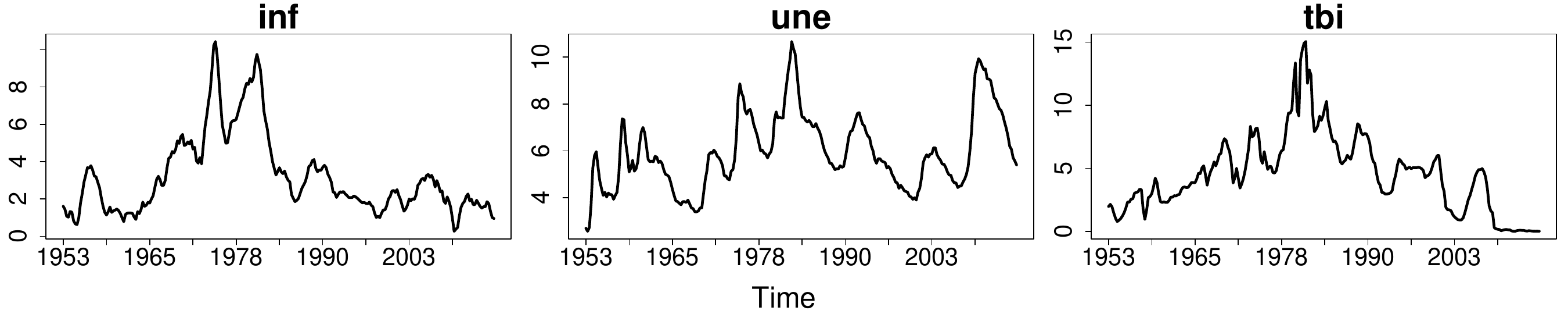}
\caption{U.S. macroeconomic data.}
\label{fig:usmacro_data}
\end{figure}

\begin{table}[H]
\centering
\small
\caption{Description of variables in U.S. inflation data.}
\label{tab:USinf_data}
\begin{tabular}{ll}
\hline
Symbol & Description \\
\hline
cons     & Real Personal Consumption Expenditures   \\
dom\_inv & Real Gross Private Domestic Investment   \\
dow      & Dow Jones Industrial Average             \\
emp      & Non-farm Payroll Employment              \\
exp      & Expected Changes in Inflation Rates      \\
gdp      & Real GDP                                 \\
hstarts  & Housing Starts                           \\
m1       & M1 Money Stock                           \\
napmpri  & NAPM Commodity Prices Index              \\
napmsdi  & NAPM Vendor Deliveries Index             \\
pmi      & ISM Manuf.: PMI Composite Index          \\
spread   & Spread 10-year T-Bond yield/3m T-Bill    \\
treas    & 3m Treasury Bill: Secondary Market       \\
unemp    & Unemployment rate                        \\
inf      & Consumer Price Index                     \\
\hline
\end{tabular}
\end{table}

\begin{figure}[H]
\centering
\includegraphics[width=0.72\textwidth]{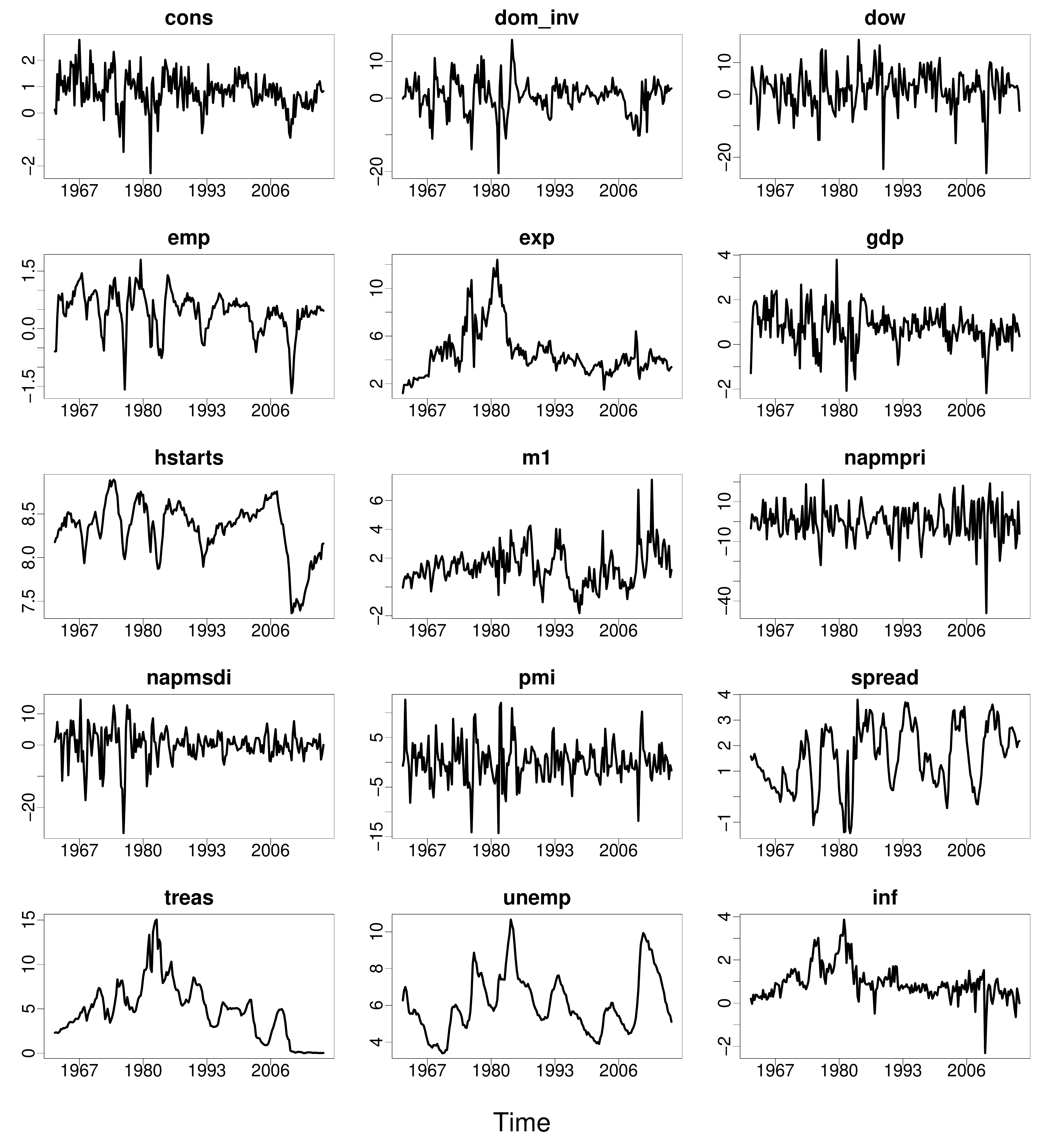}
\caption{U.S. inflation data.}
\label{fig:USinf_data}
\end{figure}

\begin{table}[H]
\centering
\small
\caption{Description of variables in ECB data.}
\label{tab:ECB_data}
\begin{tabular}{ll}
\hline
Symbol & Description \\
\hline
I-1MO   & One-month Euribor (Euro interbank offered rate) \\
I-1YR   & One-year Euribor (Euro interbank offered rate) \\
SENT    & Percentage change in economic sentiment indicator \\
STOCK-1 &
\makecell[l]{Percentage change in equity index\\
(Dow Jones, Euro Stoxx, Economic sector index financial)} \\
STOCK-2 & Percentage change in equity index -- Dow Jones Euro Stoxx 50 index \\
EXRATE  & \makecell[l]{Percentage change in ECB real effective exchange rate\\ (CPI deflated, broad group of currencies against euro)} \\
IP      & Percentage change in industrial production index \\
LOANS   & Percentage change in loans (total maturity, all currencies combined) \\
M3      & Annual percentage change in monetary aggregate M3 \\
CAR     & Registrations of new passenger cars \\
OIL     & Percentage change in oil price (Brent crude, 1-month forward) \\
ORDER   & Change in order-book levels \\
UNEMP   & Standardized unemployment rate (all ages, male \& female) \\
HICP    & Harmonized index of consumer prices \\
\hline
\end{tabular}
\end{table}

\begin{figure}[H]
\centering
\includegraphics[width=0.75\textwidth]{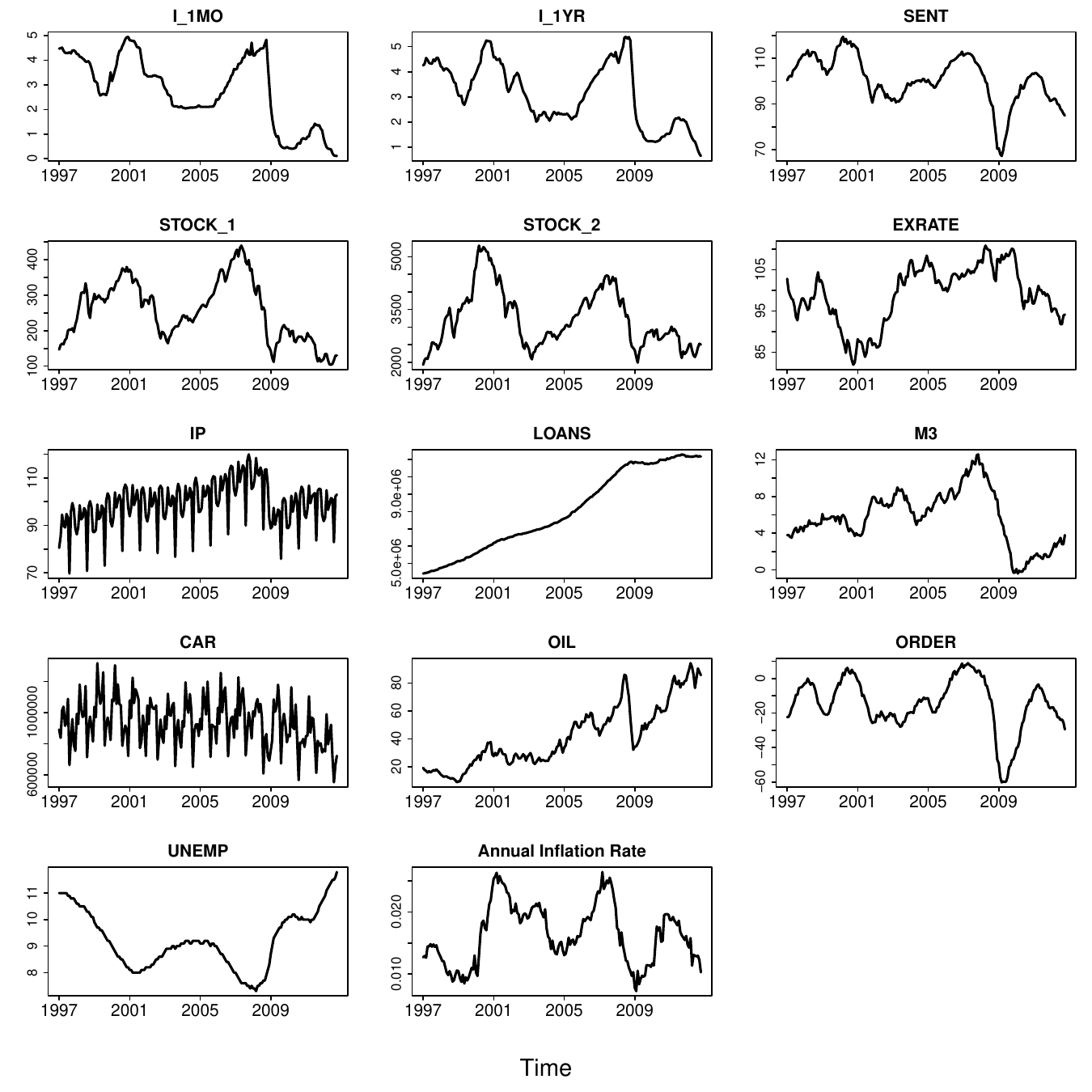}
\caption{ECB data. Original time series.}
\label{fig:ECB_data}
\end{figure}

\section{Detailed AECM algorithm results}\label{supp:EM_results}

This section provides detailed results for the applications presented in Section \ref{sec:EM_Applications}, reporting the estimated baseline coefficients $\bfgamma$, state variances $\bftheta$ and the  signal-to-noise ratios $\bftheta/\sigma^2$ and assigned parametrizations under the PNCP.

{
\small
\setlength{\tabcolsep}{2.5pt}

\begin{longtable}{crrrrrrrrrr}
    \caption{AECM estimates for simulated dataset and ground truth values.}
    \label{tab:simulation_EM_combined}\\

\hline
    \multirow{2}{*}{$j$}
    & \multicolumn{4}c{$\gamma_j$}
    & \multicolumn{4}c{$\theta_j$}
    & $\theta_j/\sigma^2$
    & $a_j$ \\
    \cline{2-5}
    \cline{6-9}
    \cline{10-10}
    \cline{11-11}
    & True & CP & NCP & PNCP
    & True & CP & NCP & PNCP
    & PNCP & PNCP \\
\hline
    \endfirsthead

\hline
    \multirow{2}{*}{$j$}
    & \multicolumn{4}c{$\gamma_j$}
    & \multicolumn{4}c{$\theta_j$}
    & $\theta_j/\sigma^2$
    & $a_j$ \\
    \cline{2-5}
    \cline{6-9}
    \cline{10-10}
    \cline{11-11}
    & Truth & CP & NCP & PNCP
    & Truth & CP & NCP & PNCP
    & PNCP & PNCP \\
\hline
    \endhead

\hline
    \endfoot

    1
    & $0$ & $-0.160549$ & $-0.129263$ & $-0.130216$
    & $0$ & $0.000360$ & $0.000000$ & $0.000000$
    & $0.000000$ & $1$ \\

    2
    & $-0.3$ & $-0.293657$ & $-0.362911$ & $-0.363042$
    & $0$ & $0.000302$ & $0.000000$ & $0.000000$
    & $0.000000$ & $1$ \\

    3
    & $1.5$ & $2.486690$ & $2.650430$ & $2.649950$
    & $0.01$ & $0.000750$ & $0.000000$ & $0.000000$
    & $0.000000$ & $1$ \\

    4
    & $0$ & $-0.236005$ & $-0.235895$ & $-0.242117$
    & $0.05$ & $0.084290$ & $0.083780$ & $0.081931$
    & $0.077778$ & $0$ \\

    5
    & $-0.3$ & $-0.919154$ & $-0.891327$ & $-0.900334$
    & $0.1$ & $0.108310$ & $0.107203$ & $0.101887$
    & $0.096722$ & $0$ \\

    6
    & $1.5$ & $1.655560$ & $1.639980$ & $1.649350$
    & $0.5$ & $0.632164$ & $0.522980$ & $0.633618$
    & $0.601499$ & $0$ \\

\end{longtable}
}

{
\small
\setlength{\tabcolsep}{3.5pt}

\begin{longtable}{lrrrrrrrr}
    \caption{AECM estimates for U.S. macro dataset.}
    \label{tab:macro_EM_combined}\\

\hline
    \multirow{2}{*}{Variable}
    & \multicolumn{3}c{$\gamma_j$}
    & \multicolumn{3}c{$\theta_j$}
    & $\theta_j/\sigma^2$
    & $a_j$ \\
    \cline{2-4}
    \cline{5-7}
    \cline{8-8}
    \cline{9-9}
    & CP & NCP & PNCP
    & CP & NCP & PNCP
    & PNCP & PNCP \\
\hline
    \endfirsthead

\hline
    \multirow{2}{*}{Variable}
    & \multicolumn{3}c{$\gamma_j$}
    & \multicolumn{3}c{$\theta_j$}
    & $\theta_j/\sigma^2$
    & $a_j$ \\
    \cline{2-4}
    \cline{5-7}
    \cline{8-8}
    \cline{9-9}
    & CP & NCP & PNCP
    & CP & NCP & PNCP
    & PNCP & PNCP \\
\hline
    \endhead

\hline
    \endfoot

    Intercept
    & $1.045750$ & $1.063220$ & $1.035250$
    & $0.007095$ & $0.006060$ & $0.006953$
    & $2.425420$ & $0$ \\

    {\tt inf\_lag}
    & $0.873011$ & $0.910856$ & $0.880890$
    & $0.004027$ & $0.004161$ & $0.004028$
    & $1.405020$ & $0$ \\

    {\tt une\_lag}
    & $-0.133204$ & $-0.138142$ & $-0.140987$
    & $0.000023$ & $0.000002$ & $0.000003$
    & $0.001002$ & $1$ \\

    {\tt tbi\_lag}
    & $0.027208$ & $0.027358$ & $0.026837$
    & $0.000016$ & $0.000000$ & $0.000000$
    & $0.000106$ & $1$ \\

\end{longtable}
}

{
\small
\setlength{\tabcolsep}{3.5pt}

\begin{longtable}{lrrrrrrrr}
    \caption{AECM estimates for U.S. inflation dataset.}
    \label{tab:usinf_EM_combined}\\

\hline
    \multirow{2}{*}{Variable}
    & \multicolumn{3}c{$\gamma_j$}
    & \multicolumn{3}c{$\theta_j$}
    & $\theta_j/\sigma^2$
    & $a_j$ \\
    \cline{2-4}
    \cline{5-7}
    \cline{8-8}
    \cline{9-9}
    & CP & NCP & PNCP
    & CP & NCP & PNCP
    & PNCP & PNCP \\
\hline
    \endfirsthead

\hline
    \multirow{2}{*}{Variable}
    & \multicolumn{3}c{$\gamma_j$}
    & \multicolumn{3}c{$\theta_j$}
    & $\theta_j/\sigma^2$
    & $a_j$ \\
    \cline{2-4}
    \cline{5-7}
    \cline{8-8}
    \cline{9-9}
    & CP & NCP & PNCP
    & CP & NCP & PNCP
    & PNCP & PNCP \\
\hline
    \endhead

\hline
    \endfoot

    Intercept
    & $1.111400$ & $1.129850$ & $1.129850$
    & $0.000020$ & $0.000000$ & $0.000000$
    & $0.000000$ & $1$ \\

    {\tt lag1}
    & $0.268465$ & $0.322039$ & $0.322039$
    & $0.000916$ & $0.000734$ & $0.000734$
    & $0.003103$ & $1$ \\

    {\tt lag2}
    & $-0.038085$ & $-0.035090$ & $-0.035090$
    & $0.000014$ & $0.000000$ & $0.000000$
    & $0.000000$ & $1$ \\

    {\tt lag3}
    & $0.315964$ & $0.328842$ & $0.328842$
    & $0.000526$ & $0.000339$ & $0.000339$
    & $0.001432$ & $1$ \\

    {\tt cons}
    & $0.096334$ & $0.123440$ & $0.123440$
    & $0.000009$ & $0.000000$ & $0.000000$
    & $0.000000$ & $1$ \\

    {\tt dom\_inv}
    & $0.021913$ & $-0.061366$ & $-0.061366$
    & $0.000104$ & $0.000000$ & $0.000000$
    & $0.000000$ & $1$ \\

    {\tt dow}
    & $-0.028668$ & $-0.007580$ & $-0.007580$
    & $0.000097$ & $0.000000$ & $0.000000$
    & $0.000000$ & $1$ \\

    {\tt emp}
    & $0.142537$ & $0.061060$ & $0.061060$
    & $0.000008$ & $0.000000$ & $0.000000$
    & $0.000000$ & $1$ \\

    {\tt exp}
    & $0.136886$ & $0.171804$ & $0.171804$
    & $0.000027$ & $0.000000$ & $0.000000$
    & $0.000000$ & $1$ \\

    {\tt gdp}
    & $-0.056348$ & $-0.046668$ & $-0.046668$
    & $0.000042$ & $0.000000$ & $0.000000$
    & $0.000000$ & $1$ \\

    {\tt hstarts}
    & $0.072232$ & $0.162069$ & $0.162069$
    & $0.000023$ & $0.000000$ & $0.000000$
    & $0.000000$ & $1$ \\

    {\tt m1}
    & $0.119115$ & $0.058541$ & $0.058541$
    & $0.000006$ & $0.000000$ & $0.000000$
    & $0.000000$ & $1$ \\

    {\tt napmpri}
    & $0.113764$ & $0.176976$ & $0.176976$
    & $0.000311$ & $0.000000$ & $0.000000$
    & $0.000000$ & $1$ \\

    {\tt napmsdi}
    & $-0.041599$ & $-0.059522$ & $-0.059522$
    & $0.000008$ & $0.000000$ & $0.000000$
    & $0.000000$ & $1$ \\

    {\tt pmi}
    & $0.085641$ & $0.071841$ & $0.071841$
    & $0.000061$ & $0.000000$ & $0.000000$
    & $0.000000$ & $1$ \\

    {\tt spread}
    & $-0.119524$ & $-0.146734$ & $-0.146734$
    & $0.000022$ & $0.000000$ & $0.000000$
    & $0.000000$ & $1$ \\

    {\tt treas}
    & $0.167627$ & $0.171760$ & $0.171760$
    & $0.004973$ & $0.000000$ & $0.000000$
    & $0.000000$ & $1$ \\

    {\tt unemp}
    & $0.019255$ & $0.022379$ & $0.022379$
    & $0.000046$ & $0.000451$ & $0.000451$
    & $0.001909$ & $1$ \\

\end{longtable}
}

{
\small
\setlength{\tabcolsep}{3.5pt}

\begin{longtable}{lrrrrrrrr}
    \caption{AECM estimates for ECB dataset.}
    \label{tab:ecb_EM_combined}\\

\hline
    \multirow{2}{*}{Variable}
    & \multicolumn{3}c{$\gamma_j$}
    & \multicolumn{3}c{$\theta_j$}
    & $\theta_j/\sigma^2$
    & $a_j$ \\
    \cline{2-4}
    \cline{5-7}
    \cline{8-8}
    \cline{9-9}
    & CP & NCP & PNCP
    & CP & NCP & PNCP
    & PNCP & PNCP \\
\hline
    \endfirsthead

\hline
    \multirow{2}{*}{Variable}
    & \multicolumn{3}c{$\gamma_j$}
    & \multicolumn{3}c{$\theta_j$}
    & $\theta_j/\sigma^2$
    & $a_j$ \\
    \cline{2-4}
    \cline{5-7}
    \cline{8-8}
    \cline{9-9}
    & CP & NCP & PNCP
    & CP & NCP & PNCP
    & PNCP & PNCP \\
\hline
    \endhead

\hline
    \endfoot

    Intercept
    & $2.724490$ & $3.027870$ & $2.607760$
    & $0.024176$ & $0.017841$ & $0.025338$
    & $67.173700$ & $0$ \\

    {\tt lag12}
    & $-0.102190$ & $-0.084281$ & $-0.082901$
    & $0.000147$ & $0.000143$ & $0.000132$
    & $0.350704$ & $0$ \\

    {\tt lag11}
    & $-0.265523$ & $-0.222045$ & $-0.219657$
    & $0.000000$ & $0.000000$ & $0.000000$
    & $0.000000$ & $1$ \\

    {\tt lag10}
    & $-0.216182$ & $-0.164990$ & $-0.169269$
    & $0.000001$ & $0.000000$ & $0.000000$
    & $0.000000$ & $1$ \\

    {\tt lag9}
    & $-0.160197$ & $-0.127455$ & $-0.136267$
    & $0.000001$ & $0.000000$ & $0.000000$
    & $0.000000$ & $1$ \\

    {\tt lag8}
    & $-0.289990$ & $-0.228750$ & $-0.227733$
    & $0.000209$ & $0.000000$ & $0.000000$
    & $0.000000$ & $1$ \\

    {\tt lag7}
    & $-0.075710$ & $-0.094025$ & $-0.086494$
    & $0.000050$ & $0.000000$ & $0.000000$
    & $0.000000$ & $1$ \\

    {\tt lag6}
    & $-0.225107$ & $-0.148067$ & $-0.149416$
    & $0.000001$ & $0.000000$ & $0.000000$
    & $0.000000$ & $1$ \\

    {\tt lag5}
    & $-0.282153$ & $-0.224895$ & $-0.224735$
    & $0.000000$ & $0.000000$ & $0.000000$
    & $0.000000$ & $1$ \\

    {\tt lag4}
    & $-0.305474$ & $-0.294914$ & $-0.285572$
    & $0.000269$ & $0.000000$ & $0.000000$
    & $0.000000$ & $1$ \\

    {\tt lag3}
    & $-0.408757$ & $-0.357777$ & $-0.347325$
    & $0.000001$ & $0.000000$ & $0.000000$
    & $0.000000$ & $1$ \\

    {\tt lag2}
    & $-0.344783$ & $-0.307471$ & $-0.293040$
    & $0.000001$ & $0.000000$ & $0.000000$
    & $0.000000$ & $1$ \\

    {\tt lag1}
    & $-0.415468$ & $-0.395706$ & $-0.395616$
    & $0.000001$ & $0.000000$ & $0.000000$
    & $0.000000$ & $1$ \\

    {\tt I\_1MO}
    & $0.841680$ & $0.756924$ & $0.655846$
    & $0.001938$ & $0.001276$ & $0.000747$
    & $1.981170$ & $0$ \\

    {\tt I\_1YR}
    & $-0.267106$ & $-0.245983$ & $-0.287491$
    & $0.000051$ & $0.000824$ & $0.000053$
    & $0.141298$ & $0$ \\

    {\tt SENT}
    & $-0.264064$ & $-0.279305$ & $-0.248460$
    & $0.001122$ & $0.000922$ & $0.000618$
    & $1.638110$ & $0$ \\

    {\tt STOCK\_1}
    & $0.278873$ & $0.296078$ & $0.320676$
    & $0.000015$ & $0.000053$ & $0.000016$
    & $0.043285$ & $0$ \\

    {\tt STOCK\_2}
    & $-0.229095$ & $-0.236794$ & $-0.252369$
    & $0.000028$ & $0.000149$ & $0.000043$
    & $0.115192$ & $0$ \\

    {\tt EXRATE}
    & $-0.131568$ & $-0.120023$ & $-0.086207$
    & $0.000009$ & $0.000010$ & $0.000007$
    & $0.019505$ & $1$ \\

    {\tt IP}
    & $0.043208$ & $0.049988$ & $0.067289$
    & $0.000006$ & $0.000346$ & $0.000344$
    & $0.912315$ & $0$ \\

    {\tt LOANS}
    & $0.227971$ & $0.459374$ & $0.005748$
    & $0.000349$ & $0.004473$ & $0.000333$
    & $0.883088$ & $0$ \\

    {\tt M3}
    & $0.210678$ & $0.004919$ & $0.247718$
    & $0.000108$ & $0.002326$ & $0.000046$
    & $0.122854$ & $0$ \\

    {\tt CAR}
    & $0.013186$ & $-0.007378$ & $-0.009523$
    & $0.000003$ & $0.000000$ & $0.000000$
    & $0.000000$ & $1$ \\

    {\tt OIL}
    & $0.020475$ & $-0.008633$ & $0.016253$
    & $0.000018$ & $0.000082$ & $0.000017$
    & $0.045762$ & $0$ \\

    {\tt ORDER}
    & $0.259846$ & $0.250409$ & $0.281922$
    & $0.000029$ & $0.000136$ & $0.000024$
    & $0.062623$ & $0$ \\

    {\tt UNEMP}
    & $-0.181185$ & $-0.182180$ & $-0.124326$
    & $0.000347$ & $0.001278$ & $0.005337$
    & $14.150000$ & $0$ \\

    {\tt D1}
    & $0.147137$ & $0.146278$ & $0.162030$
    & $0.000009$ & $0.000000$ & $0.000000$
    & $0.000000$ & $1$ \\

    {\tt D2}
    & $0.234107$ & $0.245383$ & $0.243795$
    & $0.000015$ & $0.000000$ & $0.000000$
    & $0.000000$ & $1$ \\

    {\tt D3}
    & $0.445267$ & $0.423498$ & $0.435982$
    & $0.000023$ & $0.000000$ & $0.000000$
    & $0.000000$ & $1$ \\

    {\tt D4}
    & $0.470875$ & $0.450710$ & $0.441397$
    & $0.000005$ & $0.000000$ & $0.000000$
    & $0.000000$ & $1$ \\

    {\tt D5}
    & $0.391402$ & $0.367915$ & $0.361342$
    & $0.000006$ & $0.000000$ & $0.000000$
    & $0.000000$ & $1$ \\

    {\tt D6}
    & $0.316357$ & $0.230548$ & $0.242291$
    & $0.000012$ & $0.000000$ & $0.000000$
    & $0.000000$ & $1$ \\

    {\tt D7}
    & $0.195520$ & $0.231519$ & $0.279034$
    & $0.000016$ & $0.000000$ & $0.000000$
    & $0.000000$ & $1$ \\

    {\tt D8}
    & $0.163478$ & $0.060654$ & $0.038685$
    & $0.000008$ & $0.000000$ & $0.000000$
    & $0.000000$ & $1$ \\

    {\tt D9}
    & $0.269734$ & $0.199587$ & $0.183860$
    & $0.000006$ & $0.000000$ & $0.000000$
    & $0.000000$ & $1$ \\

    {\tt D10}
    & $0.109603$ & $0.058128$ & $0.053706$
    & $0.000005$ & $0.000000$ & $0.000000$
    & $0.000000$ & $1$ \\

    {\tt D11}
    & $0.567765$ & $0.512190$ & $0.550122$
    & $0.000037$ & $0.000000$ & $0.000000$
    & $0.000000$ & $1$ \\

\end{longtable}
}

\section{Inverse of block tridiagonal matrix}\label{supp:Meurant}
This section provides additional derivations for Algorithm~\ref{Alg: Meurant}.
Following \citet{Meurant1992}, we consider the block LU, block UL and twisted factorizations of

$$
\Omega =
\begin{bmatrix}
D_1 & -A_2^\top   \\
-A_2 & D_2 & -A_3^\top & \\
& \ddots & \ddots & \ddots &\\
& & -A_{n-1} & D_{n-1} & -A_n^\top\\
& & & -A_n & D_n
\end{bmatrix}.
$$
The block LU factorization gives $|\Omega| = \prod_{i=1}^n |\Delta_i|$, which implies $\log|\Omega| = \sum_{i=1}^n \log|\Delta_i|$, leading to Step 2 of Algorithm~\ref{Alg: Meurant}.

For a given twist index $j\in\{2,\ldots,n-1\}$,
\begin{align*}
L= \begin{bmatrix}
0 & & &\\
-A_2&0&&\\
&\ddots&\ddots&\\
&&-A_n&0
\end{bmatrix}
\quad \text{and} \quad
\tilde L=
\begin{bmatrix}
0&&&&&&\\
-A_2&0&&&&&\\
&\ddots&\ddots&&&&\\
&&-A_j&0&-A_{j+1}^\top&&\\
&&&&\ddots&\ddots&\\
&&&&&0&-A_n^\top\\
&&&&&&0
\end{bmatrix}.
\end{align*}
Matching the diagonal blocks of the block LU, block UL and twisted factorizations with those of $\Omega$ yields
\begin{equation}\label{supp:Phi}
\Phi_i=
\begin{cases}
\Delta_i, & i=1,\ldots,j-1,\\[1mm]
D_j-A_j\Delta_{j-1}^{-1}A_j^\top
-A_{j+1}^\top\Sigma_{j+1}^{-1}A_{j+1} = \Delta_j+\Sigma_j-D_j,
& i=j,\\[1mm]
\Sigma_i, & i=j+1,\ldots,n.
\end{cases}
\end{equation}
The twisted factorization is particularly useful for deriving $\Omega^{-1}$. Although $\tilde{L}$ is neither lower nor upper block triangular, $(\Phi + \tilde{L})^{-1}$ can still be derived in closed form. Specifically, its $j$th block column is zero except for a single nonzero block, $\Phi_j^{-1}$, in the $j$th block row. Following \cite{Meurant1992}, let $X = (X_1^\top, \dots, X_n^\top)^\top$ denote the $j$th block column of $\Omega^{-1}$. Premultiplying $X$ with the twisted factorization of $\Omega$ yields
\begin{align*}
\Phi^{-1}(\Phi + \tilde{L})
\begin{bmatrix} X_1\\ \vdots\\ X_j\\ \vdots\\ X_n \end{bmatrix}
= (\Phi + \tilde{L})^{-1} \begin{bmatrix} 0\\ \vdots\\ I_d\\ \vdots\\ 0 \end{bmatrix}
= \begin{bmatrix} 0\\ \vdots\\ \Phi_{j}^{-1}\\ \vdots\\ 0 \end{bmatrix}.
\end{align*}
Solving this system leads to Theorem 3.1 of \cite{Meurant1992}:
\begin{align}\label{supp:Xj}
\begin{cases}
X_j = \Phi_j^{-1}, \\
X_{j-\ell} = \Delta_{j-\ell}^{-1} A_{j-\ell+1}^\top \Delta_{j - \ell+1}^{-1} \dots \Delta_{j-1}^{-1} A_j^\top \Phi_j^{-1}, &\ell = 1, \dots, j-1,
\end{cases}
\end{align}
where $2 \leq j \leq n-1$. Combining \eqref{supp:Phi} and \eqref{supp:Xj} gives the diagonal blocks of $\Omega^{-1}$:
\begin{equation}
\Omega_{j,j}^{-1}
= \begin{cases}
\Sigma_1^{-1}, & j=1,\\[1mm]
(\Delta_j+\Sigma_j-D_j)^{-1}, & 2\leq j\leq n-1,\\[1mm]
\Delta_n^{-1}, & j=n.
\end{cases}
\label{supp:Omega_diag}
\end{equation}
Furthermore, setting $\ell=1$ in \eqref{supp:Xj} yields the superdiagonal blocks:
\begin{equation}
\Omega_{j-1,j}^{-1}
= \Delta_{j-1}^{-1}A_j^\top\Omega_{j,j}^{-1},
\qquad j=2,\ldots,n,
\label{supp:Omega_offdiag}
\end{equation}
This completes the derivation of the diagonal and superdiagonal blocks of $\Omega^{-1}$ as required in Step 4 of Algorithm \ref{Alg: Meurant}.

Theorem 3.4 of \cite{Meurant1992} provides an alternative characterization of $\Omega^{-1}$ by expressing its $(i,j)$th block as $U_i V_j^\top$ for $j \geq i$, generated via the recursions,
\begin{align}\label{UV result}
\begin{cases}
U_1 = I_d, V_1 = \Sigma_1^{-1}, \\
U_i = A_i^{-\top}\Delta_{i-1} U_{i-1} &\text{for } i=2, \dots, n,\\
V_j = \Sigma_j^{-1} A_j V_{j-1}  &\text{for }j=1,\dots, n.
\end{cases}
\end{align}
Here, uniqueness is imposed by setting $U_1 = I_d$, which implies that $V_1^\top$, as the first diagonal block of $\Omega^{-1}$, is symmetric. The formulation in \eqref{UV result} can also be derived by premultiplying the first block column of $\Omega^{-1}$ by the block UL factorization of $\Omega$, and the last block column of $\Omega^{-1}$ by the block LU factorization of $\Omega$. In particular, a backward recursion of $\{U_i\}$ is given by
\begin{align} \label{U result}
U_n = \Delta_n^{-1} V_n^{-\top}
\quad \text{and} \quad
U_i = \Delta_i^{-1} A_{i+1}^\top U_{i+1}
\quad \text{for} \quad i=n-1, \dots, 1.
\end{align}
The recursions in \eqref{UV result} is used to compute the first block column of $\Omega^{-1}$ in Step 3 of Algorithm \ref{Alg: Meurant}. We initially considered using \eqref{UV result} to compute the diagonal and superdiagonal blocks of $\Omega^{-1}$. However, for large $n$, this approach suffers from numerical instability, as $\{U_i\}$ grows rapidly in magnitude, whereas $\{V_j\}$ vanishes toward zero. By contrast, the formulations in \eqref{supp:Omega_diag} and \eqref{supp:Omega_offdiag} exhibits superior numerical stability because they avoid repeated matrix multiplications over long sequences.

\section{Definitions of other shrinkage priors}\label{supp:table_shrinkage}

\begin{table}[H]
\centering
\caption{Definitions of alternative shrinkage priors compared with the triple gamma prior in Figure \ref{fig:cases_triple_gamma}}\label{tab: shrinkage priors}
\begin{tabular}{l l c c c c}
\hline
\textbf{Prior for $\surd{\theta_j}$} & \textbf{Representation} & $\bm{a^\xi}$ & $\bm{c^\xi}$ & $\bm{\kappa_B}$ & $\bm{\phi^\xi}$ \\
\hline
\rule[-2.5ex]{0pt}{6ex}
$\N(0,\tau^2\psi_j^2),\;
 \psi_j \sim t_1$
& Horseshoe
& $\dfrac{1}{2}$
& $\dfrac{1}{2}$
& $\dfrac{1}{\tau^2}$
& $\tau^2$
\\
\rule[-2.5ex]{0pt}{6ex}
$\N(0,\tau^2\tilde{\xi}_j),\;
 \tilde{\xi}_j \sim \G(a^\xi,a^\xi)$
& double gamma
& $a^\xi$
& $\infty$
& $\dfrac{2}{\tau^2}$
& --
\\
\rule[-2.5ex]{0pt}{6ex}
$\N(0,\tau^2\tilde{\xi}_j),\;
 \tilde{\xi}_j \sim \text{Exp}(1)$
& Lasso
& $1$
& $\infty$
& $\dfrac{2}{\tau^2}$
& --
\\
\hline
\end{tabular}
\label{tab:cases_triple_gamma}
\end{table}

\section{Log joint density of TVP model}\label{supp:joint_density_terms}
This section provides explicit forms of the terms constituting
$\log p(\bfy, \Theta)$ from Section \ref{sec:priors}.
\begin{itemize}
\item
$\log p(\bfy \mid \wbfbeta, \bfgamma, \bftheta, \sigma^2) = - [T \log(2\pi) + T \log \sigma^2 + \{\bfy -  X W_{1:T} \bfgamma - X (I_T \otimes D_\bftheta^{\bfa/2}) \wbfbeta_{1:T} \}^\top \{\bfy -  X W_{1:T} \bfgamma - X (I_T \otimes D_\bftheta^{\bfa/2}) \wbfbeta_{1:T} \} / \sigma^2]/2$.
\item
$\log p(\wbfbeta \mid \bfgamma, \bftheta) = - [(T+1)d\log(2\pi) + (T+1) \sum_{j=1}^d (1-a_j)\log \theta_j + \{ \wbfbeta - (I_{T+1}\otimes D_\bftheta^{-\bfa/2}) \overline{W} \bfgamma \}^\top (\Lambda \otimes D_\bftheta^{\bfa-\bfone}) \{ \wbfbeta - (I_{T+1}\otimes D_\bftheta^{-\bfa/2}) \overline{W} \bfgamma \} ]/2$.
\item
$\log p(\sigma^2 \mid C_0) = c_0 \log C_0 - \log \Gamma(c_0) - (c_0 +1) \log \sigma^2 - C_0/\sigma^2$.
\item
$\log p(C_0) = g_0 \log G_0 - \log \Gamma(g_0) + (g_0-1) \log (C_0) - G_0 C_0$.
\item
$\log p(\theta_j \mid \xi_j, \kappa_B) = \frac{1}{2}\log \{ \kappa_B/(4\xi_j) \} - \log \Gamma(1/2) - \frac{1}{2}\log \theta_j - \kappa_B\theta_j / (4\xi_j)$, for D-CP/NCP where $j \in S_\CP$.
\item
$\log p(\surd{\theta_j} \mid \xi_j, \kappa_B) = - \{ \log (2\pi) + \log (2\xi_j / \kappa_B) + \theta_j / (2 \xi_j / \kappa_B) \} / 2$, for D-CP/NCP where $j \in S_\NCP$.
\item
$\log p(\vartheta_j \mid \xi_j, \kappa_B) = \frac{1}{2}\log \{ \kappa_B / (4\xi_j) \} - \log \Gamma(1/2) + \vartheta_j / 2 - \kappa_B\exp(\vartheta_j) / (4\xi_j)$, for CLBS where $\vartheta_j = \log \theta_j$.
\item
$\log p(\xi_j \mid \kappa_j) = c^\xi \log \kappa_j - \log \Gamma(c^\xi)- (c^\xi + 1) \log \xi_j - \kappa_j /\xi_j$.
\item
$\log p(\kappa_j) = a^\xi \log(a^\xi /c^\xi) - \log \Gamma(a^\xi) + (a^\xi - 1) \log \kappa_j - a^\xi \kappa_j / c^\xi$.
\item
$\log p(\gamma_j \mid \tau_j, \lambda_B) =  -\{ \log(2\pi) + \log (2\tau_j / \lambda_B) + \lambda_B \gamma_j^2 /(2\tau_j) \}/2$.
\item
$\log p(\tau_j \mid \lambda_j) = c^\tau \log\lambda_j - \log \Gamma(c^\tau) - (c^\tau + 1) \log \tau_j - \lambda_j/\tau_j$.
\item
$\log p(\lambda_j) = a^\tau \log(a^\tau /c^\tau) - \log \Gamma(a^\tau) + (a^\tau - 1) \log \lambda_j - a^\tau \lambda_j / c^\tau$.
\item
$\log p(\kappa_B \mid d_2) = a^\xi \log d_2 - \log \Gamma(a^\xi) + (a^\xi -1) \log \kappa_B - d_2 \kappa_B$.
\item
$\log p(d_2) = c^\xi \log(2c^\xi/a^\xi) - \log \Gamma(c^\xi) + (c^\xi - 1) \log d_2 - 2c^\xi d_2 / a^\xi$.
\item
$\log p(\lambda_B \mid e_2) = a^\tau \log e_2 - \log \Gamma(a^\tau) + (a^\tau -1) \log \lambda_B - e_2 \lambda_B$.
\item
$\log p(e_2) = c^\tau \log(2c^\tau/a^\tau) - \log \Gamma(c^\tau) + (c^\tau - 1) \log e_2 - 2c^\tau e_2 / a^\tau$.
\end{itemize}

\section{Derivation of optimal densities in MFVB}\label{supp:derive_opt_dens}
This section derives the optimal variational densities under MFVB using \eqref{CAVI_update}, based on the independence assumptions established in Section \ref{sec:opt dens}.

\begin{itemize}[leftmargin=*]
\item \leavevmode \vspace{-1.275\baselineskip}
\begin{align*}
q(\wbfbeta) & \propto \exp \E_{-\wbfbeta} \left\{\log p(\bfy \mid \wbfbeta, \bfgamma, \bftheta, \sigma^2) + \log p(\wbfbeta \mid \bfgamma, \bftheta) \right\}
\\
& \propto \exp \E_{-\wbfbeta} \bigg\{ -\frac{1}{2} \bigg[ \wbfbeta^\top (\Lambda \otimes D_\bftheta^{\bfa-\bfone}) \wbfbeta - 2 \wbfbeta^\top (\Lambda \otimes D_\bftheta^{\bfa/2-\bfone}) \overline{W} \bfgamma
\\
&\quad  + \frac{\wbfbeta_{1:T}^\top (I_T \otimes D_\bftheta^{\bfa/2}) X^\top X (I_T \otimes D_\bftheta^{\bfa/2}) \wbfbeta_{1:T} - 2 \wbfbeta_{1:T}^\top (I_T \otimes D_\bftheta^{\bfa/2}) X^\top (\bfy - X W_{1:T} \bfgamma)}{\sigma^2} \bigg] \bigg\}.
\end{align*}
Thus the optimal $q(\wbfbeta)$ is $\N(\bfmu_\wbfbeta^q, \Sigma_\wbfbeta^q)$, where
\begin{align*}
\Sigma_{\wbfbeta}^q &= \left( \Lambda \otimes \E_q (D_\bftheta^{\bfa-\bfone}) +  \E_q(\sigma^{-2}) \begin{bmatrix}
\bfzr_{d \times d} & \bfzr_{d \times Td} \\
\bfzr_{Td \times d} &  \E_q\{ (I_T \otimes D_\bftheta^{\bfa/2}) X^\top X (I_T \otimes D_\bftheta^{\bfa/2}) \}
\end{bmatrix} \right)^{-1}
\\
\bfmu_{\wbfbeta}^q &= \Sigma_{\wbfbeta}^q \left( \{\Lambda \otimes \E_q(D_\bftheta^{\bfa/2-\bfone})\} \overline{W} \bfmu_\bfgamma^q + \E_q(\sigma^{-2}) \begin{bmatrix}
\bfzr_d \\   \{I_T \otimes \E_q(D_\bftheta^{\bfa/2})\} X^\top (\bfy -  X W_{1:T} \bfmu_\bfgamma^q) \end{bmatrix} \right).
\end{align*}
In the special case where $\bfa=\bfzr$, we have
\begin{align*}
\Sigma_{\wbfbeta}^q &= \left( \Lambda \otimes \E_q (D_\bftheta^{-1}) + \E_q(\sigma^{-2}) \begin{bmatrix}
\bfzr_{d \times d} & \bfzr_{d \times Td} \\
\bfzr_{Td \times d} &  X^\top X
\end{bmatrix} \right)^{-1} \\
\bfmu_{\wbfbeta}^q &= \Sigma_{\wbfbeta}^q \left( \{\Lambda \otimes \E_q(D_\bftheta^{-1})\} \overline{W} \bfmu_\bfgamma^q + \E_q(\sigma^{-2}) \begin{bmatrix}
\bfzr_d \\  X^\top (\bfy -  X W_{1:T} \bfmu_\bfgamma^q) \end{bmatrix} \right).
\end{align*}
These expressions correspond to $V_{\wbfbeta,\bfzr}$ and $\bfm_{\wbfbeta, \bfzr W}$, with $D_\bftheta^{-1}$, $\sigma^{-2}$ and $\bfgamma$ replaced by their respective expectations under the variational distribution $q$.

\item \leavevmode \vspace{-1.275\baselineskip}
\begin{align*}
q(\bfgamma) &\propto \exp \E_{-\bfgamma} \bigg\{ \log p(\bfy \mid \wbfbeta, \bfgamma, \bftheta, \sigma^2) + \log p(\wbfbeta \mid \bfgamma, \bftheta) + \sum_{j=1}^d \log p(\gamma_j \mid \tau_j, \lambda_B) \bigg\}
\\
& \propto \exp \E_{-\bfgamma} \bigg\{ -\frac{1}{2} \bigg[ \frac{\bfgamma^\top (\sum_{t=1}^T  W_t^\top \bfx_t \bfx_t^\top W_t) \bfgamma - 2 \bfgamma^\top \sum_{t=1}^T  (y_t - \bfx_t^\top D_\bftheta^{\bfa/2} \wbfbeta_t) W_t^\top\bfx_t  }{\sigma^2}
\\
& \quad + \bfgamma^\top \overline{W}^\top (\Lambda \otimes D_\bftheta^{-1})\overline{W} \bfgamma - 2 \bfgamma^\top \overline{W}^\top (\Lambda \otimes D_\bftheta^{\bfa/2-\bfone}) \wbfbeta
+ \lambda_B \bfgamma^\top \diag(\bftau^{-1}) \bfgamma/2 \bigg] \bigg\}.
\end{align*}
Thus the optimal $q(\bfgamma)$ is $\N(\bfmu_\bfgamma^q, \Sigma_\bfgamma^q)$, where
\begin{align*}
\Sigma_\bfgamma^q &= \bigg[ \E_q( \sigma^{-2} )\sum_{t=1}^T  W_t^\top \bfx_t \bfx_t^\top W_t + \overline{W}^\top \{\Lambda \otimes \E_q(D_\bftheta^{-1})\} \overline{W} + \frac{\E_q(\lambda_B)}{2} \diag\{ \E_q ( \bftau^{-1} )\} \bigg] ^{-1}
\\
\bfmu_\bfgamma^q &= \Sigma_\bfgamma^q \bigg[ \E_q( \sigma^{-2} ) \sum_{t=1}^T  \{y_t - \bfx_t^\top \E_q (D_\bftheta^{\bfa/2}) \mu_{\wbfbeta_t}^q\} W_t^\top\bfx_t + \overline{W}^\top \{ \Lambda \otimes \E_q(D_\bftheta^{\bfa/2-\bfone}) \} \mu_{\wbfbeta}^q \bigg].
\end{align*}

\item  \leavevmode \vspace{-1.275\baselineskip}
\begin{align*}
q(\sigma^2) &\propto \exp \E_{-\sigma^2} \left\{
\log p(\bfy \mid \wbfbeta, \bfgamma, \bftheta, \sigma^2) + \log p(\sigma^2 \mid C_0) \right\} \\
&\propto \exp \E_{-\sigma^2} \bigg\{
- \left(c_0 + \tfrac{T}{2} +1 \right) \log \sigma^2 - \frac{C_0 + \sum_{t=1}^T  (y_t - \bfx_t^\top W_t \bfgamma - \bfx_t^\top D_\bftheta^{\bfa/2} \wbfbeta_t)^2/2 }{\sigma^2}
\bigg\}.
\end{align*}
Thus the optimal $q(\sigma^2)$ is $\IG(c^q, C^q)$, where $c^q = c_0 + T/2$,
\begin{align*}
C^q &= \E_q(C_0) + \frac{1}{2}\sum_{t=1}^T  \E_q (y_t - \bfx_t^\top W_t \bfgamma - \bfx_t^\top D_\bftheta^{\bfa/2} \wbfbeta_t)^2
\\
&= \frac{g^q}{G^q} + \frac{1}{2}\sum_{t=1}^T  \Big[ (y_t - \bfx_t^\top W_t \bfmu_\bfgamma^q)^2 - 2(y_t - \bfx_t^\top W_t \bfmu_\bfgamma^q) \bfx_t^\top \E_q(D_\bftheta^{\bfa/2}) \bfmu_{\wbfbeta_t}^q + \bfx_t^\top W_t \Sigma_\bfgamma^q W_t^\top \bfx_t
\\
&\quad + \tr\{(\bfmu_{\wbfbeta_t}^q\bfmu_{\wbfbeta_t}^{q\top} + \Sigma_{\wbfbeta_t}^q) D_{\bfx_t} \E_q(\bftheta^{\bfa/2}\bftheta^{\bfa/2 \top})D_{\bfx_t} \} \Big].
\end{align*}

\item \leavevmode \vspace{-1.275\baselineskip}
\begin{align*}
q(\bfxi, \bftau) \propto \exp \sum_{j=1}^d \bigg\{& - \left(c^\xi + \frac{1}{2} + 1\right) \log \xi_j - \frac{\E_q(\kappa_j + \kappa_B \theta_j/4)}{\xi_j}\\
& - \left(c^\tau + \frac{1}{2} + 1 \right) \log \tau_j - \frac{\E_q(\lambda_j + \lambda_B \gamma_j^2/4)}{\tau_j} \bigg\}.
\end{align*}
Thus the optimal $q( \bfxi, \bftau) = \prod_{j=1}^{d} q(\xi_j)q(\tau_j)$, where $q(\xi_j)$ is $\IG(a_\xi^q, b_{\xi_j}^q)$ and $q(\tau_j)$ is $\IG(a_\tau^q, b_{\tau_j}^q)$. The variational parameters are given by
\begin{equation*}
\begin{aligned}
a_\xi^q =&c^\xi + 1/2, \\
a_\tau^q =&c^\tau + 1/2,
\end{aligned}
\qquad
\begin{aligned}
b_{\xi_j}^q &= \E_q(\kappa_j) + \E_q(\kappa_B) \E_q(\theta_j) / 4
= a_\kappa^q / b_{\kappa_j}^q + a_{\kappa_B}^q\E_q(\theta_j) / (4b_{\kappa_B}^q), \\
b_{\tau_j}^q &= \E_q(\lambda_j) + \E_q(\lambda_B) \E_q(\gamma_j^2) / 4
= a_\lambda^q / b_{\lambda_j}^q + a_{\lambda_B}^q(\mu_{\gamma_j}^{q2} + v_{\gamma_j}^q) / (4b_{\lambda_B}^q).
\end{aligned}
\end{equation*}

\item \leavevmode \vspace{-1.275\baselineskip}
\begin{equation*}
\begin{aligned}
q(\bfkappa, \bflambda) & \propto \exp \sum_{j=1}^d \Big\{(a^\xi + c^\xi - 1)\log \kappa_j - \left(\E_q\left(\xi_j^{-1}\right) + a^\xi / c^\xi\right) \kappa_j \\
& \qquad \qquad + (a^\tau + c^\tau - 1) \log \lambda_j - \left(\E_q\left(\tau_j^{-1}\right) + a^\tau/c^\tau\right) \lambda_j \Big\}.
\end{aligned}
\end{equation*}
Thus the optimal $q(\bfkappa, \bflambda) = \prod_{j=1}^d q(\kappa_j)q(\lambda_j)$, where $q(\kappa_j)$ is $\G(a_\kappa^q, b_{\kappa_j}^q)$ and $q(\lambda_j)$ is $\G(a_\lambda^q, b_{\lambda_j}^q)$. The variational parameters are given by
\begin{equation*}
\begin{aligned}
a_\kappa^q  & = a^\xi + c^\xi,  \\
a_\lambda^q & = a^\tau + c^\tau,
\end{aligned}
\qquad
\begin{aligned}
b_{\kappa_j}^q & = \E_q(\xi_j^{-1}) + a^\xi/c^\xi
= a_\xi^q / b_{\xi_j}^q + a^\xi / c^\xi,      \\
b_{\lambda_j}^q & = \E_q(\tau_j^{-1})+ a^\tau/c^\tau =
a_\tau^q / b_{\tau_j}^q + a^\tau/c^\tau.
\end{aligned}
\end{equation*}

\item \leavevmode \vspace{-1.275\baselineskip}
\begin{align*}
q(C_0, \kappa_B, \lambda_B) &\propto \exp \bigg[
(c_0 + g_0 - 1) \log (C_0) - \{G_0 + \E_q(\sigma^{-2})\} C_0
\\
&\quad  + (a^\xi + d/2 - 1) \log \kappa_B - \bigg\{ \sum\nolimits_{j=1}^d \E_q(\theta_j) \E_q(\xi_j^{-1}) / 4 + \E_q(d_2) \bigg\}\kappa_B
\\
&\quad  + (a^\tau + d/2 - 1) \log \lambda_B - \bigg\{\sum\nolimits_{j=1}^d \E_q(\gamma_j^2) \E_q(\tau_j^{-1}) / 4 + \E_q(e_2) \bigg\} \lambda_B\bigg].
\end{align*}
Thus, the optimal $q(C_0, \kappa_B, \lambda_B) = q(C_0)q(\kappa_B)q(\lambda_B)$, where $q(C_0)$ is $G(g^q, G^q)$, $q(\kappa_B)$ is $\G(a_{\kappa_B}^q, b_{\kappa_B}^q)$ and $q(\lambda_B)$ is $\G(a_{\lambda_B}^q, b_{\lambda_B}^q)$ where $g^q =  g_0 + c_0$, $a_{\kappa_B}^q = a^\xi + d/2$, $a_{\lambda_B}^q = a^\tau + d/2$,
\begin{equation*}
\begin{aligned}
G^q &= G_0 + \E_q(\sigma^{-2}) = G_0 + c^q / C^q, \\
b_{\kappa_B}^q &= \sum\nolimits_{j=1}^d \E_q(\theta_j) \E_q(\xi_j^{-1}) / 4 + \E_q(d_2) = (a_\xi^q / 4) \sum\nolimits_{j=1}^d \E_q(\theta_j) / b_{\xi_j}^q + a_{d_2}^q / b_{d_2}^q, \\
b_{\lambda_B}^q &= \sum\nolimits_{j=1}^d \E_q(\gamma_j^2) \E_q(\tau_j^{-1}) / 4 + \E_q(e_2)
= (a_\tau^q / 4) \sum\nolimits_{j=1}^d (\mu_{\gamma_j}^{q2} + v_{\gamma_j}^q) / b_{\tau_j}^q + a_{e_2}^q / b_{e_2}^q.
\end{aligned}
\end{equation*}

\item \leavevmode \vspace{-1.275\baselineskip}
\begin{align*}
q(d_2, e_2) & \propto \exp \left\{(a^\xi + c^\xi - 1) \log d_2 - \left(\E_q(\kappa_B) + 2 c^\xi/a^\xi \right)d_2  \right.
\\
& \left. \qquad \quad + (a^\tau + c^\tau - 1) \log e_2 - \left(\E_q(\lambda_B) + 2 c^\tau/a^\tau \right)e_2 \right\}
\end{align*}
The optimal $q(d_2, e_2) = q(d_2)q(e_2)$, where $q(d_2) \sim \G( a_{d_2}^q, b_{d_2}^q)$ and $q(e_2) \sim \G( a_{e_2}^q, b_{e_2}^q)$, with
\begin{equation*}
\begin{aligned}
 a_{d_2}^q & = a^\xi + c^\xi,   \\
 a_{e_2}^q & = a^\tau + c^\tau,
\end{aligned}
\qquad
\begin{aligned}
b_{d_2}^q & = \E_q(\kappa_B) + 2 c^\xi/a^\xi
= a_{\kappa_B}^q / b_{\kappa_B}^q + 2 c^\xi/a^\xi\\
b_{e_2}^q & = \E_q(\lambda_B) + 2 c^\tau/a^\tau
= a_{\lambda_B}^q / b_{\lambda_B}^q + 2 c^\tau/a^\tau.
\end{aligned}
\end{equation*}

\item \leavevmode \vspace{-1.275\baselineskip}
\begin{align*}
q(\bftheta) &\propto \exp \E_{-\bftheta} \bigg\{ \log p(\bfy \mid \wbfbeta, \bfgamma, \bftheta, \sigma^2) + \log p(\wbfbeta \mid \bfgamma, \bftheta) + \sum_{j=1}^d \log p(\theta_j \mid \xi_j, \kappa_B) \bigg\}
\\
& \propto \exp \E_{-\bftheta} \bigg\{\sigma^{-2} \sum_{t=1}^T  (\wbfbeta_t^\top D_\bftheta^{\bfa/2} \bfx_t) (y_t - \bfx_t^\top W_t \bfgamma) -\frac{1}{2\sigma^2} \sum_{t=1}^T  (\wbfbeta_t^\top D_\bftheta^{\bfa/2} \bfx_t)^2
\\
& \quad - \bfgamma^\top \overline{W}^\top (\Lambda \otimes D_\bftheta^{-1}) \overline{W} \bfgamma / 2 - \wbfbeta^\top (\Lambda \otimes D_\bftheta^{\bfa-\bfone} )\wbfbeta / 2 + \wbfbeta^\top (\Lambda \otimes D_\bftheta^{\bfa/2-\bfone}) \overline{W} \bfgamma
\\
& \quad -\frac{1}{2} \sum_{j=1}^d \{(T+1)(1-a_j) + 1 \} \log \theta_j  - \sum_{j=1}^d \frac{\kappa_B \theta_j}{4\xi_j}  \bigg\}.
\end{align*}
Under the D-CP/NCP scheme, we partition the index set $\{1, \dots, d\}$ into disjoint subsets, $S_\NCP$ and $S_\CP$, assigning each dimension to either CP or NCP. This formulation unifies fully CP and fully NCP schemes as special cases. Under this setup,
\begin{align*}
q(\bftheta) & \propto \exp \E_{-\bftheta} \bigg\{ \sigma^{-2}  \sum_{t=1}^T  \surd \bftheta_\NCP^\top D_{\bfx_{t, \NCP}}  \widetilde{\bfbeta}_{t,\NCP}
\big(y_t -\bfx_{t,\NCP}^{\top}\bfgamma_{\NCP}
-\bfx_{t,\CP}^{\top}\widetilde{\bfbeta}_{t,\CP} \big)
\\
& \quad -\frac{1}{2\sigma^2} \sum_{t=1}^T  \surd \bftheta_\NCP^\top D_{\bfx_{t, \NCP}} \widetilde{\bfbeta}_{t, \NCP} \widetilde{\bfbeta}_{t, \NCP}^\top D_{\bfx_{t, \NCP}} \surd{\bftheta_\NCP} - \frac{1}{2} \negmedspace \sum_{j \in S_\NCP} \negmedspace\negmedspace \log \theta_j
\\
& \quad - \negmedspace \sum_{j \in S_\NCP} \negmedspace\negmedspace \frac{\kappa_B \theta_j}{4\xi_j} - \frac{1}{2}
\sum_{t=1}^T (\widetilde{\bfbeta}_{t,\CP}-\widetilde{\bfbeta}_{t-1,\CP})^\top
D_{\bftheta_{\CP}}^{-1} (\widetilde{\bfbeta}_{t,\CP} - \widetilde{\bfbeta}_{t-1,\CP}) -\frac{T+2}{2} \negmedspace \sum_{j\in S_\CP} \negmedspace\negmedspace \log \theta_j
\\
& \quad - \frac{1}{2}(\widetilde{\bfbeta}_{0,\CP} - \bfgamma_{\CP})^\top D_{\bftheta_{\CP}}^{-1}(\widetilde{\bfbeta}_{0,\CP} - \bfgamma_{\CP}) - \sum_{j\in S_\CP} \frac{\kappa_B \theta_j}{4\xi_j} \bigg\}.
\end{align*}
Hence the variational posterior factorizes as $q(\bftheta) = q(\bftheta_\CP) q(\bftheta_\NCP)$ where $q(\bftheta_\CP) = \prod_{j \in S_\CP} q(\theta_j)$ with $q(\theta_j) = \GIG(p_\theta^q, a_{\theta_j}^q, b_{\theta_j}^q)$. The variational parameters are $p_\theta^q = -T/2$, $a_{\theta_j}^q = \E_q(\kappa_B) \E_q(\xi_j^{-1})/2 = (a_{\kappa_B}^q a_\xi^q) / (2b_{\kappa_B}^q b_{\xi_j}^q)$ and
\begin{align*}
b_{\theta_j}^q &= \sum_{t=1}^T   \E_q (\widetilde{\beta}_{tj} - \widetilde{\beta}_{t-1, j} )^2 +  \E_q(\widetilde{\beta}_{0j} - \gamma_j)^2 \\
&= \sum_{t=1}^{T}
\{ (\mu_{\widetilde{\beta}_{tj}}^q - \mu_{\widetilde{\beta}_{t-1,j}}^q)^2 + v_{\widetilde{\beta}_{tj}}^q + v_{\widetilde{\beta}_{t-1, j}}^q \negmedspace\negmedspace
- 2 \Cov_q(\widetilde{\beta}_{tj}, \widetilde{\beta}_{t-1,j}) \}
+ (\mu_{\widetilde{\beta}_{0j}}^q - \mu_{\gamma_j}^q )^2 + v_{\gamma_j}^q + v_{\widetilde{\beta}_{0j}}^q.
\end{align*}
On the other hand, for $q(\bftheta_\NCP)$, using a change-in-variable,
\begin{multline*}
q(\surd{\bftheta_\NCP}) \propto \exp \E_{-\surd{\bftheta_\NCP}} \bigg[ - \frac{\kappa_B\surd \bftheta_\NCP^\top \diag(\bfxi^{-1}) \surd{\bftheta_\NCP}}{4}
-\frac{1}{2\sigma^2} \sum_{t=1}^T  \surd\bftheta_\NCP^\top D_{\bfx_{t, \NCP}} \wbfbeta_{t, \NCP}
\\
\times  \Big\{ \wbfbeta_{t, \NCP}^\top D_{\bfx_{t, \NCP}} \surd{\bftheta_\NCP}
-2 (y_t - \bfx_{t, \NCP}^\top \bfgamma_\NCP -  \bfx_{t, \CP}^\top \widetilde{\bfbeta}_{t, \CP}) \Big\} \bigg].
\end{multline*}
Thus the optimal $q(\surd{\bftheta_\NCP})$ is $\N( \bfmu_{\surd{\bftheta_\NCP}}^q, \Sigma_{\surd{\bftheta_\NCP}}^q)$, where
\begin{align*}
\Sigma_{\surd \bftheta_\NCP}^q &= \bigg( \frac{c^q}{C^q} \sum_{t=1}^{T} D_{\bfx_{t, \NCP}} (\Sigma_{\wbfbeta_{t, \NCP}}^q + \bfmu_{\wbfbeta_{t, \NCP}}^q \bfmu_{\wbfbeta_{t, \NCP}}^{q\top}) D_{\bfx_{t, \NCP}}
+ \frac{\E_q(\kappa_B)}{2}\diag \{ \E_q (\bfxi^{-1} ) \} \bigg) ^{-1}, \\
\bfmu_{\surd{\bftheta_\NCP}}^q &= \Sigma_{\surd{\bftheta_\NCP}}^q  \frac{c^q}{C^q} \sum_{t=1}^{T} D_{\bfx_{t, \NCP}} \{ \bfmu_{\wbfbeta_{t, \NCP}}^q (y_t - \bfx_{t, \NCP}^\top\bfmu_{\bfgamma_\NCP}^q -  \bfx_{t, \CP}^\top \bfmu_{\wbfbeta_{t, \CP}}^q) \\
& -\Cov(\wbfbeta_{t, \NCP}, \widetilde{\bfbeta}_{t, \CP}) \bfx_{t, \CP} \}.
\end{align*}
\end{itemize}

\section{Derivation of evidence lower bound}\label{supp:derive_elbo}
This section derives the ELBO used in Algorithm \ref{alg:TVP_PNCP}. First, we present the terms in $\E_q\{\log p(\bfy, \Theta)\}$.
\begin{itemize}
\item
$\E_q \{\log p(\bfy \mid \wbfbeta, \bfgamma, \bftheta, \sigma^2)\} = - \Big( T \log(2\pi) + T\{ \log C^q-\psi(c^q)\} + (c^q / C^q) \sum_{t=1}^T [ (y_t - \bfx_t^\top W_t \bfmu_\bfgamma^q)^2 - 2(y_t - \bfx_t^\top W_t \bfmu_\bfgamma^q) \bfx_t^\top \E_q(D_\bftheta^{\bfa/2}) \bfmu_{\wbfbeta_t}^q + \bfx_t^\top W_t \Sigma_\bfgamma^q W_t^\top \bfx_t + \tr\{(\bfmu_{\wbfbeta_t}^q\bfmu_{\wbfbeta_t}^{q\top} + \Sigma_{\wbfbeta_t}^q) D_{\bfx_t} \E_q(\bftheta^{\bfa/2}\bftheta^{\bfa/2 \top}) D_{\bfx_t} \}] \Big) / 2$.
\item
$\E_q\{\log p(\wbfbeta \mid \bfgamma, \bftheta)\} = - \Big( (T+1)d\log(2\pi) + (T+1) \sum_{j=1}^d (1-a_j)\E_q(\log \theta_j) + \tr[ (\bfmu_\wbfbeta^q \bfmu_\wbfbeta^{q\top} + \Sigma_\wbfbeta^q) \{\Lambda \otimes \E_q(D_\bftheta^{\bfa-\bfone})\}] - 2 \bfmu_\wbfbeta^{q\top} \{\Lambda \otimes \E_q(D_\bftheta^{\bfa/2-\bfone})\} \overline{W} \bfmu_\bfgamma^q + \tr[ (\bfmu_\bfgamma^q \bfmu_\bfgamma^{q\top} + \Sigma_\bfgamma^q) \overline{W}^\top \{\Lambda \otimes \E_q(D_\bftheta^{-1})\}\overline{W} ] \Big) / 2$.
\item
$\E_q\{\log p(\theta_j \mid \xi_j, \kappa_B)\} = - \log 2 - \frac{1}{2}\log \pi + \{ \psi(a_{\kappa_B}^q)  - \log b_{\kappa_B}^q \}/2 - \{ \log b_{\xi_j}^q - \psi(a_\xi^q) \}/2  -a_{\kappa_B}^q a_\xi^q \E_q(\theta_j) / (4b_{\kappa_B}^q b_{\xi_j}^q) - \E_q(\log \theta_j)/2$, for D-CP/NCP where $j \in S_{\CP}$.
\item
$\E_q\{\log p(\surd{\theta_j} \mid \xi_j, \kappa_B)\} = - \log 2 - \log \sqrt{\pi}  + \{ \psi(a_{\kappa_B}^q)  - \log b_{\kappa_B}^q \}/2 - \{ \log b_{\xi_j}^q - \psi(a_\xi^q) \}/2 - a_{\kappa_B}^q a_\xi^q  \E_q(\theta_j) / (4b_{\kappa_B}^q b_{\xi_j}^q)$, for D-CP/NCP where $j \in S_{\NCP}$.
\item
$\E_q\{\log p(\vartheta_j \mid \xi_j, \kappa_B)\} = - \log 2 - \log \sqrt{\pi}  + \{ \psi(a_{\kappa_B}^q)  - \log b_{\kappa_B}^q \} / 2 - \{ \log b_{\xi_j}^q - \psi(a_\xi^q) \} /2  - a_{\kappa_B}^q a_\xi^q\E_q\{\exp(\vartheta_j)\} / (4b_{\kappa_B}^q b_{\xi_j}^q) + \E_q(\vartheta_j) / 2 $ for CLBS.
\item
$\E_q\{\log p(\gamma_j \mid \tau_j, \lambda_B)\} = - \log \sqrt{\pi} - \log 2 - \{ \log b_{\tau_j}^q - \psi(a_\tau^q) \}/2 + \{ \psi(a_{\lambda_B}^q)  - \log b_{\lambda_B}^q \}/2  - a_{\lambda_B}^q a_{\tau}^q  (\mu_{\gamma_j}^{q2} + v_{\gamma_j}^q) / (4 b_{\lambda_B}^q b_{\tau_j}^q)$.
\item
$\E_q\{\log p(\sigma^2 \mid C_0)\} = c_0 \{ \psi(g^q) - \log G^q \} - \log \Gamma(c_0) - (c_0 +1) \{ \log C^q - \psi(c^q) \} - g^qc^q / (G^qC^q)$.
\item
$\E_q\{\log p(C_0)\} = g_0 \log G_0 - \log \Gamma(g_0) + (g_0-1) \{ \psi(g^q) - \log G^q \} - G_0 g^q / G^q$.
\item
$\E_q\{\log p(\xi_j \mid \kappa_j)\} = c^\xi \{\psi(a_\kappa^q) - \log b_{\kappa_j}^q\} - \log \Gamma(c^\xi)- (c^\xi + 1) \{ \log b_{\xi_j}^q - \psi(a_\xi^q) \} - a_\kappa^q a_\xi^q / (b_{\kappa_j}^q b_{\xi_j}^q)$.
\item
$\E_q\{\log p(\tau_j  \mid \lambda_j)\} = c^\tau \{ \psi(a_\lambda^q) - \log b_{\lambda_j}^q \} - \log \Gamma(c^\tau) - (c^\tau + 1) \{ \log b_{\tau_j}^q - \psi(a_\tau^q) \} - a_\lambda^q a_\tau^q / (b_{\lambda_j}^q b_{\tau_j}^q)$.
\item
$\E_q\{\log p(\kappa_j)\} = a^\xi \log(a^\xi /c^\xi) - \log \Gamma(a^\xi) + (a^\xi - 1) \{ \psi(a_\kappa^q) - \log b_{\kappa_j}^q \} - a^\xi a_\kappa^q / (c^\xi b_{\kappa_j}^q)$.
\item
$\E_q\{\log p(\lambda_j)\} = a^\tau \log(a^\tau /c^\tau) - \log \Gamma(a^\tau) + (a^\tau - 1) \{ \psi(a_\lambda^q) - \log b_{\lambda_j}^q \} - a^\tau a_\lambda^q / (c^\tau b_{\lambda_j}^q)$.
\item
$\E_q\{\log p(\kappa_B \mid d_2)\} = a^\xi \{\psi(a_{d_2}^q) - \log b_{d_2}^q\} - \log \Gamma(a^\xi) + (a^\xi -1) \{ \psi(a_{\kappa_B}^q)  - \log b_{\kappa_B}^q \} - a_{d_2}^q a_{\kappa_B}^q / (b_{d_2}^q b_{\kappa_B}^q)$.
\item
$\E_q\{\log p(\lambda_B \mid e_2)\} = a^\tau \{\psi(a_{e_2}^q) - \log b_{e_2}^q\} - \log \Gamma(a^\tau) + (a^\tau -1) \{ \psi(a_{\lambda_B}^q)  - \log b_{\lambda_B}^q \} - a_{e_2}^q a_{\lambda_B}^q /(b_{e_2}^q b_{\lambda_B}^q)$.
\item
$\E_q\{\log p(d_2)\} = c^\xi \log(2c^\xi/a^\xi) - \log \Gamma(c^\xi) + (c^\xi - 1) \{\psi(a_{d_2}^q) - \log b_{d_2}^q\} - 2c^\xi a_{d_2}^q/(a^\xi b_{d_2}^q)$.
\item
$\E_q\{\log p(e_2)\} = c^\tau \log(2c^\tau/a^\tau) - \log \Gamma(c^\tau) + (c^\tau - 1) \{\psi(a_{e_2}^q) - \log b_{e_2}^q\} -  2c^\tau a_{e_2}^q / (a^\tau b_{e_2}^q) $.
\end{itemize}

Next, we present the terms in $\E_q\{ \log q(\Theta)\}$.
\begin{itemize}
\item
$\E_q\{ \log q(\wbfbeta)\} = -\{ (T+1)d\log(2\pi) + \log|\Sigma_\wbfbeta^q|+ (T+1)d \} / 2$.
\item
$\E_q\{ \log q(\bfgamma)\} = -\{d\log(2\pi) + \log |\Sigma_\bfgamma^q| + d\} /2$.
\item
$\E_q\{\log q(\bfvartheta)\} = - \{ d\log(2\pi) + \log|\Sigma_\bfvartheta^q| + d \} /2$, for CLBS.
\item $\E_q\{ \log q(\surd{\bftheta_\NCP}) \} = - \{ |S_\NCP| \log(2\pi) + \log |\Sigma_{\surd{\bftheta_\NCP}}^q| + |S_\NCP| \} /2$, for D-CP/NCP.
\item $\E_q\{ \log q(\theta_j) \} = (p_\theta^q / 2) \log (a_{\theta_j}^q / b_{\theta_j}^q) - \log2 - \log K_{p_\theta^q} (\sqrt{a_{\theta_j}^q b_{\theta_j}^q}) + (p_\theta^q -1) \E_q(\log \theta_j) - a_{\theta_j}^q \E_q(\theta_j) / 2 - b_{\theta_j}^q \E_q(\theta_j^{-1}) / 2$, for D-CP/NCP where $j \in S_{\CP}$.
\item $\E_q\{ \log q(\sigma^2)\} = c^q \log C^q - \log \Gamma(c^q) - (c^q + 1) \{ \log C^q - \psi(c^q) \} - c^q$.
\item
$\E_q\{ \log q(C_0)\} = g^q \log G^q - \log \Gamma(g^q) + (g^q-1) \{ \psi(g^q) - \log G^q \} - g^q$.
\item
$\E_q\{ \log q(\xi_j)\} = a_{\xi}^q \log b_{\xi_j}^q - \log \Gamma(a_{\xi}^q) - (a_{\xi}^q + 1) \{ \log b_{\xi_j}^q - \psi(a_{\xi}^q) \} - a_{\xi}^q$.
\item
$\E_q\{ \log q(\tau_j)\} = a_{\tau}^q \log b_{\tau_j}^q - \log \Gamma(a_{\tau}^q) - (a_{\tau}^q + 1) \{ \log b_{\tau_j}^q - \psi(a_{\tau}^q) \} - a_{\tau}^q$.
\item
$\E_q\{ \log q(\kappa_j)\} = a_{\kappa}^q \log b_{\kappa_j}^q - \log \Gamma(a_{\kappa}^q) + (a_{\kappa}^q-1) \{ \psi(a_{\kappa}^q) - \log b_{\kappa_j}^q \} - a_{\kappa}^q$.
\item
$\E_q\{ \log q(\lambda_j)\} = a_{\lambda}^q \log b_{\lambda_j}^q - \log \Gamma(a_{\lambda}^q) + (a_{\lambda}^q-1) \{ \psi(a_{\lambda}^q) - \log b_{\lambda_j}^q \} - a_{\lambda}^q$.
\item
$\E_q\{ \log q(\kappa_B)\} = a_{\kappa_B}^q \log b_{\kappa_B}^q - \log \Gamma(a_{\kappa_B}^q) + (a_{\kappa_B}^q - 1) \{ \psi(a_{\kappa_B}^q)  - \log b_{\kappa_B}^q \} - a_{\kappa_B}^q$.
\item
$\E_q\{ \log q(\lambda_B)\} = a_{\lambda_B}^q \log b_{\lambda_B}^q - \log \Gamma(a_{\lambda_B}^q) + (a_{\lambda_B}^q - 1) \{ \psi(a_{\lambda_B}^q)  - \log b_{\lambda_B}^q \} - a_{\lambda_B}^q$.
\item
$\E_q\{ \log q(d_2)\} = a_{d_2}^q \log b_{d_2}^q - \log \Gamma(a_{d_2}^q) + (a_{d_2}^q-1) \{ \psi(a_{d_2}^q) - \log b_{d_2}^q \} - a_{d_2}^q$.
\item
$\E_q\{ \log q(e_2)\} = a_{e_2}^q \log b_{e_2}^q - \log \Gamma(a_{e_2}^q) + (a_{e_2}^q-1) \{ \psi(a_{e_2}^q) - \log b_{e_2}^q \} - a_{e_2}^q$.
\end{itemize}

Combining the above terms leads to the following expression for the ELBO,
\begin{align*}
& \mL  = \E_q\{ \log p(\bfy, \Theta) - \log q(\Theta) \} \\
&= - \dfrac{c^q}{2C^q}\sum_{t=1}^T [ (y_t - \bfx_t^\top W_t \bfmu_\bfgamma^q)^2 - 2(y_t - \bfx_t^\top W_t \bfmu_\bfgamma^q) \bfx_t^\top \E_q(D_\bftheta^{\bfa/2}) \bfmu_{\wbfbeta_t}^q + \bfx_t^\top W_t \Sigma_\bfgamma^q W_t^\top \bfx_t
\\
&+ \tr\{(\bfmu_{\wbfbeta_t}^q\bfmu_{\wbfbeta_t}^{q\top} + \Sigma_{\wbfbeta_t}^q) D_{\bfx_t} \E_q(\bftheta^{\bfa/2}\bftheta^{\bfa/2 \top})D_{\bfx_t} \}] - \frac{1}{2} \tr[ (\bfmu_\wbfbeta^q \bfmu_\wbfbeta^{q\top} + \Sigma_\wbfbeta^q) \{\Lambda \otimes \E_q(D_\bftheta^{\bfa-\bfone})\}]
\\
&+ \bfmu_\wbfbeta^{q\top} \{\Lambda \otimes \E_q(D_\bftheta^{\bfa/2-\bfone})\} \overline{W} \bfmu_\bfgamma^q - \frac{1}{2}\tr[ (\bfmu_\bfgamma^q \bfmu_\bfgamma^{q\top} + \Sigma_\bfgamma^q) \overline{W}^\top \{\Lambda \otimes \E_q(D_\bftheta^{-1})\}\overline{W} ]
\\
&+ \sum_{j=1}^d \left[ -\dfrac{a_{\kappa_B}^q a_\xi^q }{4b_{\kappa_B}^q b_{\xi_j}^q}\E_q(\theta_j) - \dfrac{a_{\lambda_B}^q a_{\tau}^q}{4 b_{\lambda_B}^q b_{\tau_j}^q} (\mu_{\gamma_j}^{q2} + v_{\gamma_j}^q)
- \dfrac{a_\kappa^q}{b_{\kappa_j}^q} \left( \dfrac{a_\xi^q}{b_{\xi_j}^q} + \dfrac{a^\xi}{c^\xi} \right) - \dfrac{a_\lambda^q}{b_{\lambda_j}^q} \left( \dfrac{a^\tau}{c^\tau} + \dfrac{a_\tau^q}{b_{\tau_j}^q}\right)\right.
\\
&  - \dfrac{a_{d_2}^q}{b_{d_2}^q} \left( \frac{a_{\kappa_B}^q}{b_{\kappa_B}^q}
+ \dfrac{2c^\xi}{a^\xi} \right) - \dfrac{a_{e_2}^q }{b_{e_2}^q } \left( \frac{a_{\lambda_B}^q}{b_{\lambda_B}^q} + \dfrac{2c^\tau}{a^\tau} \right) - \dfrac{g^q}{G^q} \left( G_0 + \dfrac{c^q}{C^q} \right) + \frac{1}{2}\log|\Sigma_\wbfbeta^q|
+ \frac{1}{2}\log |\Sigma_\bfgamma^q|
\\
&
- c^q \log C^q + \log \Gamma(c^q)  + c^q - g^q \log G^q + \log \Gamma(g^q) + g^q
- a_{\kappa_B}^q \log b_{\kappa_B}^q + \log \Gamma(a_{\kappa_B}^q)  + a_{\kappa_B}^q
\\
&
- a_{\lambda_B}^q \log b_{\lambda_B}^q + \log \Gamma(a_{\lambda_B}^q) + a_{\lambda_B}^q
- a_{d_2}^q \log b_{d_2}^q + \log \Gamma(a_{d_2}^q) + a_{d_2}^q
- a_{e_2}^q \log b_{e_2}^q + \log \Gamma(a_{e_2}^q)
\\
&
+ a_{e_2}^q - \sum_{j=1}^d [a_{\xi}^q \log b_{\xi_j}^q + a_{\tau}^q \log b_{\tau_j}^q + a_{\kappa}^q \log b_{\kappa_j}^q + a_{\lambda}^q \log b_{\lambda_j}^q] + d \{  \log \Gamma(a_{\xi}^q) + a_{\xi}^q + \log \Gamma(a_{\tau}^q)
\\
& + a_{\tau}^q  + \log \Gamma(a_{\kappa}^q) + a_{\kappa}^q + \log \Gamma(a_{\lambda}^q) + a_{\lambda}^q\}
+ C + \widetilde{\mathcal{L}}_\bftheta ,
\end{align*}
where $C$ denotes a constant independent of the variational parameters and $\mathcal{L}_{\bftheta}$ depends on the chosen scheme (D-CP/NCP or CLBS). In evaluating $\mL$, we have used the following fixed updates,
\begin{gather*}
a_\xi^q = c^\xi + 1/2, \quad
a_{\kappa_B}^q = a^\xi + d/2, \quad
a_\kappa^q = a^\xi + c^\xi, \quad
a_{d_2}^q = a^\xi + c^\xi, \quad
c^q = c_0 + T/2,  \\
a_\tau^q = c^\tau + 1/2, \quad
a_{\lambda_B}^q = a^\tau + d/2, \quad
a_\lambda^q = a^\tau + c^\tau, \quad
a_{e_2}^q = a^\tau + c^\tau, \quad
g^q = g_0 + c_0, \quad
\end{gather*}
and $p_\theta^q=-T/2$ to obtain further simplification. For the D-CP/NCP scheme,
\begin{align*}
\widetilde{\mathcal{L}}_\bftheta &= \mathbbm{1}_{|S_\NCP| \geq 1} \log |\Sigma_{\surd{\bftheta_\NCP}}^q| / 2 + |S_\NCP| \{\log(2\pi) + 1\} / 2 + |S_\CP| \log 2\\
&\quad + \sum_{j\in S_{\CP}} \left( \log K_{p_\theta^q} (\sqrt{a_{\theta_j}^q b_{\theta_j}^q}) + \frac{a_{\theta_j}^q}{2} \E_q(\theta_j) + \frac{b_{\theta_j}^q}{2} \E_q(\theta_j^{-1}) - \frac{p_\theta^q}{2} \log \frac{a_{\theta_j}^q}{b_{\theta_j}^q} \right).
\end{align*}
For the CLBS scheme,
\begin{align*}
    \widetilde{\mathcal{L}}_\bftheta &=
    \frac{1}{2}\log|\Sigma_{\bfvartheta}^q|
    + \frac{d}{2}\{\log(2\pi) + 1\} - \sum_{j=1}^d \frac{(T+1)(1-a_j) - 1}{2} \mu_{\vartheta_j}^q.
\end{align*}
We retained the normalizing constants in $\mathcal{L}_{\bftheta}$ so that the ELBO remains directly comparable across both schemes.

\section{Updating variational posterior of $\log \bftheta$} \label{supp:derive_theta_PNCP}
This section derives the gradients of the ELBO with respect to the variational parameters $(\bfmu_\bfvartheta^q, C_\bfvartheta^q)$ of $q(\bfvartheta)$ for gradient ascent optimization in CLBS. The relevant terms in the ELBO involving $(\bfmu_\bfvartheta^q, C_\bfvartheta^q)$ are:
\begin{align*}
&\mL_\bfvartheta = \dfrac{c^q}{C^q}\sum_{t=1}^T [(y_t - \bfx_t^\top W_t \bfmu_\bfgamma^q) \bfx_t^\top \E_q(D_\bftheta^{\bfa/2}) \bfmu_{\wbfbeta_t}^q - \frac{1}{2}\tr\{( \bfmu_{\wbfbeta_t}^q\bfmu_{\wbfbeta_t}^{q\top} + \Sigma_{\wbfbeta_t}^q) D_{\bfx_t} \E_q(\bftheta^{\bfa/2}\bftheta^{\bfa/2 \top})D_{\bfx_t} \}]
\\
& \quad  - \frac{\tr[ (\bfmu_\wbfbeta^q \bfmu_\wbfbeta^{q\top} + \Sigma_\wbfbeta^q) \{\Lambda \otimes \E_q(D_\bftheta^{\bfa-\bfone})\}] +  \tr[ (\bfmu_\bfgamma^q \bfmu_\bfgamma^{q\top} + \Sigma_\bfgamma^q) \overline{W}^\top \{\Lambda \otimes \E_q(D_\bftheta^{-1})\}\overline{W} ] - \log |\Sigma_\bfvartheta^q|}{2}
\\
& \quad + \bfmu_\wbfbeta^{q\top} \{\Lambda \otimes \E_q(D_\bftheta^{\bfa/2-\bfone})\} \overline{W} \bfmu_\bfgamma^q - \sum_{j=1}^d \left[ \frac{(T+1)(1-a_j) - 1}{2} \mu_{\vartheta_j}^q + \dfrac{a_{\kappa_B}^q a_\xi^q }{4b_{\kappa_B}^q b_{\xi_j}^q}\E_q(\theta_j) \right]
\\
&= \dfrac{c^q}{C^q}\sum_{j=1}^d \exp\{ (a_j/2)\mu_{\vartheta_j}^q + (a_j/2)^2 v_{\vartheta_j}^q/2 \} \sum_{t=1}^T (y_t - \bfx_t^\top W_t \bfmu_\bfgamma^q) x_{tj} \mu_{\widetilde{\beta}_{tj}}^q
\\
&\quad - \dfrac{c^q}{2C^q} \sum_{j=1}^d \exp\{a_j\mu_{\vartheta_j}^q + a_j^2 v_{\vartheta_j}^q/2 \} \sum_{t=1}^T x_{tj}^2 (\mu_{\widetilde{\beta}_{tj}}^{q2} + \Sigma_{\widetilde{\beta}_{tj}}^q) - \dfrac{c^q}{C^q}\sum_{j=1}^d \sum_{k > j} \exp\bigg\{ \frac{a_j}{2} \mu_{\vartheta_j}^q
\\
&\quad  + \frac{a_k}{2} \mu_{\vartheta_k}^q
+ \frac{a_j^2 v_{\vartheta_j}^q + 2 a_ja_k v_{\vartheta_j, \vartheta_k}^q + a_k^2 v_{\vartheta_k}^q}{8}  \bigg\} \sum_{t=1}^T x_{tj} x_{tk} (\mu_{\widetilde{\beta}_{tj}}^{q}\mu_{\widetilde{\beta}_{tk}}^{q} + \Sigma_{\widetilde{\beta}_{tj}, \widetilde{\beta}_{tk}}^q)
\\
&\quad - \frac{1}{2} \sum_{j=1}^d \exp\{ (a_j-1)\mu_{\vartheta_j}^q + (a_j-1)^2 v_{\vartheta_j}^q/2 \} \tr\{ (\mu_{\wbfbeta_{[j]}}^q \mu_{\wbfbeta_{[j]}}^{q\top} + \Sigma_{\wbfbeta_{[j]}}^q) \Lambda \} + \frac{1}{2}\log |\Sigma_\bfvartheta^q|
\\
&\quad + \sum_{j=1}^d \exp\{ (a_j/2-1)\mu_{\vartheta_j}^q + (a_j/2-1)^2 v_{\vartheta_j}^q/2 \} \mu_{\wbfbeta_{[j]}}^{q\top} \Lambda \overline{W}_{[j]} \bfmu_\bfgamma^q - \sum_{j=1}^d \frac{(T+1)(1-a_j) - 1}{2} \mu_{\vartheta_j}^q
\\
&\quad - \frac{1}{2} \sum_{j=1}^d \exp\{ -\mu_{\vartheta_j}^q + v_{\vartheta_j}^q/2 \} \tr\{ (\bfmu_\bfgamma^q \bfmu_\bfgamma^{q\top} + \Sigma_\bfgamma^q) \overline{W}_{[j]}^\top \Lambda \overline{W}_{[j]}\}
- \sum_{j=1}^d \dfrac{a_{\kappa_B}^q a_\xi^q }{4b_{\kappa_B}^q b_{\xi_j}^q}\exp(\mu_{\vartheta_j}^q + v_{\vartheta_j}^q/2).
\end{align*}
Differentiating $\mL_\bfvartheta$ with respect to $\mu_{\vartheta_j}^q$,
\begin{align*}
\nabla_{\mu_{\vartheta_j}^q}  \mL_\bfvartheta
&= \dfrac{a_j c^q }{2C^q} \exp \left\{ \frac{a_j\mu_{\vartheta_j}^q}{2} + \frac{a_j^2 v_{\vartheta_j}^q}{8} \right\} \sum_{t=1}^T (y_t - \bfx_t^\top W_t \bfmu_\bfgamma^q) x_{tj} \mu_{\widetilde{\beta}_{tj}}^q - \frac{(T+1)(1-a_j) - 1}{2}
\\
&\quad - \dfrac{a_j c^q}{2C^q} \exp \left\{a_j\mu_{\vartheta_j}^q +\frac{ a_j^2 v_{\vartheta_j}^q}{2} \right\}\sum_{t=1}^T x_{tj}^2 (\mu_{\widetilde{\beta}_{tj}}^{q2} + \Sigma_{\widetilde{\beta}_{tj}}^q) - \dfrac{a_j c^q}{2C^q} \sum_{k \neq j} \exp\bigg\{ \frac{a_j}{2} \mu_{\vartheta_j}^q
\\
&\quad + \frac{a_k}{2} \mu_{\vartheta_k}^q + \frac{a_j^2 v_{\vartheta_j}^q + 2 a_ja_k v_{\vartheta_j, \vartheta_k}^q + a_k^2 v_{\vartheta_k}^q}{8}  \bigg\} \sum_{t=1}^T x_{tj} x_{tk} (\mu_{\widetilde{\beta}_{tj}}^{q}\mu_{\widetilde{\beta}_{tk}}^{q} + \Sigma_{\widetilde{\beta}_{tj}, \widetilde{\beta}_{tk}}^q)
\\
&\quad - \frac{a_j-1}{2} \exp\{ (a_j-1)\mu_{\vartheta_j}^q + (a_j-1)^2 v_{\vartheta_j}^q/2 \} \tr\{ (\mu_{\wbfbeta_{[j]}}^q \mu_{\wbfbeta_{[j]}}^{q\top} + \Sigma_{\wbfbeta_{[j]}}^q) \Lambda \}
\\
&\quad + \left( \frac{a_j}{2} - 1 \right)\exp \left\{ \left( \frac{a_j}{2} - 1 \right)\mu_{\vartheta_j}^q + \frac{v_{\vartheta_j}^q}{2}\left( \frac{a_j}{2} - 1 \right)^2  \right\} \mu_{\wbfbeta_{[j]}}^{q\top} \Lambda \overline{W}_{[j]} \bfmu_\bfgamma^q
\\
&\quad + \frac{1}{2} \exp\{ -\mu_{\vartheta_j}^q + v_{\vartheta_j}^q/2 \} \tr\{ (\bfmu_\bfgamma^q \bfmu_\bfgamma^{q\top} + \Sigma_\bfgamma^q) \overline{W}_{[j]}^\top \Lambda \overline{W}_{[j]}\}
- \dfrac{a_{\kappa_B}^q a_\xi^q }{4b_{\kappa_B}^q b_{\xi_j}^q}\exp\left(\mu_{\vartheta_j}^q + \frac{v_{\vartheta_j}^q}{2} \right).
\end{align*}

Let $\mL_\bfvartheta = \mL_\bfvartheta^1 + \mL_\bfvartheta^2$, where $\mL_\bfvartheta^2 = \frac{1}{2}\log |\Sigma_\bfvartheta^q|$. Write $\Sigma_{\bfvartheta}^q = C_{\bfvartheta}^q C_{\bfvartheta}^{q\top}$. Then $\mL_\bfvartheta^2 = \frac{1}{2}\log |\Sigma_\bfvartheta^q| = \log |C_{\bfvartheta}^q|$ and $\nabla_{\vech(C_{\bfvartheta}^q)} \mL_\bfvartheta^2 = \vech(C_{\bfvartheta}^{q -\top})$. In addition, $\nabla_{\vech(C_{\bfvartheta}^q)} \mL_\bfvartheta^1 = 2\vech(\nabla_{\Sigma_{\bfvartheta}^q} \mL_\bfvartheta^1 \; C_{\bfvartheta}^q)$. Differentiating with respect to $v_{\vartheta_j}^q$,
\begin{align*}
\nabla_{v_{\vartheta_j}^q} \mL_\bfvartheta^1
& = \dfrac{a_j^2 c^q}{8C^q} \exp \left\{ \frac{a_j\mu_{\vartheta_j}^q}{2} + \frac{a_j^2 v_{\vartheta_j}^q}{8}  \right\} \sum_{t=1}^T (y_t - \bfx_t^\top W_t \bfmu_\bfgamma^q) x_{tj} \mu_{\widetilde{\beta}_{tj}}^q
\\
&\quad  - \dfrac{a_j^2 c^q}{4C^q} \exp \left\{a_j\mu_{\vartheta_j}^q + \frac{a_j^2 v_{\vartheta_j}^q}{2} \right\}\sum_{t=1}^T x_{tj}^2 (\mu_{\widetilde{\beta}_{tj}}^{q2} + \Sigma_{\widetilde{\beta}_{tj}}^q) - \dfrac{a_j^2 c^q}{8C^q} \sum_{k \neq j} \exp\bigg\{ \frac{a_j}{2} \mu_{\vartheta_j}^q
\\
&\quad
+ \frac{a_k}{2} \mu_{\vartheta_k}^q + \frac{a_j^2 v_{\vartheta_j}^q + 2 a_ja_k v_{\vartheta_j, \vartheta_k}^q + a_k^2 v_{\vartheta_k}^q}{8} \bigg\} \sum_{t=1}^T x_{tj} x_{tk} (\mu_{\widetilde{\beta}_{tj}}^{q}\mu_{\widetilde{\beta}_{tk}}^{q} + \Sigma_{\widetilde{\beta}_{tj}, \widetilde{\beta}_{tk}}^q)
\\
&\quad - \frac{(a_j-1)^2}{4} \exp \left\{ (a_j-1)\mu_{\vartheta_j}^q + \frac{(a_j-1)^2 v_{\vartheta_j}^q}{2} \right\} \tr\{ (\mu_{\wbfbeta_{[j]}}^q \mu_{\wbfbeta_{[j]}}^{q\top} + \Sigma_{\wbfbeta_{[j]}}^q) \Lambda \}
\\
&\quad + \frac{1}{2}\left( \frac{a_j}{2} - 1 \right)^2 \exp \left\{ \left( \frac{a_j}{2} - 1 \right)\mu_{\vartheta_j}^q + \frac{ v_{\vartheta_j}^q}{2}\left( \frac{a_j}{2} - 1 \right)^2 \right\} \mu_{\wbfbeta_{[j]}}^{q\top} \Lambda \overline{W}_{[j]} \bfmu_\bfgamma^q
\\
&\quad - \frac{1}{4} \exp\{ -\mu_{\vartheta_j}^q + v_{\vartheta_j}^q/2 \} \tr\{ (\bfmu_\bfgamma^q \bfmu_\bfgamma^{q\top} + \Sigma_\bfgamma^q) \overline{W}_{[j]}^\top \Lambda \overline{W}_{[j]}\}- \dfrac{a_{\kappa_B}^q a_\xi^q }{8b_{\kappa_B}^q b_{\xi_j}^q}\exp\left(\mu_{\vartheta_j}^q + \frac{v_{\vartheta_j}^q}{2} \right).
\end{align*}

Differentiating with respect to $v_{\vartheta_j, \vartheta_k}^q$,
\begin{align*}
\nabla_{v_{\vartheta_j, \vartheta_k}^q} \mL_\bfvartheta^1
&= - \dfrac{a_j a_k c^q}{4C^q}\exp\left\{ \frac{a_j}{2} \mu_{\vartheta_j}^q + \frac{a_k}{2} \mu_{\vartheta_k}^q + \frac{a_j^2 v_{\vartheta_j}^q + 2 a_ja_k v_{\vartheta_j, \vartheta_k}^q + a_k^2 v_{\vartheta_k}^q}{8}  \right\} \\
& \quad \times \sum_{t=1}^T x_{tj} x_{tk} (\mu_{\widetilde{\beta}_{tj}}^{q}\mu_{\widetilde{\beta}_{tk}}^{q} + \Sigma_{\widetilde{\beta}_{tj}, \widetilde{\beta}_{tk}}^q).
\end{align*}

\section{Numerical evaluation of modified Bessel function}\label{supp:Bessel}
As discussed in Section \ref{sec:MFVB_TVP}, under the CP, $q(\theta_j) = \GIG(p_\theta^q, a_{\theta_j}^q, b_{\theta_j}^q)$, with pdf,
\[
f(x) = \frac{(a/b)^{p/2}}{2K_p(\sqrt{ab})} x^{p-1} \exp\left( -\frac{ax + b/x}{2}\right), \quad p\in \mathbb{R}, a>0, b>0,
\]
where $K_p(\cdot)$ is the modified Bessel function of the second kind. Furthermore,
\[
\E(X) = \sqrt{\frac{b}a} \frac{K_{p+1}(\sqrt{ab})}{K_{p}(\sqrt{ab})},
\quad \text{and} \quad
\E \left( \frac{1}{X} \right) = \sqrt{\frac{a}{b}} \frac{K_{p+1}(\sqrt{ab})}{K_{p}(\sqrt{ab})} - \frac{2p}{b}.
\]
Accurate numerical evaluation of $K_\nu(x)$ is crucial for updating the variational parameters and evaluating the ELBO. However, standard routines such as {\tt besselk(nu, x)} from Julia’s {\tt SpecialFunctions} package (based on the AMOS algorithm), often suffer from overflow, underflow, or severe precision loss when the order $|\nu|$ or argument $x$ is large.

To ensure numerical stability when $|\nu|$ is large, we adopt the Debye asymptotic approximation \citep[Section 10.7 of][]{Olver1974AsymptoticsAS}. Let $x = \nu z$ and define
\[
t = \frac{1}{\sqrt{1+z^2}},
\qquad
\xi = \sqrt{1+z^2} + \log\!\left( \frac{z}{1+\sqrt{1+z^2}} \right).
\]
For $0<z<\infty$, as $|\nu| \rightarrow \infty$ ($K_{-\nu}(x) = K_\nu(x)$ via symmetry), the asymptotic expansion is
\begin{align*}
K_\nu(x) = K_{\nu}(\nu z)\sim\sqrt{\frac{\pi}{2\nu}}\,\frac{e^{-\nu\xi}}{(1+z^2)^{1/4}}\sum_{s=0}^{\infty}(-1)^s\frac{U_s(t)}{\nu^s},
\end{align*}
where the Debye polynomials $U_s(t)$ satisfy the recursion,
\[
U_{s+1}(t)=\frac{1}{2}t^{2}(1-t^{2})U_s'(t)+\frac{1}{8}\int_{0}^{t}\left(1-5q^{2}\right)U_s(q)\,dq.
\]
Following \cite{Bitto2019}, we retain terms up to $s=4$. The first five Debye polynomials \citep[pg. 366][]{Abramowitz1972} are:
\begin{align*}
U_0(t) &= 1, \quad
U_1(t) = (3t-5t^3)/24, \quad
U_2(t) = (81t^2-462t^4+385t^6)/1152,\\
U_3(t) &= (30375t^3-369603t^5+765765t^7-425425t^9)/414720,\\
U_4(t) &= \frac{4465125t^4-94121676t^6+349922430t^8-446185740t^{10}+185910725t^{12}}{39813120}.
\end{align*}

To balance numerical stability against truncation error, we switch between the AMOS-based algorithm and this 5-term Debye approximation based on empirical performance. Grid evaluations demonstrate that  the truncated expansion achieves relative errors below approximately $10^{-9}$ on the log scale for $|\nu| \ge 30$. Consequently, we apply the Debye approximation whenever $|\nu| \ge \max(30,\,2x)$. Furthermore, to prevent underflow in the standard evaluation of $K_{\nu}(x)$ via {\tt besselk()}, we invoke the Debye approximation whenever $x > 600$, regardless of $\nu$.

\section{Additional MFVB results}\label{supp:VA_results}
This section provides detailed results for the MFVB applications in Section \ref{sec:VA_Applications}. For each dataset, we report posterior mean estimates for $\bftheta$ and $\bfgamma$ to assess agreement with MCMC.

\subsection{Simulation study}
For $j\in \{1, 2, 3\}$, where true $\theta_j = 0$, NCP and D-CP/NCP shrink $\theta_j$ toward 0 more aggressively than other methods, whereas CP and CLBS yield posterior means closer to those of MCMC. For $j \in \{4, 5, 6\}$, all methods correctly identify larger magnitudes for $\theta_j$ and preserve their relative ordering, with $\theta_6$ displaying the strongest time variation. The estimates of $\gamma_5$ and $\gamma_6$ are particularly shrunk toward 0 under CP and NCP.

{\small
\begin{table}[H]
\centering
\caption{Ground truth values and posterior mean estimates of $\bftheta$ and $\bfgamma$ for simulated data.}
\label{tab:simulation_VA_estimates}
\small
\begin{tabular}{cccrrrrr}
\hline
& $j$ & Truth & CP & NCP & D-CP/NCP & CLBS & MCMC \\  \hline
\multirow{6}{*}{$\theta_j$}
    & 1 & $0$ & $0.001583$ & $0.000082$ & $0.000077$ & $0.000519$ & $0.000823$ \\
    & 2 & $0$ & $0.001555$ & $0.000080$ & $0.000075$ & $0.000661$ & $0.001145$ \\
    & 3 & $0.01$ & $0.004892$ & $0.000096$ & $0.000091$ & $0.004948$ & $0.009047$ \\
    & 4 & $0.05$ & $0.083624$ & $0.079756$ & $0.084098$ & $0.082915$ & $0.091448$ \\
    & 5 & $0.1$ & $0.116893$ & $0.106761$ & $0.108903$ & $0.113908$ & $0.122695$ \\
    & 6 & $0.5$ & $0.627776$ & $0.545855$ & $0.633442$ & $0.622243$ & $0.654714$ \\
\hline
\multirow{6}{*}{$\gamma_j$}
    & 1 & $0$ & $-0.003384$ & $-0.056216$ & $-0.082669$ & $-0.065962$ & $-0.049446$ \\
    & 2 & $-0.3$ & $-0.007675$ & $-0.321737$ & $-0.337446$ & $-0.159053$ & $-0.165907$ \\
    & 3 & $1.5$ & $1.981320$ & $2.627570$ & $2.631620$ & $2.023010$ & $2.037340$ \\
    & 4 & $0$ & $-0.010030$ & $-0.005253$ & $-0.033021$ & $-0.094777$ & $-0.100605$ \\
    & 5 & $-0.3$ & $-0.023762$ & $-0.009946$ & $-0.087826$ & $-0.194594$ & $-0.198264$ \\
    & 6 & $1.5$ & $0.017417$ & $0.003972$ & $0.090469$ & $0.173415$ & $0.272802$ \\
    \hline
    \end{tabular}
\end{table}}

\subsection{U.S. Macroeconomic data}
Across all methods, the intercept and {\tt inf\_lag} exhibit the largest $\theta_j$ values, indicating pronounced time variation, alongside consistently positive $\gamma_j$ estimates. Conversely, {\tt une\_lag} yields smaller posterior means for $\theta_j$ and negative $\gamma_j$ estimates, reflecting weaker time variation. For {\tt tbi\_lag}, NCP, D-CP/NCP and CLBS shrink strongly $\theta_j$ toward 0, suggesting a nearly constant coefficient, whereas CP and MCMC yield somewhat larger estimates. Nonetheless, posterior means of $\gamma_j$ for {\tt tbi\_lag} remain near 0 across all methods, indicating a limited contribution from this predictor.

\begin{table}[H]
    \centering
    \caption{Posterior mean estimates of $\bftheta$ and $\bfgamma$ for U.S. macro data.}
    \label{tab:usmacro_VA_estimates}
    \small
    \begin{tabular}{llrrrrr}
    \hline
& Variable & CP & NCP & D-CP/NCP & CLBS & MCMC \\
    \hline
\multirow{4}{*}{$\theta_j$}
        & Intercept
        & $0.004194$ & $0.002363$ & $0.003877$ & $0.004575$ & $0.005004$ \\
        & {\tt inf\_lag}
        & $0.003188$ & $0.002581$ & $0.002995$ & $0.003313$ & $0.003591$ \\
        & {\tt une\_lag}
        & $0.000731$ & $0.000992$ & $0.000917$ & $0.000408$ & $0.000313$ \\
        & {\tt tbi\_lag}
        & $0.000221$ & $0.000000$ & $0.000000$ & $0.000007$ & $0.000117$ \\
    \hline
        \multirow{4}{*}{$\gamma_j$}
        & Intercept
        & $1.091920$ & $1.326210$ & $1.107770$ & $1.071540$ & $1.053910$ \\
        & {\tt inf\_lag}
        & $0.792006$ & $0.914901$ & $0.798461$ & $0.810041$ & $0.815358$ \\
        & {\tt une\_lag}
        & $-0.059085$ & $-0.001483$ & $-0.051547$ & $-0.084529$ & $-0.097801$ \\
        & {\tt tbi\_lag}
        & $0.000882$ & $0.002967$ & $0.000486$ & $0.012001$ & $0.016315$\\
    \hline
    \end{tabular}
\end{table}

\subsection{U.S. inflation data}
Both NCP and D-CP/NCP strongly shrink $\theta_j$ toward 0 for nearly all predictors,  except {\tt lag1}, suggesting predominantly constant coefficients. In contrast, CP and MCMC yield larger $\theta_j$ estimates for the intercept, {\tt lag1}, {\tt lag3}, {\tt napmpri}, {\tt spread}, and {\tt treas}, indicating more pronounced time variation. CLBS produces larger $\theta_j$ estimates for the intercept, {\tt lag3}, and {\tt spread}, while strongly shrinking remaining dimensions toward 0. CP shrinks all estimates of $\gamma_j$ closer to 0.

{\small
\begin{longtable}{llrrrrr}
    \caption{Posterior mean estimates of $\bftheta$ and $\bfgamma$ for U.S. inflation data.}
    \label{tab:usinf_VA_estimates}
    \small
    \\
\hline
&    Variable & CP & NCP & D-CP/NCP & CLBS & MCMC \\
\hline
    \endfirsthead
\hline
&    Variable & CP & NCP & D-CP/NCP & CLBS & MCMC \\
\hline
    \endhead

\hline
    \endfoot
    \multirow{18}{*}{$\theta_j$}
    & Intercept & $0.013545$ & $0.000001$ & $0.000001$ & $0.012787$ & $0.013979$ \\
    & {\tt lag1} & $0.001335$ & $0.000566$ & $0.000566$ & $0.000012$ & $0.001061$ \\
    & {\tt lag2} & $0.000130$ & $0.000001$ & $0.000001$ & $0.000013$ & $0.000130$ \\
    & {\tt lag3} & $0.000989$ & $0.000001$ & $0.000001$ & $0.000844$ & $0.001143$ \\
    & {\tt cons} & $0.000133$ & $0.000001$ & $0.000001$ & $0.000014$ & $0.000185$ \\
    & {\tt dom\_inv} & $0.000227$ & $0.000001$ & $0.000001$ & $0.000011$ & $0.000150$ \\
    & {\tt dow} & $0.000205$ & $0.000001$ & $0.000001$ & $0.000009$ & $0.000184$ \\
    & {\tt emp} & $0.000188$ & $0.000001$ & $0.000001$ & $0.000012$ & $0.000203$ \\
    & {\tt exp} & $0.000251$ & $0.000001$ & $0.000001$ & $0.000014$ & $0.000243$ \\
    & {\tt gdp} & $0.000165$ & $0.000001$ & $0.000001$ & $0.000009$ & $0.000070$ \\
    & {\tt hstarts} & $0.000220$ & $0.000001$ & $0.000001$ & $0.000008$ & $0.000511$ \\
    & {\tt m1} & $0.000404$ & $0.000001$ & $0.000001$ & $0.000007$ & $0.000266$ \\
    & {\tt napmpri} & $0.001177$ & $0.000001$ & $0.000001$ & $0.000022$ & $0.000842$ \\
    & {\tt napmsdi} & $0.000081$ & $0.000002$ & $0.000002$ & $0.000009$ & $0.000036$ \\
    & {\tt pmi} & $0.000412$ & $0.000001$ & $0.000001$ & $0.000009$ & $0.000141$ \\
    & {\tt spread} & $0.000765$ & $0.000001$ & $0.000001$ & $0.001936$ & $0.000901$ \\
    & {\tt treas} & $0.001041$ & $0.000001$ & $0.000001$ & $0.000009$ & $0.007015$ \\
    & {\tt unemp} & $0.000474$ & $0.000001$ & $0.000001$ & $0.000010$ & $0.000489$ \\
    \hline

    \multirow{18}{*}{$\gamma_j$}
    & Intercept & $0.000069$ & $1.155680$ & $1.155680$ & $0.007260$ & $0.095199$ \\
    & {\tt lag1} & $0.000006$ & $0.354163$ & $0.354163$ & $0.005433$ & $0.007256$ \\
    & {\tt lag2} & $-0.000047$ & $-0.000592$ & $-0.000592$ & $-0.009552$ & $-0.007370$ \\
    & {\tt lag3} & $0.000047$ & $0.196567$ & $0.196567$ & $0.006796$ & $0.027961$ \\
    & {\tt cons} & $0.000168$ & $0.104003$ & $0.104003$ & $0.039414$ & $0.019547$ \\
    & {\tt dom\_inv} & $0.000069$ & $-0.025800$ & $-0.025800$ & $-0.005289$ & $0.000871$ \\
    & {\tt dow} & $-0.000063$ & $0.002384$ & $0.002384$ & $-0.003793$ & $-0.001496$ \\
    & {\tt emp} & $0.000094$ & $0.033710$ & $0.033710$ & $0.029074$ & $0.014655$ \\
    & {\tt exp} & $0.000003$ & $0.241046$ & $0.241046$ & $0.016397$ & $0.012615$ \\
    & {\tt gdp} & $-0.000007$ & $-0.012467$ & $-0.012467$ & $0.003903$ & $0.000862$ \\
    & {\tt hstarts} & $0.000007$ & $0.043571$ & $0.043571$ & $-0.000139$ & $0.003108$ \\
    & {\tt m1} & $0.000067$ & $0.053724$ & $0.053724$ & $0.008707$ & $0.012686$ \\
    & {\tt napmpri} & $0.000051$ & $0.154925$ & $0.154925$ & $0.115895$ & $0.031626$ \\
    & {\tt napmsdi} & $-0.000025$ & $-0.022680$ & $-0.022680$ & $-0.004349$ & $-0.000994$ \\
    & {\tt pmi} & $0.000138$ & $0.050042$ & $0.050042$ & $0.005330$ & $0.010289$ \\
    & {\tt spread} & $-0.000101$ & $-0.137584$ & $-0.137584$ & $-0.006458$ & $-0.052246$ \\
    & {\tt treas} & $0.000023$ & $0.072172$ & $0.072172$ & $0.002362$ & $0.012193$ \\
    & {\tt unemp} & $-0.000104$ & $0.000041$ & $0.000041$ & $-0.015649$ & $-0.027818$ \\
    \hline
\end{longtable}}

\subsection{ECB data}
Across all methods, the intercept, {\tt I\_1MO}, and {\tt M3} exhibit larger $\theta_j$ estimates, indicating pronounced time variation. For other predictors, NCP, D-CP/NCP and CLBS heavily shrink $\theta_j$ toward zero, implying nearly constant coefficients. An exception is {\tt LOANS}, where D-CP/NCP yields a larger $\theta_j$ estimate than NCP and CLBS, aligning more closely with CP and MCMC. Conversely, CP and MCMC produce larger $\theta_j$ estimates for {\tt SENT}, {\tt UNEMP}, and {\tt I\_1YR}. Variational methods yield posterior means of $\gamma_j$ near zero across most predictors, whereas MCMC estimates are larger for {\tt I\_1MO}, {\tt STOCK\_1}, {\tt M3}, and the December dummy {\tt D11}.

\begingroup
\small
\setlength{\LTleft}{\fill}
\setlength{\LTright}{\fill}

\begin{longtable}{llrrrrr}
\caption{Posterior mean estimates of $\bftheta$ and $\bfgamma$ for ECB data.}\label{tab:ecb_VA_estimates}\\
\hline
& Variable & CP & NCP & D-CP/NCP & CLBS & MCMC \\
\hline
\endfirsthead
\hline
& Variable & CP & NCP & D-CP/NCP & CLBS & MCMC \\
\hline
\endhead
\hline
\endfoot
\multirow{37}{*}{$\theta_j$} & Intercept & $0.021282$ & $0.025882$ & $0.025757$ & $0.027268$ & $0.040427$ \\
 & {\tt lag12} & $0.000014$ & $0.000000$ & $0.000000$ & $0.000000$ & $0.000003$ \\
 & {\tt lag11} & $0.000010$ & $0.000000$ & $0.000000$ & $0.000000$ & $0.000002$ \\
 & {\tt lag10} & $0.000011$ & $0.000000$ & $0.000000$ & $0.000000$ & $0.000003$ \\
 & {\tt lag9} & $0.000010$ & $0.000000$ & $0.000000$ & $0.000000$ & $0.000002$ \\
 & {\tt lag8} & $0.000032$ & $0.000000$ & $0.000000$ & $0.000000$ & $0.000004$ \\
 & {\tt lag7} & $0.000010$ & $0.000000$ & $0.000000$ & $0.000000$ & $0.000002$ \\
 & {\tt lag6} & $0.000009$ & $0.000000$ & $0.000000$ & $0.000000$ & $0.000002$ \\
 & {\tt lag5} & $0.000007$ & $0.000000$ & $0.000000$ & $0.000000$ & $0.000002$ \\
 & {\tt lag4} & $0.000016$ & $0.000000$ & $0.000000$ & $0.000000$ & $0.000006$ \\
 & {\tt lag3} & $0.000010$ & $0.000000$ & $0.000000$ & $0.000000$ & $0.000003$ \\
 & {\tt lag2} & $0.000014$ & $0.000000$ & $0.000000$ & $0.000000$ & $0.000005$ \\
 & {\tt lag1} & $0.000010$ & $0.000000$ & $0.000000$ & $0.000000$ & $0.000002$ \\
 & {\tt I\_1MO} & $0.009341$ & $0.009012$ & $0.011870$ & $0.011752$ & $0.003440$ \\
 & {\tt I\_1YR} & $0.000112$ & $0.000000$ & $0.000000$ & $0.000000$ & $0.000027$ \\
 & {\tt SENT} & $0.002369$ & $0.000000$ & $0.000000$ & $0.000000$ & $0.000798$ \\
 & {\tt STOCK\_1} & $0.000088$ & $0.000000$ & $0.000000$ & $0.000000$ & $0.000089$ \\
 & {\tt STOCK\_2} & $0.000087$ & $0.000000$ & $0.000000$ & $0.000000$ & $0.000290$ \\
 & {\tt EXRATE} & $0.000058$ & $0.000000$ & $0.000000$ & $0.000000$ & $0.000023$ \\
 & {\tt IP} & $0.000051$ & $0.000000$ & $0.000000$ & $0.000000$ & $0.000007$ \\
 & {\tt LOANS} & $0.000130$ & $0.000000$ & $0.000519$ & $0.000000$ & $0.000243$ \\
 & {\tt M3} & $0.000695$ & $0.001896$ & $0.001637$ & $0.001670$ & $0.000404$ \\
 & {\tt CAR} & $0.000016$ & $0.000000$ & $0.000000$ & $0.000000$ & $0.000006$ \\
 & {\tt OIL} & $0.000047$ & $0.000000$ & $0.000000$ & $0.000000$ & $0.000014$ \\
 & {\tt ORDER} & $0.000073$ & $0.000000$ & $0.000000$ & $0.000000$ & $0.000033$ \\
 & {\tt UNEMP} & $0.000829$ & $0.000000$ & $0.000000$ & $0.000000$ & $0.000188$ \\
 & {\tt D1} & $0.000057$ & $0.000000$ & $0.000000$ & $0.000000$ & $0.000018$ \\
 & {\tt D2} & $0.000049$ & $0.000000$ & $0.000000$ & $0.000000$ & $0.000029$ \\
 & {\tt D3} & $0.000078$ & $0.000000$ & $0.000000$ & $0.000000$ & $0.000052$ \\
 & {\tt D4} & $0.000037$ & $0.000000$ & $0.000000$ & $0.000000$ & $0.000007$ \\
 & {\tt D5} & $0.000034$ & $0.000000$ & $0.000000$ & $0.000000$ & $0.000007$ \\
 & {\tt D6} & $0.000046$ & $0.000000$ & $0.000000$ & $0.000000$ & $0.000010$ \\
 & {\tt D7} & $0.000060$ & $0.000000$ & $0.000000$ & $0.000000$ & $0.000015$ \\
 & {\tt D8} & $0.000035$ & $0.000000$ & $0.000000$ & $0.000000$ & $0.000008$ \\
 & {\tt D9} & $0.000043$ & $0.000000$ & $0.000000$ & $0.000000$ & $0.000009$ \\
 & {\tt D10} & $0.000034$ & $0.000000$ & $0.000000$ & $0.000000$ & $0.000007$ \\
 & {\tt D11} & $0.000189$ & $0.000000$ & $0.000000$ & $0.000000$ & $0.000033$ \\
\hline
\multirow{37}{*}{$\gamma_j$} & Intercept & $2.553500$ & $2.874060$ & $2.540080$ & $2.568590$ & $2.595070$ \\
 & {\tt lag12} & $-0.000012$ & $-0.001354$ & $-0.001236$ & $-0.006589$ & $-0.004225$ \\
 & {\tt lag11} & $-0.000019$ & $-0.000996$ & $-0.000974$ & $-0.004538$ & $-0.002744$ \\
 & {\tt lag10} & $-0.000018$ & $0.000216$ & $0.000167$ & $-0.000182$ & $0.000070$ \\
 & {\tt lag9} & $0.000015$ & $-0.000692$ & $-0.000615$ & $-0.000898$ & $-0.000450$ \\
 & {\tt lag8} & $-0.000018$ & $0.000726$ & $0.000760$ & $0.001319$ & $-0.000401$ \\
 & {\tt lag7} & $0.000019$ & $0.001494$ & $0.001368$ & $0.005879$ & $0.003735$ \\
 & {\tt lag6} & $0.000012$ & $-0.000782$ & $-0.000797$ & $-0.001400$ & $-0.000205$ \\
 & {\tt lag5} & $-0.000001$ & $-0.000692$ & $-0.000660$ & $-0.002249$ & $-0.001251$ \\
 & {\tt lag4} & $-0.000015$ & $0.000436$ & $0.000391$ & $0.000652$ & $0.000153$ \\
 & {\tt lag3} & $0.000000$ & $-0.001246$ & $-0.001146$ & $-0.004021$ & $-0.001925$ \\
 & {\tt lag2} & $-0.000007$ & $0.001001$ & $0.000997$ & $0.002985$ & $0.001123$ \\
 & {\tt lag1} & $-0.000035$ & $0.001187$ & $0.001062$ & $0.002746$ & $0.001382$ \\
 & {\tt I\_1MO} & $0.000004$ & $-0.000015$ & $0.000038$ & $0.000291$ & $0.023327$ \\
 & {\tt I\_1YR} & $-0.000003$ & $-0.000077$ & $-0.000036$ & $-0.000400$ & $-0.000998$ \\
 & {\tt SENT} & $-0.000003$ & $0.000144$ & $0.000113$ & $0.001093$ & $0.004388$ \\
 & {\tt STOCK\_1} & $0.000011$ & $0.000310$ & $0.000230$ & $0.002611$ & $0.021598$ \\
 & {\tt STOCK\_2} & $-0.000001$ & $0.000167$ & $0.000085$ & $0.000934$ & $0.001713$ \\
 $\gamma_j$ & {\tt EXRATE} & $-0.000015$ & $-0.000110$ & $-0.000091$ & $-0.001047$ & $-0.005700$ \\
 & {\tt IP} & $0.000010$ & $0.001811$ & $0.001574$ & $0.008288$ & $0.004651$ \\
 & {\tt LOANS} & $-0.000001$ & $0.000047$ & $-0.000023$ & $-0.000227$ & $-0.005219$ \\
 & {\tt M3} & $0.000006$ & $0.000069$ & $0.000031$ & $0.000236$ & $0.015218$ \\
 & {\tt CAR} & $-0.000008$ & $0.000802$ & $0.000739$ & $0.001677$ & $0.001345$ \\
 & {\tt OIL} & $-0.000001$ & $0.000089$ & $0.000042$ & $0.000466$ & $0.000965$ \\
 & {\tt ORDER} & $0.000004$ & $0.000190$ & $0.000137$ & $0.001455$ & $0.012343$ \\
 & {\tt UNEMP} & $0.000001$ & $-0.000060$ & $0.000013$ & $0.000149$ & $-0.000183$ \\
 & {\tt D1} & $-0.000016$ & $-0.000559$ & $-0.000544$ & $-0.003591$ & $-0.006694$ \\
 & {\tt D2} & $-0.000007$ & $0.000270$ & $0.000252$ & $0.000587$ & $0.001025$ \\
 & {\tt D3} & $0.000021$ & $0.000240$ & $0.000242$ & $0.002310$ & $0.003989$ \\
 & {\tt D4} & $0.000003$ & $0.000079$ & $0.000059$ & $0.000400$ & $0.001144$ \\
 & {\tt D5} & $0.000005$ & $0.000238$ & $0.000235$ & $0.000956$ & $0.002499$ \\
 & {\tt D6} & $-0.000014$ & $-0.000309$ & $-0.000301$ & $-0.001829$ & $-0.004321$ \\
 & {\tt D7} & $0.000004$ & $-0.000263$ & $-0.000274$ & $-0.000710$ & $-0.000226$ \\
 & {\tt D8} & $-0.000008$ & $0.000090$ & $0.000086$ & $-0.000415$ & $-0.001404$ \\
 & {\tt D9} & $0.000013$ & $0.000177$ & $0.000184$ & $0.001454$ & $0.003152$ \\
 & {\tt D10} & $0.000000$ & $0.000064$ & $0.000065$ & $0.000136$ & $-0.000153$ \\
 & {\tt D11} & $0.000018$ & $0.000423$ & $0.000406$ & $0.003242$ & $0.017632$ \\
\end{longtable}

\end{document}